\documentclass[11pt]{article}
\usepackage[margin=1in]{geometry}
\usepackage{amsmath, amssymb, amsthm, mathtools}
\usepackage{bbm}
\usepackage{hyperref}
\usepackage{cleveref}
\usepackage{braket}
\usepackage{tikz-cd}
\usetikzlibrary{arrows.meta, positioning}
\usepackage[numbers,square]{natbib}
\usepackage{pifont}
\usepackage{enumitem}
\theoremstyle{plain}

\newtheorem{theorem}{Theorem}[section]
\newtheorem{proposition}[theorem]{Proposition}
\newtheorem{lemma}[theorem]{Lemma}
\newtheorem{corollary}[theorem]{Corollary}

\newtheorem{claim}[theorem]{Claim}

\usepackage{thmtools}
\usepackage{refcount}
\usepackage[normalem]{ulem}
\newcounter{formalstmt}
\newenvironment{formalas}[1]{%
  \stepcounter{formalstmt}%
  \renewcommand{\themaintheorem}{\getrefnumber{#1}}%
  \renewcommand{\themainproposition}{\getrefnumber{#1}}%
  \renewcommand{\themaincorollary}{\getrefnumber{#1}}%
  }%
  {\addtocounter{maintheorem}{-1}}

\declaretheoremstyle[
    headfont=\bfseries,
    notefont=\bfseries, 
    bodyfont=\normalfont,
    postheadspace=0.5em,
    spaceabove=6pt,
    spacebelow=18pt,
    qed=$\diamond$ 
]{defstyle}

\declaretheorem[
    style=defstyle,
    name=Definition,
    numberlike=theorem
]{definition}

\declaretheoremstyle[
    headfont=\bfseries,
    notefont=\bfseries, 
    bodyfont=\normalfont,
    postheadspace=0.5em,
    spaceabove=6pt,
    spacebelow=18pt,
    qed=$\diamond$ 
]{example}

\declaretheorem[
    style=example,
    name=Example,
    numberlike=theorem
]{example}

\newtheoremstyle{maintheoremstyle}
  {\topsep}
  {18pt}
  {\itshape}
  {}
  {\bfseries}
  {.}
  {0.5em}
  {}

\theoremstyle{maintheoremstyle}
\newtheorem{maintheorem}{Theorem}
\newtheorem{mainproposition}[maintheorem]{Proposition}
\newtheorem{maincorollary}[maintheorem]{Corollary}

\newtheorem{informaltheorem}[maintheorem]{Theorem (informal)}
\newtheorem{informalproposition}[maintheorem]{Proposition (informal)}
\newtheorem{informalcorollary}[maintheorem]{Corollary (informal)}

\usepackage{xcolor}

\newcommand{\physrate}{\eta}

\newcommand{\pcth}{p_{\mathrm c}}

\newcommand{\effrate}{p_{\mathrm{eff}}}

\newcommand{\memerr}{\epsilon_{\mathrm{mem}}}
\usepackage[affil-it]{authblk}

\theoremstyle{remark}
\newtheorem{remark}[theorem]{Remark}

\begin{document}

\title{Autonomous Quantum Memories from Classical Stochastic Processes}

\author[1]{Nadav Drechsler}
\author[1]{Yuval Oreg}
\author[2]{Dorit Aharonov}
\affil[1]{Department of Condensed Matter Physics, Weizmann Institute of Science, Rehovot, Israel}
\affil[2]{School of Computer Science and Engineering, The Hebrew University of Jerusalem, Jerusalem, Israel}
\date{\today}

\maketitle

\begin{abstract}

The study of quantum memories out of thermal equilibrium is motivated both by the need for memories that withstand realistic noise in fault-tolerant quantum computers and by the fundamental question of when quantum memory can emerge in autonomous many-body systems. We establish a general equivalence between CSS-based quantum memories subject to Pauli noise and classical memories via two associated probabilistic processes. In the asynchronous setting, where the quantum dynamics is generated by a Lindbladian, we identify a broad class of syndrome-based Lindbladians for which the quantum memory problem is equivalent to memory in continuous-time classical Markov dynamics. In the synchronous setting, we introduce syndrome-based probabilistic cellular automata (SPCAs) and establish a one-to-one correspondence between their classical memories and CSS-based quantum memories under local quantum channels. This framework provides a systematic route for deriving and constructing quantum memories from classical probabilistic processes, and recovers known results, including the fact that the 4D toric code is a quantum memory in the presence of thermal noise. 

As an application, we give the first rigorous proof that the 4D toric code with Toom’s rule is a quantum memory in the presence of local Pauli noise, provided the error is below a constant threshold. We prove this both for the corresponding synchronous, probabilistic cellular automata (PCA) type quantum dynamics, as well as for a Floquet Lindbladian. We further show that more generally, classical memory in an SPCA can be lifted to a quantum memory under Floquet Lindbladian dynamics; whether this can be done with a time-independent Lindbladian remains open. Our results provide a general framework for relating classical and quantum memory, shed light on the conditions for fully autonomous quantum memories, and may guide the construction of provable autonomous quantum memories in low dimensions.

\end{abstract}

\section{Introduction}\label{sec:intro}

\subsection{Motivation and context: thermal self-correcting quantum memories}

The problem of noise leading to the degradation of quantum information is almost as old as the field of 
quantum computation itself. 
The standard solution of quantum fault tolerance \cite{shor_fault-tolerant_1996,aharonov_fault-tolerant_1997,kitaev_quantum_1997,knill_resilient_1998,aliferis_quantum_2006} makes use of active quantum error correction which suppresses  environmental noise and preserves encoded quantum information. Many fault tolerant constructions rely on measurements of multi-qubit parity operators in the $X$ and $Z$ basis, combined with the processing of syndrome data via classical decoding algorithms, and the execution of reactive Pauli corrections~\cite{terhal_quantum_2015}.
These constructions involve an external controller to perform the classical calculations, and also to control the choice of the operations. Crucially, these choices strongly depend on the exact time step as well as the exact location in space.  

Motivated both by practical realizations as well as fundamental physics questions, a major question arose - can the retainment of quantum information for long times, be achieved in an autonomous way: without an external controller, without involving interaction with a classical computer, and in a translationally invariant (TI) and time independent way. The internal dynamics would suppress noise passively, with no measurement, decoding, or reactive correction at all~\cite{brown_quantum_2016,bombin_self-correcting_2012}. 
Such systems are far closer to systems that are the subject of study in condensed matter physics; the physics of quantum memory should thus be understood using such models. The same framework could be useful for the study of the stability of these models, a key property of noise-robust phases of matter~\cite{coser_classification_2019,rakovszky_defining_2024}. Further, such translationally invariant and time independent more physically motivated models are likely to be more amenable to mathematical and physics analysis tools, unlike the more controlled and engineered  quantum fault tolerant constructions. 
We call this direction the study of {\it autonomous quantum information processing devices.} Here we are interested in autonomous quantum {\it memories}, but one might consider universal quantum computation performed by such systems. 

Research in this area has focused over the years mainly on systems that exhibit quantum memory in thermal equilibrium, under the name {\it self correcting quantum memory}~\cite{terhal_quantum_2015}; to avoid confusion, we will often refer here to such systems as {\it thermal self correcting quantum memories}. A thermal self-correcting quantum memory is defined to be a many body quantum system, whose quantum particles are often (though not always) set on a low dimensional lattice. At $0$ temperature, the dynamics is governed by a Hamiltonian, so that information can be encoded in its degenerate ground space, which serves as the {\it code space}. At non-zero temperature the system interacts with a {\it thermal bath} that induces thermal noise; no other noise is assumed.  The idea of thermal self correction is that
at sufficiently low temperatures, the bath drives the system toward its Gibbs state, which is concentrated near the ground space, through local relaxation processes satisfying quantum detailed balance; in the weak-coupling limit this is the Davies generator~\cite{davies_markovian_1974}. Protection then rests on energy barrier between code words; when those barriers are high enough, this relaxation process makes it overwhelmingly unlikely for the original codeword
to drift to a different one. In such systems
the thermal noise processes serve a dual role - their local transitions, obeying detailed balance, constitute both the noise process causing errors, as well as the source and way for the system to correct the effect of these errors.  
The four-dimensional toric code provides a celebrated example: Alicki, N. Horodecki, P. Horodecki and R. Horodecki famously proved that the model described by a local Lindbladian comprising of the Hamiltonian terms corresponding to the 4D toric code weakly coupled to a thermal bath exhibits an exponential quantum memory lifetime 
\cite{alicki_thermal_2010}. 

Unfortunately, quantum thermal self-correction turns out to be  constrained by dimensional limitations: Bravyi and Terhal \cite{bravyi_no-go_2009} showed that macroscopic energy barriers do not exist in two-dimensional geometrically local stabilizer codes. Establishing a construction in three dimensions has remained a long-standing challenge for a long time. Earlier three-dimensional models, such as Haah's cubic code~\cite{haah_local_2011}, have only logarithmic barriers~\cite{bravyi_energy_2011}, leading to memory times that grow with system size only up to a temperature-dependent cutoff size~\cite{bravyi_quantum_2013}; Michnicki's welded code reaches a power-law barrier without translation invariance~\cite{michnicki_3-d_2014}. 
.
In a recent breakthrough result, Balasubramanian, Davydova, and Lin \cite{balasubramanian_passive_2026-1} achieved a geometrically local 3D Hamiltonian with exponential-time thermal memory, tremendous progress along this long standing open question. However their constructed Hamiltonian is inherently non-translationally invariant and therefore does not lie in the autonomous setting. It remains open whether autonomous 3D thermal quantum memory exists. 

\subsection{Limitations on thermal equilibrium quantum memories}
From the condensed matter physics perspective, thermal equilibrium quantum memories constitute a very appealing model, studied within the established area of quantum systems in thermal equilibrium. However,
restricting the dynamics of the systems to correspond to a Hamiltonian coupled to a thermal bath, implies that they suffer from important limitations
as autonomous quantum memories.

First, as mentioned above, the geometrical dimensionality of thermal equilibrium quantum memories must be at least $3$, by the no-go result of Bravyi and Terhal \cite{bravyi_no-go_2009}. The fact that thermal equilibrium quantum self correction is impossible in lower dimension makes their practical implementation  challenging.  We note that many current quantum computation implementations are 
restricted to a $2D$ plane ~\cite{kim_evidence_2023,google_quantum_ai_and_collaborators_quantum_2025}. 
This geometrical bottleneck for the thermal equilibrium setting is not fundamentally quantum: it already appears, in a sharper form, at the classical level.
One-dimensional classical systems governed by local, reversible
stochastic processes (respecting detailed balance) cannot retain information at positive noise rates, due
to the unconstrained proliferation of domain
walls~\cite{taati_reversible_2021}. Breaking detailed balance circumvents this
limitation: Gács demonstrated that a one-dimensional probabilistic
cellular automaton with irreversible update rules can reliably store
classical information against general local noise below a positive threshold~\cite{gacs_reliable_2001}.
In two dimensions, Toom's rule achieves an even stronger result by
the same mechanism: a local, translationally invariant, irreversible
majority-vote update that erodes islands of error and preserves a
stored bit for an exponential time below a finite noise
threshold~\cite{toom_stable_1980}. 
In the quantum setting, while Bravyi and Terhal's no go result forbids thermal equilibrium quantum memories in $2D$ or $1D$ systems, 
the fault-tolerance theorems of Aharonov and
Ben-Or~\cite{aharonov_fault-tolerant_2008} and Gottesman~\cite{gottesman_fault-tolerant_2000} imply that quantum memory at constant noise rate is achievable even in one spatial
dimension, as long as one is allowed to use a controller,  and repeatedly reset ancilla qubits. 
We note that the fault tolerant constructions completely break 
spatial and time translation invariance and are thus not examples of quantum autonomous systems. However 
given the rich set of tools in the quantum complexity literature for achieving translational invariance~\cite{gacs_reliable_2001,gottesman_quantum_2013,cubitt_undecidability_2015,bausch_complexity_2017}
one is naturally led to suspect \cite{aharonov_physics_2023} that it might be possible to carry over the low dimensional fault tolerant error correction results to the autonomous memory setting. 
Indeed, a recent breakthrough result \cite{dunnweber_quantum_2026} does exactly this, leading to a beautiful construction of autonomous quantum memory in $2D$. This proves that the more general setting of autonomous quantum memories evades the restrictions of the thermal equilibrium
settings.

From a practical point of view, there is an even more pressing problem, that casts strong doubts regarding the practical relevance of the thermal noise assumption. In all known thermal self-correcting constructions, the Hamiltonian that governs the system is $k$-local with $k=4$ or more. It is thus not expected that in physical implementations, such a Hamiltonian could be the {\it physical} Hamiltonian governing the system. 
Instead, it is expected that in physics implementations, the self correcting Hamiltonian will be  an {\it abstraction}, namely the {\it effective} Hamiltonian resulting from an underlying hardware physics governed by a different Hamiltonian or Lindbladian. If this is indeed the case, there is no longer a good reason to believe that the dominant noise in the system is thermal with respect to the {\it effective} Hamiltonian. 
 (in fact, in engineered settings, many noise processes often come into play, and the thermal assumption will probably break down even with respect to the physical hardware Hamiltonians). 
Importantly, once dissipative non-thermal noise is present (e.g., depolarizing noise) quantum systems
governed by a unitary ideal dynamics, namely by a Hamiltonian, cannot exhibit memory for more than logarithmic time~\cite{aharonov_limitations_1996,pastawski_how_2009}. Handling realistic noise models thus requires lifting the requirement that the ideal dynamics be described by a Hamiltonian. 

Last but not least, from a condensed matter physics point of view, focusing on thermal equilibrium systems 
seems to unnecessarily restrict the rich set of physical systems that might be capable of exhibiting autonomous quantum memory and possibly noise robust phases of matter. Removing the entropy accumulating in the system due to noise, by reset gates or other engineered means, is long known to be necessary for quantum fault tolerance protecting against general local noise \cite{aharonov_limitations_1996, pastawski_how_2009}.
The thermal equilibrium requirement, while a significant simplification, does not allow incorporating the physics analogue of this crucial ingredient of engineered entropy removal.

The above are all strong motivations to turn the focus of study of quantum self correcting memories beyond what had been the central focus of research in the area for many years, namely thermal equilibrium quantum self correction, to Lindbladians based quantum systems (that in general do not need to satisfy detailed balance with respect to some Hamiltonian), or, in the synchronous case, to local quantum channels that are not necessarily unitary (those can be viewed as non-unitary quantum cellular automata). 
We turn to such autonomous quantum systems next. 

\subsection{Autonomous quantum memories}

\subsubsection{Fully Autonomous quantum memories - time-independent Lindbladian}
The possibility of scalable autonomous quantum memories was introduced by Pastawski, Clemente, and Cirac~\cite{pastawski_quantum_2011}, who proposed to use engineered local dissipation to autonomously protect a topological quantum memory. They constructed a four-dimensional toric-code model with local dissipative dynamics inspired by Toom's rule, and numerically found a regime in which the encoded information has an exponentially increasing lifetime with system size. This work established engineered dissipation as a route toward autonomous self-correction, complementary to passive thermal self-correction. 
Pastawski, Clemente and Cirac \cite{pastawski_quantum_2011} 
modeled an autonomous quantum information process
using a standard Master equation: 
\begin{equation}
\frac{d\rho}{dt}=\mathcal{L}_{diss}(\rho) + \mathcal{L}_{noise}(\rho)
\end{equation}
with both Lindbladians $\mathcal{L}_{diss}$ and $ \mathcal{L}_{noise}$
consisting of a sum of local terms, independent of location and time;  $\mathcal{L}_{diss}$ is supposed to capture the engineered dissipative part of the Lindbladian, on which we have control, while $ \mathcal{L}_{noise}$ captures the uncontrolled part, namely the noise. 
There is no external controller, and the system is translationally invariant both in time and space, justifying the term {\it autonomous}. 

We note that thermal self-correction is a special case, with coupling to a cold bath locally favoring kicks that lower the energy, respecting detailed balance (We show this to be a special case below (Remark~\ref{rem:thermal})). 
Indeed, Alicki's proof that the 4D Toric code is a thermal equilibrium quantum memory \cite{alicki_thermal_2010} can be cast within this model. 

The autonomous correction mechanism can be pictured as a collection of abstract ratchets: a fixed, local rule recognizes which configurations are locally correct and lets a noise-driven kick take effect only when it pushes the state toward one of them, with no external controller measuring the state and deciding on a correction. 
We note that the dynamics generated by a 
Lindbladian can be viewed as
an {\it asynchronous}, stochastic process; 
this is to be contrasted with the 
synchronous setting to be discussed soon.

 We believe that the study of quantum autonomous information processing systems in the more general, Lindbladian based setting,  
is a more faithful model within which the physics of quantum self correction, and more generally, quantum fault tolerance, in autonomous physical systems should be understood~\cite{aharonov_physics_2023}.
This direction is tightly related to the emerging field of active quantum matter, to be discussed in the related work Subsection below~\ref{subsec:related-work}.  

The intuition behind the proposal for autonomous quantum memory by Pastawski et al. \cite{pastawski_quantum_2011} builds on 
two known results: 
first, that the 4D Toric code is already known by Alicki et al. to support
exponential memory lifetimes under thermal, detailed-balance-respecting
dynamics~\cite{alicki_thermal_2010}. 
Second, that a discrete-time probabilistic
cellular automaton based on Toom’s rule suppresses local faults in the 4D toric code below a finite
noise threshold, as shown by Ahn~\cite[\S 7]{ahn_extending_2004}. Pastawski et al. introduced engineered dissipative terms to the Lindbladian, that implement Ahn's $4D$ Toom update rule.
Numerical evidence provided by Pastawski, Clemente, and Cirac~\cite{pastawski_quantum_2011}
supported memory retention under the resulting Lindblad
dynamics, however a rigorous proof that the system 
indeed exhibits quantum memory remained open.

\subsection{PCA-type autonomous quantum memories}
While the Lindbladian framework is physically the more natural one (Subsection~\ref{sec:WhyLindbladians}), it is also natural to study autonomous quantum memories in the \emph{synchronous} setting, whose analysis is often much simpler, both in the quantum and in the classical case. We refer to such models as \emph{PCA-type}, where PCA stands for probabilistic cellular automata: they are translationally invariant, apply the same local update rule at every round, and are geometrically local, but assume a global clock coordinating the updates. The word ``type'' reflects the many variants of this setting in the literature.

One way to define a synchronous non-unitary model is roughly as follows (see Subsection~\ref{subsec:sync-setup} for a rigorous definition).  
At every time step a local, not-necessarily-unitary, translationally
invariant local quantum channel is applied. In between every two consecutive time steps, a local noise channel is applied. 
This model parallels the notion of probabilistic noisy cellular automata in the classical world; related variants of non-unitary quantum cellular automata appear in~\cite{farrelly_review_2020,wintermantel_unitary_2020,nigmatullin_directed_2021,guedes_quantum_2024}.

Lake~\cite{lake_local_2025} gave a PCA-type autonomous memory based on the 2D surface code with a provable threshold,
but with a caveat: it relies on $O(\mathrm{polylog}\,L)$ noiseless
classical bits per qubit (he also discussed a possible strengthening using a time independent Lindbladian). In a recent work, 
Dünnweber, Styliaris and Trivedi \cite{dunnweber_quantum_2026}
then provided a construction of PCA-type quantum autonomous 
memory in $2D$, making major progress on the question mentioned above, 
of reducing dimensionality for autonomous memory in quantum devices. The construction of ~\cite{dunnweber_quantum_2026} builds on the fault-tolerant construction of Aharonov and Ben-Or~\cite{aharonov_fault-tolerant_2008} and turns its hierarchical structure and system-size-dependent schedule into translationally invariant dynamics, using the self-simulation techniques of G\'acs~\cite{gacs_self-correcting_1989,gacs_reliable_2001}. 
This result confirms the intuition~\cite{aharonov_physics_2023} suggested by early works providing 
quantum fault tolerance in low dimensions
\cite{aharonov_fault-tolerant_2008, gottesman_fault-tolerant_2000}, and shows that indeed the more general setting of autonomous quantum memories can evade restrictions of the thermal equilibrium settings.

We mention that in both the classical and the quantum cases, deducing asynchronous results from the synchronous setting is notoriously difficult. 
In particular, a synchronous PCA-type memory should not be viewed as automatically implying a fully autonomous memory; the passage to asynchronous dynamics generally requires new synchronization or stability arguments, adapted to the specific system.

\subsubsection{Lindbladians or channels}\label{sec:WhyLindbladians}

In the quantum case, the synchronous and asynchronous formalisms differ primarily in their temporal organization. In the synchronous setting, a fixed quantum channel is applied in discrete, globally coordinated rounds. By contrast, a time-independent Lindbladian generates continuous-time evolution without externally imposed global update rounds and, in a quantum-trajectory representation, can be viewed as describing local jumps occurring at stochastic times.

A synchronous PCA-type model thus relies on a global clock, whereas a time-independent Lindbladian describes a fully autonomous system (that of course may be engineered in advance to possess the required properties). 
From a physics point of view, working with Lindbladians
is desirable, due to 
physical faithfulness: an autonomous system is not expected to come
equipped with global synchronization. 
Furthermore, even if a global clock is assumed, it is unclear how to implement a PCA-like dynamics in the lab. A Lindbladian provides a physics embedding of the discrete update rule in continuous-time local evolution. 
Hence, the holy grail is to design Lindbladian based autonomous quantum
memories.

To the best of our knowledge, 
{\it outside of the thermal-equilibrium setting}, there is no known 
rigorous proof of a fully autonomous quantum memory -- i.e., one that is described by a time independent, translationally invariant, local Lindbladian -- in any geometrical dimension.  
Results along these lines are known only for Lindbladian constructions that relax one or more of the requirements imposed here. Those include the results of ~\cite{lake_local_2025,lake_fast_2025} 
that assumed 
a noiseless classical processing layer; the results of  ~\cite{chirame_stabilizing_2024} that assumed heralded noise, and the results of ~\cite{fujii_measurement-free_2014} in which the local terms depend on parameters growing with the system size. 
Dünnweber, Styliaris, and Trivedi~\cite{dunnweber_quantum_2026} recently outlined a continuous-time Lindbladian realization of their two-dimensional PCA-type autonomous quantum memory construction, but leave a formal extension of their memory result to the continuous-time setting for future work.

\subsection{Informal statement of main results}
 \subsubsection{Quantum-classical memory equivalences}
Our first main contribution is to provide a general framework within which autonomous quantum memories can be studied and constructed beyond the thermal equilibrium self correction setting.
To this end we first provide reductions, both in the PCA-type and in the fully autonomous type, from the analysis of memories in the quantum domain, to that of Markov processes in the classical domain. This quantum to classical reduction enables solving the quantum memory problem by turning to the easier associated classical memory problem.  

We begin with the setting most directly motivated by physics, namely the fully autonomous quantum-memory setting.
We provide a general equivalence between fully autonomous quantum memory (described by a time independent Lindbladian), 
and classical memory of two induced continuous-time stochastic processes.

We consider here a logical qubit encoded by a CSS code \cite{calderbank_good_1996,steane_multiple-particle_1996}. 
The Lindbladian describes an autonomous 
error correction process: 
the engineered correction dynamics are implemented by syndrome-conditioned dissipative jumps that occur with no external controller.
Our equivalence holds for a wide class of Lindbladians, 
that we call 
\emph{syndrome-based Lindbladians} (Definition~\ref{def:syndrome-based-lindbladian}): in each of the two CSS sectors ($X,Z$), the
total rate at which a qubit is flipped depends only on that sector's syndrome. 
We identify a strong algebraic structure in the evolution of observables propagating under such  Lindbladians, that 
guarantees that each sector evolves as
a classical continuous-time Markov process on its error
configurations (bit strings of length $n$);  we call these continuous time Markov processes the \emph{induced sector processes}. We have: 

\begin{informaltheorem}[Fully autonomous quantum--classical memory equivalence]\label{thm:main-async-informal}
Consider a CSS code evolving under a syndrome-based Lindbladian.
The evolution induces two classical continuous-time Markov processes,
one describing the $X$-error sector and one describing the
$Z$-error sector. Fix a time $t$. Suppose that at every readout time up to $t$, the probability that the $X$-sector process has
changed its decoded logical bit is at most $q_X$, and the
corresponding probability in the $Z$ sector is at most $q_Z$.
Then the encoded qubit is a quantum $(q_X+q_Z)$-memory up to time~$t$.

Conversely, if the encoded qubit is a quantum $\epsilon$-memory up
to time $t$, then in each of the two induced classical sector
processes the probability of changing the decoded logical bit is at
most $\epsilon$, at every readout time up to $t$.
\end{informaltheorem}

The formal statement is Theorem~\ref{thm:main-async-formal} in Section \ref{sec:results}.  Thus,
within the class of syndrome-based Lindbladian dynamics, quantum
memory of the designated logical qubit is equivalent to memory in the
two associated continuous time classical processes.
This equivalence reduces the problem of constructing 
and analyzing autonomous quantum memories in CSS based systems, to proving memory in the induced classical continuous-time processes, which is potentially significantly easier.  

This result generalizes and formalizes in the Lindbladian setting a long assumed common intuition in the quantum computation community, 
regarding the "decoupling" to $X$ and $Z$ sectors in CSS based physical evolutions, 
inspired for example by Alicki et al.\ thermal memory result \cite{alicki_thermal_2010} and by the well known fact that in the ideal noiseless cases, CSS error correction can be done separately for the $X$ and $Z$ errors \cite{calderbank_good_1996,sarma_gauge_2013}. 
A general result along these lines has not, to our knowledge, previously been stated and proved with its precise assumptions and scope of applicability clearly identified. The central notion of syndrome-based Lindbladians that we introduce, and the algebraic structure that we prove for the dynamics with respect to such Lindbladians, enable us to formulate and prove such a general result. The required definitions and the proof are provided in Sections~\ref{sec:setup},\ref{sec:async}. 
As an example of the applicability of this result we show in Remark~\ref{rem:alicki}, that Alicki et al.'s result \cite{alicki_thermal_2010} of quantum thermal memory for the 4D toric code
follows as a direct application

However, for the purposes of constructing new quantum memory constructions, a reduction to classical asynchronous processes might be problematic, since 
it is often the case that 
the analysis of asynchronous classical systems could be 
extremely challenging. We thus 
provide a synchronous, PCA-type variant of the equivalence result. To this end we consider PCA-type autonomous quantum memories. Such a system is defined by a local translationally invariant quantum channel iterated in synchronous rounds with noise between rounds.  We impose additional conditions on the channel, namely, that it is syndrome-based (in full analogy with our syndrome based Lindbladians). Our synchronous equivalence connects memory in such PCA-type quantum systems to memories in the associated 
classical probabilistic cellular automata (PCA), that are also syndrome based. Such classical PCAs, which we call {\it syndrome based probabilistic cellular automata} and denote by 
SPCA, are a central definition in our work (See Def.~\ref{def:syndrome-pca-classical} for a formal definition); constructing memories in such systems is a "recipe" to constructing PCA-type quantum memory. 
The equivalence proof relies on similar syndrome-based 
algebraic structure as that of Theorem \ref{thm:main-async-informal}.

\begin{informaltheorem}[Equivalence between PCA-type autonomous quantum memory and  classical SPCAs
]\label{thm:main-spca-sync-informal}
Consider a CSS code whose two classical component codes admit local SPCAs. First,  there exists a local translationally invariant syndrome-based quantum channel, which, under local Pauli noise, induces, in each CSS
sector, the discrete-time classical Markov process corresponding to the respective SPCA. 

 If the \(X\)- and
\(Z\)-sector SPCAs $q_X-$ and $q_Z$-remember their logical bit up to round $m$, respectively, then the quantum system is a PCA-type $(q_X+q_Z)$-quantum memory up to round $m$.
Conversely, if the resulting local quantum  channel defines a PCA-type autonomous quantum $q$-memory up to round $m$, then each of the two induced classical sector processes $q$-remembers their logical bit up to round $m$.
\end{informaltheorem}

This result is an informal summary of two results, given formally in  
Theorem~\ref{thm:main-sync} and Corollary~\ref{cor:sync-spca-memory}. 
Importantly, the synchronous equivalence theorem (Theorem \ref{thm:main-spca-sync-informal}) enables 
reducing the problem of PCA-type autonomous quantum memory, to the classical problem of proving classical memory in SPCAs. 
In the last section of the paper (Section~\ref{sec:classical-problem}) we flesh out the exact requirements on classical PCAs in order to meet the requirements on SPCAs and such that once memory is established in such a PCA, that would immediately imply a quantum memory for a given CSS code. Hopefully, this could lead to new and interesting PCA-type quantum memory constructions.  

The following Figure~\ref{fig:thms} provides a scheme of the equivalence results described so far; those correspond to the vertical arrows in the figure. 
Drawing connections between the left and right sides of the figure, namely between the synchronous and 
asynchronous settings, is both highly desirable as well as notoriously difficult. Here we provide progress on this difficult problem, in the form of the diagonal arrow in the figure, discussed in Subsection \ref{subsec:discrete-shortcut} below.

\begin{figure}[h]
\centering
\begin{tikzpicture}[>=Stealth, font=\small,
  obj/.style={draw, thick, rounded corners, fill=white, minimum width=4.4cm, minimum height=1.3cm, align=center},
  elbl/.style={font=\scriptsize, align=center, fill=white, inner sep=1.5pt}
]
\node[obj] (UL) at (-4.3, 3.2) {Syndrome-based\\ Lindbladian};
\node[obj] (UR) at ( 4.3, 3.2) {Syndrome-based\\ local quantum channel};
\node[obj] (LL) at (-4.3, 0)   {Induced continuous-time\\ sector process};
\node[obj] (LR) at ( 4.3, 0)   {Induced discrete-time sector process \\ {\scriptsize (SPCA by Corollary~\ref{cor:sync-spca-memory})}};

\draw[dotted, gray!70, thick] (-6.6,1.6) -- (6.6,1.6);
\node[gray!90, anchor=east, font=\small\bfseries] at (-6.7, 3.2) {Quantum};
\node[gray!90, anchor=east, font=\small\bfseries] at (-6.7, 0)   {Classical};
\node[anchor=south, font=\small\bfseries] at (-4.3, 4.1) {Asynchronous (continuous time)};
\node[anchor=south, font=\small\bfseries] at ( 4.3, 4.1) {Synchronous, PCA-type (discrete time)};
\draw[dotted, gray!70, thick] (0,-0.9) -- (0,4.0);

\draw[->, thick] (-4.6, 2.55) -- node[elbl, left=2pt] {induces\\ (Prop.~\NoHyper\ref{prop:reduction-async-formal}\endNoHyper)} (-4.6, 0.65);
\draw[<->, thick] (-4.0, 2.55) -- node[elbl, right=2pt] {memory iff\\ (Thm.~\NoHyper\ref{thm:main-async-formal}\endNoHyper)} (-4.0, 0.65);

\draw[->, thick] (4.0, 2.55) -- node[elbl, left=2pt] {induces\\ (Prop.~\NoHyper\ref{prop:reduction-sync-formal}\endNoHyper)} (4.0, 0.65);
\draw[<->, thick] (4.6, 2.55) -- node[elbl, right=2pt] {memory iff\\ (Main Thm.~\NoHyper\ref{thm:main-sync}\endNoHyper)} (4.6, 0.65);

\draw[->, thick, dashed] (LR.north west) -- node[sloped, above, elbl] {Floquet construction\\ (Main Thm.~\NoHyper\ref{thm:periodic-formal}\endNoHyper)} (UL.south east);

\end{tikzpicture}
\caption{The reduction diagram. Rows: quantum dynamics (top) and the
classical process it induces (bottom). Columns: continuous-time
(asynchronous) dynamics (left) and synchronous discrete-time dynamics
(right). In each column, the vertical arrows are the reduction: the
quantum dynamics induces a classical sector process
(Propositions~\ref{prop:reduction-async-formal}
and~\ref{prop:reduction-sync-formal}), and memory of one is equivalent to
memory of the other, up to a constant factor (Theorems~\ref{thm:main-async-formal} and~\ref{thm:main-sync}). The
diagonal arrow is the periodic construction of Section~\ref{sec:periodic}:
from an SPCA it builds a Lindbladian whose sampled sector
process is an SPCA run (Theorem~\ref{thm:periodic-formal}). }
\label{fig:thms}
\end{figure}
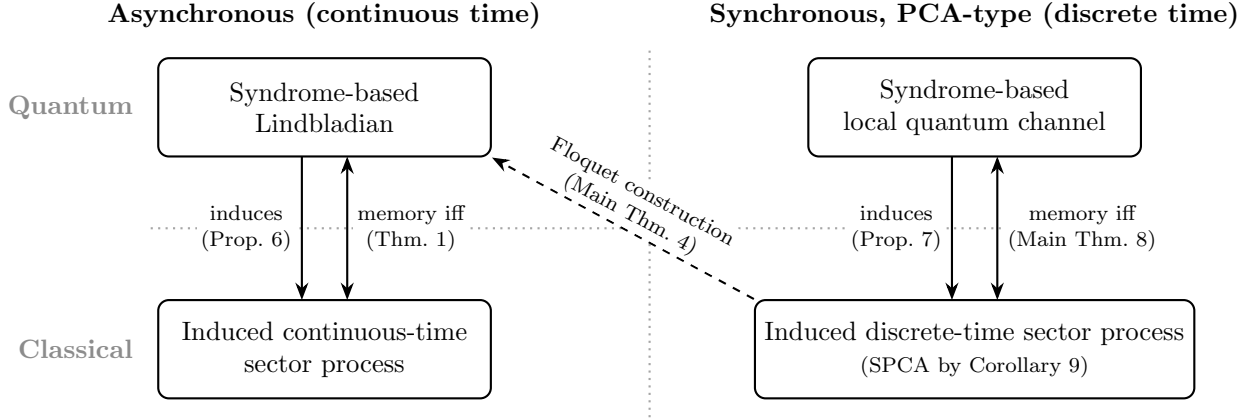

We reserve the term \emph{fully autonomous quantum memory} for the strongest setting: protection against general local noise by a time-independent, local, translationally invariant Lindbladian. Floquet and PCA-type autonomous quantum memories relax this setting by allowing, respectively, a time-periodic drive or a global clock; Table~\ref{tab:taxonomy} summarizes the terminology. Throughout, \emph{synchronous} and \emph{asynchronous} refer to the time structure of the dynamics: a local channel applied in globally clocked rounds, versus continuous-time Lindbladian evolution in which local jumps occur at stochastic times. In this sense Floquet Lindbladians are asynchronous, although their periodic schedule acts as a global clock.

\begin{table}[h]
\centering
\begin{tabular}{lll}
\hline
Dynamics & Thermal noise & General local noise \\
\hline
Synchronous local channel & PCA-type self-correcting & PCA-type autonomous \\
Time-periodic local Lindbladian & Floquet self-correcting & Floquet autonomous \\
Time-independent local Lindbladian & Self-correcting & Fully autonomous \\
\hline
\end{tabular}
\caption{Taxonomy of autonomous quantum memories. Each entry names a class of quantum memories; in all of them the dynamics is local and translationally invariant, with parameters independent of the system size.}
\label{tab:taxonomy}
\end{table}

\subsubsection{Application: Autonomous quantum memory in 4D Toric code with Toom's rule}
We now derive a "low hanging fruit" application of the above quantum-classical equivalence synchronous result, Theorem \ref{thm:main-spca-sync-informal}. 
We provide a rigorous proof of a synchronous variant of the long claimed folklore statement (proposed in \cite{pastawski_quantum_2011}) that the 4D Toric code with 4D Toom's rule corrections, exhibits an autonomous quantum memory against Pauli noise.  

\begin{informaltheorem}
[PCA-type 4D Toom quantum memory]\label{thm:toom-sync-informal}
For the 4D toric code, take the two classical correction rules to be
the 4D Toom rules. There exists a nonzero local Pauli-noise threshold below
which the encoded qubit, evolving under the resulting local quantum channel and subjected to Pauli noise in between consecutive applications of the channel, is a quantum $2Cme^{-cL}$-memory through round \(m\), for constants
\(C,c>0\) independent of the linear system size \(L\). The channel is
translationally invariant whenever the noise model is.
\end{informaltheorem}

The formal statement is Theorem~\ref{thm:toom-sync-formal}. Its classical input is Ahn's
memory theorem for the noisy 4D Toom rule~\cite{ahn_extending_2004}; It is easy to verify that Ahn's classical probabilistic cellular automaton indeed satisfies the condition of 
an SPCA. The PCA-type quantum-classical memory equivalence, Theorem~\ref{thm:main-spca-sync-informal}, immediately converts this classical result into a quantum one. By this we confirm that a variant of the physical system  
proposed in \cite{pastawski_quantum_2011} indeed provides a PCA-type autonomous quantum memory protecting against Pauli noise.

\subsubsection{Quantum memories with Floquet Lindbladians from classical PCA memories}\label{subsec:discrete-shortcut}

So far, we have worked either fully in the asynchronous setting, connecting memory in quantum Lindbladian systems to memory in continuous classical processes (Theorem \ref{thm:main-async-informal}),  
or fully in the synchronous setting, drawing an equivalence between PCA-type quantum memory and classical memory in the corresponding classical PCAs (Theorem \ref{thm:main-spca-sync-informal}). 
Theorem \ref{thm:main-async-informal} and
Theorem~\ref{thm:main-spca-sync-informal} hold essentially verbatim for the PCA-type setting
as for Lindbladians, with the same proof structure throughout. However, the choice of setting is very important from the point of view of applications, as discussed in Subsection \ref{sec:WhyLindbladians}.

Specifically, as described in \ref{sec:WhyLindbladians}, from a physics point of view, the object we would like to construct is a fully time-\emph{independent},
translationally invariant Lindbladian, as in 
Theorem~\ref{thm:main-async-formal}. 
By this theorem, the construction of such a Lindbladian with quantum memory would
follow immediately from an \emph{asynchronous} PCA with a nonzero
threshold --- an update rule in which sites act at independent,
uncoordinated times. 
Unfortunately, such asynchronous continuous time classical stochastic processes, are often much harder 
to analyze and control. Instead, it is often the 
case, that classical memory is only known to hold in the 
discrete cellular automaton setting; such is the case in Ahn's proof of memory in the 4D Toom's rule cellular automaton \cite{ahn_extending_2004}, which holds in the synchronous setting. Proving memory in the asynchronous setting for this model is an open problem (see Section~\ref{subsec:outlook}). 
Alas, the guarantee of a classical memory only in the PCA setting forces us to apply the equivalence in the synchronous, 
rather than the asynchronous setting, and thus derive  autonomous quantum memory in the less realistic PCA-type quantum setting, rather than a Lindbladian based one. 

A natural question arises - can one enjoy both worlds? 
Is there a way to start from the easier synchronous setting in the classical world - namely with a guarantee on a discrete time classical PCA having a memory -- and deduce 
quantum Lindbladian-based memory?

Unfortunately, we are not able to provide such a satisfying theorem. 
We provide a non-trivial partial result: starting from a classical memory in the synchronous setting, we construct a Floquet autonomous quantum memory, namely, a system described by a  
Lindbladian that is periodic in time, with a period independent of the system
size. The generator itself is geometrically local and
translationally invariant, making it more realistic to implement than a PCA-type quantum channel, which is unclear how to implement physically.

\begin{informaltheorem}[Floquet autonomous quantum memory from SPCAs]\label{thm:periodic-informal}
Consider a CSS code whose two classical component codes admit local SPCAs, and
suppose that below a fault-rate threshold $\pcth>0$, every run of each SPCA
started from $\mathbf 0$ changes its decoded logical bit with probability at most
$q_X(m,L)$, respectively $q_Z(m,L)$, at every round up to $m$. Then one can
explicitly construct a geometrically local Floquet Lindbladian, acting on the
code qubits together with local classical registers, under which, for local Pauli
noise on the code qubits and the registers that is below a threshold inherited from
$\pcth$, the encoded qubit is a quantum $(q_X(m,L)+q_Z(m,L))$-memory at cycle
boundaries through cycle $m$. The Lindbladian is translationally invariant
whenever the CSS code and both SPCAs are.
\end{informaltheorem}

The formal statement of the theorem is Theorem~\ref{thm:periodic-formal} in Section~\ref{sec:results}. Its 
non-trivial proof, which is one of the two main 
technical contributions of the paper, alongside the proof of the equivalence theorems (See proof overview in Subsection \ref{subsec:periodic-overview}) is given  
in Section~\ref{sec:periodic}.
While possibly counterintuitive, there are examples of time-periodic Lindbladians that do not have an 
effective time-independent Lindbladian that leads to an equivalent dynamics~\cite{schnell_is_2020}, and thus, straightforward attempts to turn the time-periodic Lindbladian into a time independent one do not seem to work;  whether the above theorem can be strengthened to provide a black-box result that starts with a PCA classical memory construction and provides a Lindbladian-based quantum memory with a time-independent
Lindbladian, remains as an important open problem
(See Subsection \ref{subsec:outlook}).

The theorem can be immediately applied to derive 
a Floquet autonomous quantum memory for the 
4D Toric code with Toom rule corrections:

\begin{informalcorollary}
[Floquet autonomous 4D Toom quantum memory]\label{thm:toom-periodic-informal}
For the 4D toric code, take the two classical correction rules to be
the 4D Toom rules. There exists a nonzero local Pauli-noise threshold below
which the encoded qubit, evolving under the resulting local, translationally
invariant time-periodic Lindbladian, is a quantum $2Cme^{-cL}$-memory at cycle
boundaries through cycle \(m\), for constants \(C,c>0\) independent of the
linear system size \(L\).
\end{informalcorollary}

This Corollary strengthens Theorem \ref{thm:toom-sync-informal} from the PCA-type quantum memory setting 
to deriving a more physically realistic Floquet Lindbladian memory, for the 4D toric code with Toom based rules. The question of whether this can be achieved with a time independent Lindbladian, e.g., for the Lindbladian suggested in \cite{pastawski_quantum_2011}, 
is left open. Such a result would imply 
a resolution of the more general open question, of providing a fully autonomous quantum memory beyond the thermal equilibrium setting.

To summarize, we believe these results significantly clarify the landscape of autonomous quantum memories, provide a general framework within which to study and construct them, and beyond providing a resolution to the open problem of the 4D toric code memory against non-thermal noise in the PCA-type setting, make significant progress towards the construction and the clarification of the power of fully autonomic quantum memories, bringing the field significantly closer to physical
settings.

\subsection{Overview of the Main Technical Ideas in the Proofs}

We describe here the main ideas in the proofs of this work. We start with the ideas underlying the proofs of  
Theorem~\ref{thm:main-async-informal} and Theorem~\ref{thm:main-spca-sync-informal} 
(whose proofs are very essentially the same), 
and proceed to the proof of Theorem \ref{thm:periodic-informal}. 
Those proofs contain the 
ideas and concepts that constitute the main contributions of this work, from the technical point of view. 

\subsubsection{Proofs of quantum-classical memory equivalences: Theorems~\ref{thm:main-async-informal} and~\ref{thm:main-spca-sync-informal}}  
The hearts of these proofs are given, respectively, in  
Prop.~\ref{prop:reduction-async-informal} and Prop.~
\ref{prop:reduction-sync-informal}. The first proposition connects the Lindbladian in the quantum setting, to the classical continuous Markov processes in the classical setting: retaining the relevant quantum logical
correlation is equivalent to retaining the classical
decoded logical bit under the corresponding classical Markov process.  

\begin{informalproposition}[Asynchronous quantum-to-classical reduction]\label{prop:reduction-async-informal}
Consider a CSS code evolving under a syndrome-based Lindbladian.
The evolution induces two classical continuous-time Markov processes,
one describing the $X$-error sector and one describing the
$Z$-error sector.
For each sector, the autocorrelation of the corresponding decoded
quantum logical observable is exactly one minus twice the probability
that the decoded logical bit of the associated classical process has
changed between the initial and readout times.
\end{informalproposition}

The formal statement is given in Section~\ref{sec:results}
(Proposition~\ref{prop:reduction-async-formal}).
We state the synchronous counterpart of
Proposition~\ref{prop:reduction-async-formal}: 

\begin{informalproposition}[Quantum-to-classical reduction for syndrome-based channels]\label{prop:reduction-sync-informal}
The same exact correspondence holds in synchronous discrete time.
A syndrome-based local translationally invariant channel induces, in each CSS
sector, a discrete-time classical Markov process describing the
evolution of the corresponding error configuration.

After any number of rounds, the autocorrelation of the decoded quantum
logical observable is exactly one minus twice the probability that the
decoded logical bit of the corresponding classical process has changed.
\end{informalproposition}

The formal statement is Proposition~\ref{prop:reduction-sync-formal} in Section \ref{sec:results}. 
The proofs of the two propositions 
are essentially the same.

Our starting point for proving the propositions is 
the  
framework of Alicki et al.~\cite{alicki_thermal_2010}, who showed 
(in the case of the $4D$ toric code in a thermal bath) that the Heisenberg evolution of the logical observables decouples into $X$ and $Z$ sectors, reducing 
quantum memory lifetime to the decay of equilibrium autocorrelation 
functions. 
The decoupling proof of~\cite{alicki_thermal_2010}
does not extend to non-equilibrium
syndrome-based Lindbladians straightforwardly, since it works directly 
on the specific quantum evolution, and no classical 
processes are even defined. 
The key to making progress is the central notion that 
we define of  {\it syndrome based Lindbladian} 
(Definition \ref{def:syndrome-based-lindbladian}). 
Roughly, these are Lindbladians whose $X$ and $Z$ jump operators depend only on the values of the $X$ and $Z$ syndromes, respectively. 
The main technical contribution is in setting up the right framework,  identifying 
the algebraic structure that such Lindbladians induce, and making use of this structure to define Markov processes on classical {\it functions}, rather than on observables. 

The initial setting is standard: 
for each classical component of the CSS code, we partition codewords into two logical cosets, labeled by the eigenvalue of the logical $X$ or $Z$ to be remembered, and assign them bit values $0$ and $1$ - these are the classical logical bits to be remembered. We associate to each component an $X$- or $Z$-sector process, respectively, initialized at a codeword but allowed to evolve to arbitrary configurations. At readout, we compute the syndrome, apply a fixed noiseless decoder, and define the {\it decoded logical bit} $b$ as the logical bit of the resulting codeword.

To proceed, we consider the notion of {\it decoded logical observable} - this is an observable that can be derived from the bare logical observable by conjugating with ideal decoding. 
The fundamental observation is that, for a CSS code, the Heisenberg
evolution of each decoded logical observable 
stays within a certain operator subspace, 
under any syndrome-based Lindbladian. More specifically, the decoded observable
$T_X$ is the bare logical $\bar X$ times an operator that when restricted to a given syndrome, is simply a syndrome-dependent sign (i.e., a scalar operator) that
lies in the $X$-stabilizer algebra $\mathcal A_X$. This means that $T_X$  belongs to the
\emph{logical sector subspace} 
 $\mathcal V_X$ spanned by the operators
$\bar XA$ with $A\in\mathcal A_X$. Any allowed term of the syndrome-based Lindbladian maps one element of $\mathcal{V}_X$ to another; hence, $\mathcal{V}_X$ is closed under its dynamics. 
Of course, all this applies to the $Z$ sector in the same way.  

The key point is that with closure established, the sector dynamics translate to a
classical description, namely, to a Markov process on \emph{classical functions} on $\mathcal{R}^{2^n}$.
We define linear maps $\Phi_X$ and $\Phi_Z$
that identify elements of $\mathcal V_X$ and $\mathcal V_Z$ with
scalar functions on the binary error vector spaces. A direct
computation shows that the Heisenberg generator, conjugated through
$\Phi_X$, coincides with the backward generator of a classical
continuous-time Markov chain on the $X$-error vectors, and that
$\Phi_X(T_X)=(-1)^{b_X}$, with $b_X$ the decoded logical value. This
gives an exact identity between the quantum decoded logical
{\it autocorrelation} and the probability that the classical logical-error
indicator has flipped, which is exactly the statement of the proposition.

The above can be summarized in one commuting
square, established as Proposition~\ref{prop:classical-repr} 
and
reproduced here (notation of Section~\ref{sec:async}):
\begin{equation*}
\begin{tikzcd}[row sep=3em, column sep=5em, ampersand replacement=\&]
O \in \mathcal{V}_X
  \arrow[r, "\Phi_X"]
  \arrow[d, "\mathcal{W}^*(t{,}0)"']
\& \Phi_X(O) \in \mathbb{R}^{\Omega_X}
  \arrow[d, "M^{(X)}(t{,}0)"] \\
\mathcal{W}^*(t{,}0)(O) \in \mathcal{V}_X
  \arrow[r, "\Phi_X"']
\& M^{(X)}(t{,}0)\,\Phi_X(O) \in \mathbb{R}^{\Omega_X}
\end{tikzcd}
\end{equation*}
Heisenberg evolution on the left; the classical sector chain
on the right; the representation map $\Phi_X$ across. This diagram is the underlying technical engine of essentially all results in this paper.

To prove the equivalence theorems 
(Theorem~\ref{thm:main-async-informal}, Theorem~\ref{thm:main-spca-sync-informal}), 
it remains to deduce a bound on the fidelity of the expectation values with respect to any logical observable, given the classical memory. 
This simply follows from showing that the effective logical channel of the entire process --- the 
combined action of encoding, noise and correction, and decoding on the 
encoded qubit --- is a Pauli channel. This in turn can be shown from a standard symmetry argument: each stage of the channel commutes with conjugation by the
logical Paulis up to at most a sign, and this symmetry is restrictive
enough to force the channel's Heisenberg adjoint to act diagonally in the
Pauli basis.
Once the channel is known to be Pauli, its worst-case fidelity is 
controlled by how close it is to the identity, which is precisely 
what the sector autocorrelations measure.

\subsubsection{Realizing classical PCAs by Floquet Lindbladians (Theorem \ref{thm:periodic-informal}).}
\label{subsec:periodic-overview}

Here we describe the main technical ideas in the proof of Theorem~\ref{thm:periodic-informal}, where we start with a pair of SPCAs
and, based on them, construct a time-periodic Lindbladian whose induced classical stochastic processes at boundaries of the periodicity cycles is  exactly that described by the SPCAs. Each cycle of this
Lindbladian consists of three stages; the cycles repeat indefinitely. In the first
stage, ``Reset'', the classical registers of every site are cleared. In the
second, ``Read'', each site reads its neighboring syndrome bits and records
in a decision register whether a correction (in the respective sector \(X\)
or \(Z\)) is needed. In the third, ``Execute'', the correction is applied
conditioned on this decision, and an execution flag ensures that it is
applied only once.

Implementing the Lindbladian makes use of 
classical noisy registers to
store the correction decision and its execution status during each
cycle. We note that formally, adding classical registers takes the generator outside of the
syndrome-based class as we defined it in \ref{def:syndrome-based-lindbladian}. This is because we preferred to keep the  definition in the first parts of the paper simple, for readability. However the definition and proofs of   Proposition~\ref{prop:reduction-async-formal}
and Theorem~\ref{thm:main-async-formal} extend to the setting of additional classical registers in a completely straightforward way (Subsection~\ref{subsec:register-extension}). 

To prove the theorem, the goal is to show that when considering the induced continuous time sector processes  
of the quantum system, given by the extended version of Theorem~\ref{thm:main-async-informal}, at the boundary of the correction cycles, namely at the end of each three-period cycle, the $X$ and $Z$ sectors 
are valid runs of the SPCAs we started with.

The challenge is that an SPCA applies its update rule synchronously at every
site, whereas the Lindbladian induces a stochastic continuous-time process,
which is random not only in its noise but also in whether and when each
correction jump occurs.

We first handle the simple case, and restrict attention to sites and cycles, such that the register's site operations all
complete their action ("reset", "read" or "execute") within the cycle time period (namely, the Poisson clock fired at least once); and furthermore, we assume in this simple case that no noise occurred in the site, nor in its nearby neighborhood, within that cycle.
It is easy to see that such a site performs
exactly its SPCA update for that cycle.  

We now have to handle the sites that do not belong to this simple case; we call those {\it faulty} sites. To proceed, we  show  that the correlations between faulty sites is very restricted. As  a first observation, we note that the faulty sites of a cycle are determined by the Poisson clocks
within that cycle, hence there are no correlations between different cycles. Secondly, spatially, correlations are confined to a constant size range. We can then apply a standard mechanism is classical probability theory, for bounding such locally correlated processes by a stochastic domination
argument, due to Liggett--Schonmann--Stacey theorem
\cite{liggett_domination_1997}. The faults in our locally correlated process can be covered by an independent fault set at a slightly larger rate (the covering requires a standard coupling argument~\cite{strassen_existence_1965}), so that the data process sampled at cycle boundaries is itself
an SPCA trajectory, with a somewhat larger noise rate.  
By taking the correction rate to be a large enough constant and the noise rate to be a small enough constant, this larger noise rate of the induced sector process that is easily related to the noise rate  of the SPCAs (Claim~\ref{clm:memory-transfer}) can be forced to be smaller than the threshold noise rate of the SPCAs.

\subsection{Related work}
\label{subsec:related-work}

Beyond the works discussed above, we mention three adjacent directions that are
not inputs to our results.

\subsubsection{Autonomous quantum error correction}

Continuous-time error correction~\cite{paz_continuous_1998} led to autonomous
schemes in which measurement and feedback are replaced by engineered
dissipation, from nanophotonic circuits~\cite{kerckhoff_designing_2010} to
bosonic and superconducting encodings~\cite{harrington_engineered_2022}. Sarma
and Mabuchi used the separability of $X$- and $Z$-type syndrome extraction to
decompose such CSS dynamics~\cite{sarma_gauge_2013}, a separability that our
sector reduction formalizes and makes exact. These schemes protect codes of fixed size, so
their lifetime does not diverge with the number of qubits. In extended systems,
measurement-free protection of the 2D surface code was proposed by coupling the
syndrome to a cooled auxiliary layer of classical
spins~\cite{fujii_measurement-free_2014}, driven-dissipative two-dimensional
candidates were supported
numerically~\cite{lieu_candidate_2024,kumar_self-correction_2026}, local
Lindbladians were shown numerically to stabilize topological order against
heralded noise~\cite{chirame_stabilizing_2024}, and upper bounds on lifetime of autonomous
memories built from Markovian decoders were derived~\cite{shtanko_bounds_2025}.

\subsubsection{Local decoders with classical processing}

Local cellular-automaton decoders for topological
codes~\cite{breuckmann_local_2016,kubica_cellular-automaton_2019}, including a
numerically studied 4D Toom-rule decoder for the 4D toric
code~\cite{breuckmann_local_2016}, assume a noiseless classical layer. Lake's
message-passing construction~\cite{lake_local_2025} has a provable threshold,
also under asynchronous updates, given $O(\mathrm{polylog}\,L)$ noiseless
classical bits per qubit, and was extended to a two-dimensional fault-tolerant
architecture assuming noiseless classical processing and a global
clock~\cite{selub_local_2026}. In the same paper, Lake provides an asynchronous construction in the same setting of  $O(\mathrm{polylog}\,L)$ noiseless
classical bits per qubit, supported by numerical analysis.   

\subsubsection{Active quantum matter}
Autonomous quantum memories are closely related to active quantum matter, where nonequilibrium dynamics can sustain steady-state currents and collective order. Khasseh et al.\ found a quantum
flocking phase in one-dimensional local, translationally invariant Lindblad
dynamics~\cite{khasseh_active_2025}. Engineered dissipation can stabilize quantum information, as shown numerically in~\cite{pastawski_quantum_2011} and rigorously here; an open question is whether analogous stabilization can instead emerge from the collective dynamics of an active phase, for instance through nonreciprocal currents that transport and annihilate topological defects, without the dynamics being explicitly designed as a decoder.

\subsection{Further Questions and Outlook}\label{subsec:outlook} 

This work makes significant progress toward autonomous quantum memories and raises several interesting and timely questions. 

\begin{itemize} 
\item Given the equivalence theorems that provide a path from classical memory to quantum memory, one is led to ask - what {\bf new autonomous quantum memories constructions} can be found, making use of the reduction to the classical PCA memory question? The exact conditions on SPCAs required for this matter, for a given CSS code \ref{sec:classical-problem}, reduce the problem to an entirely classical one, enabling us to borrow results from the well established research area of classical PCAs.   Specifically, can our framework help in achieving autonomous synchronous quantum memories in $1$D? 
And can new, less demanding, and more efficient constructions of synchronous quantum memories in $2D$ be constructed using our framework? The $2D$ memory construction of \cite{dunnweber_quantum_2026} inherently relies on a hierarchical, concatenated structure, and 
 it is not clear how to make it more efficient or applicable to other codes. It is interesting to see whether the general framework presented here can be applied to derive synchronous quantum memories for other CSS codes. 
 \item 
The current reduction framework applies to local Pauli noise channels. Addressing {\bf general Markovian noise}, which includes coherent errors and mixed off-diagonal terms, remains an open problem for autonomous quantum memories. In contrast, active fault-tolerance protocols mitigate these terms through projective syndrome measurements, which form the basis for established threshold theorems against general local noise.
As another important generalization, it is desirable to generalize the framework proven here, to other classes of codes, and possibly beyond codes, to {\bf other settings exhibiting certain symmetry restrictions}.
\item 
A fundamental open problem for autonomous architectures is {\bf structural stability}, tightly related to the study of these questions from the broader point of view of noisy phases of matter. In Hamiltonian models, stability refers to the robustness of the topological degeneracy against local Hamiltonian perturbations~\cite{bravyi_topological_2010,bravyi_short_2011}. However, the corresponding notion of stability in the dissipative Lindbladian framework is not quite clear. While physical imperfections manifest as additional local terms within the master equation—such as coherent dynamics or miscalibrated dissipative jumps—defining and proving the robustness of the macroscopic memory lifetime against superoperator perturbations of general types remains an open conceptual and analytical challenge. 
\end{itemize}

We end with two fundamental questions. 
First, an important question that this work highlights is whether autonomous quantum memories in the time-independent asynchronous setting, namely, without any periodicity, exist. 
Our Theorem \ref{thm:periodic-informal} requires the Lindbladian to be periodic in time; 
the reliance on scheduled stages necessitates a global clock. This global synchronization introduces a hidden form of non-locality. 
As mentioned in Subsubsection \ref{sec:WhyLindbladians}, to the best of our knowledge, no fully autonomous quantum memory, i.e., one that uses a time-independent local TI Lindbladian, is rigorously proven to exist (beyond the thermal equilibrium setting). 

One may hope to derive fully autonomous quantum memories from PCA-type ones; 
To the best of our knowledge, our Corollary \ref{thm:toom-sync-informal} for the 4D Toom's rule toric code, together with the 2D construction of \cite{dunnweber_quantum_2026} are the only known rigorously proven autonomous PCA-type quantum memories protecting against Pauli noise. 
Can such PCA-type results be upgraded to the fully autonomous domain? 
If it is 
possible to prove a variant of Theorem 
\ref{thm:periodic-informal} that 
outputs a fully autonomous quantum memory, with a time-independent Lindbladian, that would constitute 
a black-box converter from classical synchronous memories to quantum asynchronous ones. Interestingly, such a black-box result would also imply a black-box synchronous-to-asynchronous result that (while derived via a quantum route), lives solely within the classical domain -- a result that would be surprising contribution outside of the quantum domain.

Finally, a major question is to arrive at a more physical characterization of the conditions and requirements for 
physical systems to exhibit autonomous quantum memory. Much of the current literature in the area is based on computer-science-like, protocol-inspired constructions.  The hope is to be able to 
get a more physics-based understanding. 
Our framework bridges the languages; we note that the notion of syndrome-based Lindbladians we use here has a natural physical meaning when we consider the syndrome as the energy excitations of the system. Our hope is that this work will lead to a better understanding of the physics of quantum memory and also quantum fault tolerance.  

\subsection{Organization of the remainder of the paper}
The remainder of the paper is organized as follows.
Section~\ref{sec:setup} sets up the framework. It first introduces the
objects common to both settings: CSS codes and their sector configuration
spaces, decoded logical observables, the effective logical channel, and the
quantum and classical memory definitions. It then defines syndrome-based
Lindbladians and their induced continuous-time sector processes, and
syndrome-based quantum channels together with their induced
discrete-time sector processes, locality, and SPCAs.
Section~\ref{sec:results} states the main results formally.
Section~\ref{sec:async} proves the asynchronous reduction: the algebraic
decoupling of the two sectors, the invariance of the sector subspaces under
the Lindbladian, the classical representation and propagator intertwining,
and the conversion of the resulting autocorrelation identity into a memory
bound through the Pauli structure of the effective logical channel. It
closes by showing how the thermal memory result of Alicki et
al.~\cite{alicki_thermal_2010} for the 4D toric code is recovered
(Subsection~\ref{subsec:alicki}).
Section~\ref{sec:sync-reduction} proves the synchronous analogue and
constructs the quantum lift of a pair of SPCAs
(Subsection~\ref{subsec:sync-spca-realization}).
Section~\ref{sec:toom} verifies that the 4D Toom rules satisfy the
hypotheses of this construction, using Ahn's threshold analysis, and thereby
proves Theorem~\ref{thm:toom-sync-formal}.
Section~\ref{sec:periodic} returns to continuous time. It extends the
reduction to Lindbladians with classical registers, constructs the periodic
Lindbladian, and shows that its sector processes, sampled at cycle
boundaries, are runs of the original SPCAs, which proves
Theorem~\ref{thm:periodic-formal}.
Section~\ref{sec:classical-problem} isolates the purely classical problem
that suffices as input to both constructions.
Appendices~\ref{app:setup}, \ref{app:reduction}, \ref{app:sync},
\ref{app:registers} and~\ref{app:pca} contain the deferred proofs for
Section~\ref{sec:setup}, Section~\ref{sec:async},
Section~\ref{sec:sync-reduction},
Subsection~\ref{subsec:register-extension}, and the remainder of
Section~\ref{sec:periodic}, respectively.

\section{Setup and notation}\label{sec:setup}

This section establishes the mathematical formalism for autonomous quantum memory and its classical reduction. We define the CSS stabilizer code setting, and the class of syndrome-based Lindbladians describing a continuous-time physical dynamics depending only on the syndrome of the system. We define the induced continuous classical Markov processes and the memory time. To bound the memory time, we work in the Heisenberg picture and introduce {\it  decoded logical observables}, that incorporate the decoding map directly into the observable. 

\subsection{General setup}
\label{subsec:general-setup}

\subsubsection{CSS codes, syndromes, error vectors}
We recall the standard definitions of CSS, syndromes and error vectors, to set up the notation.  
\begin{definition}[CSS codes and stabilizer algebras]
\label{def:css-code} Let the Hilbert space of our physical system be $\mathcal{H}$.
Let $H_X \in \mathbb{F}_2^{m_X \times n}$ and $H_Z \in \mathbb{F}_2^{m_Z \times n}$ be parity-check matrices satisfying the CSS compatibility condition $H_XH_Z^T=0$, and define the classical component codes by $\mathcal{C}_X:=\ker H_X$ and $\mathcal{C}_Z:=\ker H_Z$. The rows of these matrices are permitted to be linearly dependent. An $n$-qubit CSS stabilizer code is defined by the corresponding generating sets of stabilizer operators:
\begin{equation}
    \mathcal S_X=\{ S^X_1,\dots,S^X_{m_X}\},
    \qquad
    \mathcal S_Z=\{S^Z_1,\dots,S^Z_{m_Z}\},
\end{equation}
where $S_i^X, S_i^Z\in P_n$, the Pauli group on $n$ qubits. The \emph{code space} $\mathcal H_{\mathrm{code}}\subseteq\mathcal H$ is the joint $+1$ eigenspace of all stabilizers in $\mathcal S_X\cup\mathcal S_Z$, with orthogonal projector $\Pi_{\mathrm{code}}$. 
We also define the associated Abelian stabilizer algebras $\mathcal A_X:=\mathrm{Alg}(\mathcal S_X)$ and $\mathcal A_Z:=\mathrm{Alg}(\mathcal S_Z)$, and the full stabilizer algebra $\mathcal A:=\mathrm{Alg}(\mathcal S_X\cup\mathcal S_Z)$, which is Abelian and contains both.
The stabilizer Hamiltonian is
\begin{equation}
    H = -\sum_{a=1}^{m_X} S_a^X - \sum_{a=1}^{m_Z} S_a^Z 
    \in \mathcal{A}_X + \mathcal{A}_Z.
\end{equation}
\end{definition}
More generally, the Hamiltonian may be any Hermitian element of $\mathcal A$ (Definition~\ref{def:syndrome-based-lindbladian}).

Distances between qubits are measured in the interaction graph of the
code, in which two qubits are adjacent if some stabilizer generator acts
on both. For a set of qubits \(A\), we write \(N_R(A)\) for the set of
qubits within distance \(R\) of \(A\), \(\bar A\) for its complement,
and \(\mathcal H_A\) for the Hilbert space of the qubits in \(A\).

We introduce the binary configuration spaces associated with the two CSS sectors.

\begin{definition}[Sector configuration spaces and syndromes]
Denote
\begin{equation}
\Omega_X:=\mathbb F_2^n,
\qquad
\Omega_Z:=\mathbb F_2^n .
\end{equation}
An element $\mathbf x\in\Omega_X$ is called an \emph{$X$-sector configuration}, and its syndrome is
$H_X\mathbf x\in\mathbb F_2^{m_X}.$
The space of valid $X$-syndromes is therefore
$\Sigma_X:=\operatorname{im}H_X.$
Likewise, for $\mathbf z\in\Omega_Z$, its $Z$-syndrome is $H_Z\mathbf z$, and
$\Sigma_Z:=\operatorname{im}H_Z.$
\end{definition}

The sector label follows the logical observable being tracked. Thus the
$X$ sector tracks the $X$-logical observable (defined below) and is sensitive to physical $Z$-type Pauli
flips, which are detected by the $X$ stabilizers; symmetrically, the $Z$
sector tracks the $Z$-logical observable and is sensitive to physical $X$-type flips. The
convention is summarized by
\begin{center}
\begin{tabular}{lllll}
\hline
Sector
& Tracks
& Configuration space
& Detecting stabilizers
& Syndrome map
\\
\hline
$X$
& logical $X$
& $\Omega_X=\mathbb{F}_2^n$
& $\mathcal{S}_X$
& $\mathbf{x}\mapsto H_X\mathbf{x}$
\\
$Z$
& logical $Z$
& $\Omega_Z=\mathbb{F}_2^n$
& $\mathcal{S}_Z$
& $\mathbf{z}\mapsto H_Z\mathbf{z}$
\\
\hline
\end{tabular}
\end{center}
More generally, any object related to the $X$-sector carries a superscript or
subscript $X$, except for the $Z$-Pauli errors themselves that are those relevant to the $X$-sector but retain their own
letter and are still
written with $Z$'s.

\begin{definition}[Syndrome projectors]
\label{def:syndrome-projectors}
Recall that the valid $X$-syndrome space $\Sigma_X \subseteq \mathbb{F}_2^{m_X}$ is the image of $H_X$. For a valid syndrome string $s \in \Sigma_X$, define the corresponding $X$-syndrome projector $\Pi_s^X$. This projects onto the set of states in $\cal{H}$  with syndrome $s$; this subspace is a simultaneous eigenspace of all $X$- stabilizers, where the associated eigenvalues are the relevant syndrome bit.
We thus define the $X$  syndrome projector (with respect to the syndrome $s$) as: 
\begin{equation}
    \Pi_s^X := \prod_{a=1}^{m_X} \frac{\mathbb{I} + (-1)^{s_a}S_a^X}{2}.
\end{equation}
The corresponding $Z$-syndrome projectors $\Pi_s^Z$ are defined symmetrically.

We further introduce the following notation of $A^X(s)$. 
Because the Abelian algebra $\mathcal{A}_X$ is generated by the mutually commuting stabilizers $\mathcal{S}_X$, any operator $A \in \mathcal{A}_X$ acts as a scalar multiple of the identity on $\Pi_s^X(\mathcal{H})$. Thus, for 
\(A^X\in\mathcal A_X\) and \(s\in\Sigma_X\), we can write
\begin{equation}\label{def:eigenvalue-functions}
    A^X\Pi_s^X\equiv A^X(s)\Pi_s^X.
\end{equation}
with $A^X(s)$ being the 
scalar eigenvalue of \(A^X\) acting on the \(X\)-syndrome projector. 
Note that with this
identification, products can be evaluated pointwise:
    $(A^XB^X)(s)=A^X(s)B^X(s).$ The $Z$ version \(A^Z(s)\) for $s\in\Sigma_Z$ is defined accordingly.

\end{definition}

\subsubsection{Classical logical values}
\label{subsec:classical-logical-labels}
The sector processes evolve binary strings in $\Omega_X$, which need not be
codewords. We first define the logical value carried by a codeword, then extend
this value to an arbitrary string by readout decoding. We develop the $X$
sector explicitly; the $Z$ sector is obtained by symmetry.

\begin{definition}[Logical classes and logical values]
\label{def:logical-values}
The CSS condition $H_XH_Z^T=0$ of Definition~\ref{def:css-code} places
every row of $H_Z$ in $\ker H_X=\mathcal C_X$, so
$\operatorname{row}(H_Z)\subseteq\mathcal C_X$. Two codewords
$\mathbf c,\mathbf c'\in\mathcal C_X$ lie in the same \emph{$X$-logical
class} when
$\mathbf c\oplus\mathbf c'\in\operatorname{row}(H_Z)$. Thus the
$X$-logical classes are the cosets
\begin{equation}
    [\mathbf c]_X
    :=
    \mathbf c+\operatorname{row}(H_Z)
    \in
    \mathcal C_X/\operatorname{row}(H_Z).
\end{equation}
An \emph{$X$-logical value} is a nonzero linear map
\begin{equation}\label{eq:mu-definition}
    \mu_X:
    \mathcal C_X/\operatorname{row}(H_Z)
    \longrightarrow
    \mathbb F_2,
\end{equation}
and $\mu_X([\mathbf c]_X)$ is the binary value stored by the codeword
$\mathbf c$.
\end{definition}

Requiring $\mu_X$ to be nonzero ensures that it reads out a genuine binary
value rather than the identically zero function. If
$\dim\bigl(\mathcal C_X/\operatorname{row}(H_Z)\bigr)>1$, there is more
than one admissible choice in Eq.~\eqref{eq:mu-definition}; the choice used in
this paper is fixed in Definition~\ref{def:logical-coordinate-identification}.

A string produced by the dynamics need not be a codeword. Its logical value
at readout is obtained by first decoding its syndrome.

\begin{definition}[Readout decoders]
\label{def:readout-decoder}
An \emph{$X$-readout decoder} is a map
\begin{equation}\label{eq:decoder-properties}
    \mathsf{Dec}_X:\Sigma_X\longrightarrow\Omega_X,
    \qquad
    H_X\mathsf{Dec}_X(s)=s
    \quad\text{for all }s\in\Sigma_X,
    \qquad
    \mathsf{Dec}_X(\mathbf 0)=\mathbf 0.
\end{equation}
The \emph{decoded string} associated with $\mathbf x\in\Omega_X$ is
$\mathbf x\oplus\mathsf{Dec}_X(H_X\mathbf x)$.
\end{definition}

The decoder is applied only at readout and need not be geometrically
local, in contrast with the local correction dynamics considered in
the applications below. No
linearity of $\mathsf{Dec}_X$ is assumed. By
Eq.~\eqref{eq:decoder-properties},
$H_X\bigl(\mathbf x\oplus\mathsf{Dec}_X(H_X\mathbf x)\bigr)=\mathbf 0$, so the
decoded string is a codeword and its logical value is well defined.

\begin{definition}[Decoded logical value]
\label{def:decoded-logical-value}
The \emph{decoded $X$-logical value} of a string $\mathbf x\in\Omega_X$ is
\begin{equation}\label{eq:b-definition}
    b_X:\Omega_X\longrightarrow\mathbb F_2,
    \qquad
    b_X(\mathbf x)
    :=
    \mu_X\!\left(
        \left[
            \mathbf x\oplus\mathsf{Dec}_X(H_X\mathbf x)
        \right]_X
    \right).
\end{equation}
\end{definition}

If $\mathbf w\in\mathcal C_X$ then $H_X\mathbf w=\mathbf 0$ and
$\mathsf{Dec}_X(\mathbf 0)=\mathbf 0$, so the decoded string of $\mathbf w$ is
$\mathbf w$ itself and $b_X(\mathbf w)=\mu_X([\mathbf w]_X)$; in particular
$b_X(\mathbf 0)=0$. Thus $b_X$ extends the logical value from
$\mathcal C_X$ to all of $\Omega_X$.
 It is not linear, what does hold is additivity under shifts by
codewords.
\begin{claim}[Additivity by codeword]
    The decoded logical value $b_X$ satisfies additivity for codewords:
    \begin{equation}
           b_X(\mathbf x\oplus\mathbf w)
    =
    b_X(\mathbf x)\oplus b_X(\mathbf w),
    \qquad
    \mathbf x\in\Omega_X,\;
    \mathbf w\in\mathcal C_X. 
    \end{equation}
\end{claim}
\begin{proof}
 For $\mathbf w\in\mathcal C_X$ and $\mathbf x\in\Omega_X$,
    $H_X(\mathbf x\oplus\mathbf w)
    =
    H_X\mathbf x$,
so the decoder returns the same error vector on $\mathbf x\oplus\mathbf w$
and on $\mathbf x$, and their decoded strings differ by $\mathbf w$. Using linearity of $\mu_X$,
\begin{equation}\label{eq:observable-covariance}
    b_X(\mathbf x\oplus\mathbf w)
    =
    b_X(\mathbf x)\oplus\mu_X([\mathbf w]_X)
    =
    b_X(\mathbf x)\oplus b_X(\mathbf w).
\end{equation}
\end{proof}

\subsubsection{Quantum readout and decoded logical observables}
\label{subsec:quantum-readout}
Having defined the classical sector processes, the logical values
$\mu_X,\mu_Z$, their decoded extensions $b_X,b_Z$, and the normalized
sector error processes, we now introduce the designated quantum logical
qubit and its readout observables. Their relation to the classical
logical-error indicators is established in the following subsection.

\begin{definition}[Logical encoding]
\label{def:logical-subsystem}
The CSS code encodes $k$ logical qubits. Throughout this paper we study the ability for a system to remember a single logical qubit: prior to encoding, we designate one input
qubit $Q$ as the carrier of the logical state to be remembered, and refer to
the remaining $n-1$ physical qubits as the remainder subsystem. When $k>1$, the
other $k-1$ logical degrees of freedom are absorbed into this remainder
subsystem and held in a fixed state. This defines the tensor
decomposition of the unencoded physical Hilbert space:
\begin{equation}
    \mathcal H = \mathcal H_Q \otimes \mathcal H_{\mathrm{rest}}, \qquad \mathcal H_Q = \mathbb C^2.
\end{equation}
Let $X_Q, Y_Q, Z_Q \in \mathcal{B}(\mathcal{H}_Q)$ denote the Pauli operators acting on the designated carrier qubit $Q$ before encoding.

Let $U$ be a unitary encoding operator on $\mathcal H$ such that
$U(\mathcal H_Q\otimes|0\rangle_{\mathrm{rest}})$ lies in the code space
$\Pi_{\mathrm{code}}\mathcal H$ of Definition~\ref{def:css-code}
(for $k=1$ this image is the whole code space). We choose $U$ so that the
Pauli operators of the designated input qubit are represented on the code
space by a conjugate pair of CSS logical Pauli strings
$\bar X\in\{I,X\}^{\otimes n}$ and
$\bar Z\in\{I,Z\}^{\otimes n}$:
\begin{equation}\label{eq:encoding-condition}
    \bar X\Pi_{\mathrm{code}}
    =
    U(X_Q\otimes I_{\mathrm{rest}})U^\dagger\Pi_{\mathrm{code}},
    \qquad
    \bar Z\Pi_{\mathrm{code}}
    =
    U(Z_Q\otimes I_{\mathrm{rest}})U^\dagger\Pi_{\mathrm{code}}.
\end{equation}
\end{definition}

The second requirement is a restriction on $U$ which can always be met for a CSS code~\cite[§3.2, §4.2, §5.3]{gottesman_stabilizer_1997}.

\begin{definition}[Bare logical operators]
\label{def:bare-logicals}
The bare logical operators are the $n$-qubit Pauli strings $\bar{X}$ and $\bar{Z}$ of Eq.~\eqref{eq:encoding-condition}, which implement the encoding of the operators $X_Q$ and $Z_Q$ on the code space. Let $\ell_X, \ell_Z \in \mathbb{F}_2^n$ be $n$-bit strings that indicate the locations of the non-identity Pauli operators in $\bar{X}$ and $\bar{Z}$, respectively.
\end{definition}
We note that there could be many choices of these bare logical operators and their support. We fix one choice.

\begin{definition}[Classical logical values of the designated qubit]
\label{def:logical-coordinate-identification}
We choose the logical-value maps $\mu_X,\mu_Z$ of
Definition~\ref{def:logical-values} to be
\begin{equation}
    \mu_X([\mathbf c]_X):=\ell_X\cdot\mathbf c,
    \quad \mathbf c\in\mathcal C_X,
    \qquad
    \mu_Z([\mathbf c]_Z):=\ell_Z\cdot\mathbf c,
    \quad \mathbf c\in\mathcal C_Z,
\end{equation}
with inner products over $\mathbb F_2$.
\end{definition}

The corresponding functions $b_X$ and $b_Z$ are the decoded logical
values of Definition~\ref{def:decoded-logical-value}. When evaluated on
the normalized error processes,
$b_X(\xi_t^{(X)})$ and $b_Z(\xi_t^{(Z)})$ are the associated
logical-error indicators.

\begin{claim}[Validity of logical value choices]
\label{clm:logical-values-valid}
The maps of Definition~\ref{def:logical-coordinate-identification} are well
defined, linear and nonzero, and are therefore logical values in the sense of
Definition~\ref{def:logical-values}.
\end{claim}

\begin{proof}
The bare logical operators commute with every stabilizer, so
$H_Z\ell_X=\mathbf 0$ and $H_X\ell_Z=\mathbf 0$, and they anticommute with
each other, so $\ell_X\cdot\ell_Z=1$. If $[\mathbf c']_X=[\mathbf c]_X$, then
$\mathbf c'=\mathbf c\oplus\mathbf h$ for some
$\mathbf h\in\operatorname{row}(H_Z)$, and $\ell_X\cdot\mathbf h=0$ since
$H_Z\ell_X=\mathbf 0$; hence $\ell_X\cdot\mathbf c'=\ell_X\cdot\mathbf c$ and
$\mu_X$ is well defined. It is linear because the inner product is. Finally,
$\ell_Z\in\ker H_X=\mathcal C_X$ and
$\mu_X([\ell_Z]_X)=\ell_X\cdot\ell_Z=1$, so $\mu_X\neq0$. The $Z$-sector
argument is identical under $X\leftrightarrow Z$.
\end{proof}

The quantum memory is assessed after a final, one-shot readout. We
incorporate this procedure algebraically by pulling the final syndrome
measurement and Pauli correction back onto the logical observable, thereby
obtaining an effective observable evaluated on the physical state before
readout~\cite{alicki_thermal_2010}. We call this final operation
\emph{decoding}, reserving the term \emph{correction} for the continuous
dissipative operations applied during the evolution.

\begin{definition}[Pauli readout corrections]
\label{def:pauli-readout-corrections}
Let $\mathsf{Dec}_X$ and $\mathsf{Dec}_Z$ be the classical readout
decoders fixed in Definition~\ref{def:readout-decoder}. Their associated
Pauli readout corrections are
\begin{equation}
    R_X(s_X):=Z^{\mathsf{Dec}_X(s_X)},
    \qquad
    R_Z(s_Z):=X^{\mathsf{Dec}_Z(s_Z)},
\end{equation}
where $Z^{\mathbf v}$ and $X^{\mathbf v}$ denote the corresponding Pauli
strings supported on $\mathbf v\in\mathbb F_2^n$.
Thus $R_X(s_X)$ removes the $X$-syndrome $s_X$, while $R_Z(s_Z)$ removes
the $Z$-syndrome $s_Z$. For a joint syndrome $(s_X,s_Z)$, we take the
readout correction to be
    $R_X(s_X)R_Z(s_Z).$
\end{definition}
Since each correction commutes with the stabilizers of the opposite
sector, this joint correction maps the corresponding joint syndrome
sector to the code space. The normalization
$\mathsf{Dec}_X(\mathbf 0)=\mathsf{Dec}_Z(\mathbf 0)=\mathbf 0$ gives
    $R_X(\mathbf 0)=R_Z(\mathbf 0)=\mathbb I$.
\begin{definition}[Physical decoding channel]
\label{def:physical-decoding-channel} 
The final physical decoding channel $\mathcal{E}_{\mathrm{dec}}: \mathcal{T}(\mathcal{H}) \to \mathcal{T}(\mathcal{H})$ projects the state onto the joint eigenspaces of the two syndrome projectors $\Pi_{s}^X$ and $\Pi_{s'}^Z$ (these projectors commute --- see Remark~\ref{rem:decoder-commute} below) and applies the corresponding Pauli corrections assigned by the decoders $R_X,R_Z$.
Because the CSS decoder processes sectors independently (Definition~\ref{def:pauli-readout-corrections}), the global channel is the composition $\mathcal{E}_{\mathrm{dec}} = \mathcal{E}_{\mathrm{dec}}^{(X)} \circ \mathcal{E}_{\mathrm{dec}}^{(Z)}$. The $X$-sector decoding channel is defined by the Kraus representation:
\begin{equation}
    \mathcal{E}_{\mathrm{dec}}^{(X)}(\rho) := \sum_{s \in \Sigma_X} R_X(s) \Pi_s^X \rho \Pi_s^X R_X(s)^\dagger.
\end{equation}
Its adjoint maps a physical observable $O \in \mathcal{B}(\mathcal{H})$ via:
\begin{equation}\label{eq:adjoint-decoding}
    \mathcal{E}_{\mathrm{dec}}^{(X)*}(O) = \sum_{s \in \Sigma_X} \Pi_s^X R_X(s)^\dagger O R_X(s) \Pi_s^X.
\end{equation}
The $Z$-sector channel $\mathcal{E}_{\mathrm{dec}}^{(Z)}$ and its adjoint are defined symmetrically.
\end{definition}

\begin{remark}\label{rem:decoder-commute}
The sector decoding channels commute, $[\mathcal{E}_{\mathrm{dec}}^{(X)}, \mathcal{E}_{\mathrm{dec}}^{(Z)}] = 0$. This follows because the CSS projectors commute ($[\Pi_s^X, \Pi_{s'}^Z]=0$), and the restricted Pauli assignments ensure $R_X(s) \in \{I,Z\}^{\otimes n}$ commutes with $\Pi_{s'}^Z$ and $R_Z(s') \in \{I,X\}^{\otimes n}$ commutes with $\Pi_s^X$. While the Pauli corrections $R_X$ and $R_Z$ may anticommute, any acquired global phase cancels during conjugation, namely, when considering the operation as super-operators. Therefore, the order of composition in $\mathcal{E}_{\mathrm{dec}} = \mathcal{E}_{\mathrm{dec}}^{(X)} \circ \mathcal{E}_{\mathrm{dec}}^{(Z)}$ is arbitrary.
\end{remark}

\begin{remark}[Trace-preservation of the decoding channel]
The decoding channel $\mathcal{E}_{\mathrm{dec}}^{(X)}$ is 
trace-preserving on all of $\mathcal{H}$. Since the stabilizers 
$\{S_a^X\}$ are commuting Hermitian operators with eigenvalues $\pm 1$, 
their joint eigenspace projectors partition the identity:
$\sum_{s \in \mathbb{F}_2^{m_X}} \Pi_s^X = \mathbb{I}$.
For $s \notin \Sigma_X$, the projector $\Pi_s^X = 0$ vanishes, 
so the sum reduces to $\sum_{s \in \Sigma_X} \Pi_s^X = \mathbb{I}$. 
Since each $R_X(s)$ 
is unitary, the Kraus completeness relation 
$\sum_{s \in \Sigma_X} \Pi_s^X R_X(s)^\dagger R_X(s) \Pi_s^X = \mathbb{I}$ 
holds, confirming that $\mathcal{E}_{\mathrm{dec}}^{(X)}$ is 
trace-preserving.
\end{remark}

\begin{definition}[Decoded logical observables]
\label{def:decoded-observables}
The decoded logical observables $T_X$ and $T_Z$ incorporate the decoding by applying the adjoint decoding channel (Eq.~\eqref{eq:adjoint-decoding}) to the bare logical observables (namely, the bare logical operator from Definition~\ref{def:bare-logicals}, for operators that are the observables), producing an effective observable evaluated on the uncorrected physical state, as follows:
\begin{equation}
    T_X := \mathcal{E}_{\mathrm{dec}}^{(X)*}(\bar{X}) = \sum_{s \in \Sigma_X} \Pi_s^X R_X(s)^\dagger \bar{X} R_X(s) \Pi_s^X,
\end{equation}
\begin{equation}
    T_Z := \mathcal{E}_{\mathrm{dec}}^{(Z)*}(\bar{Z}) = \sum_{s \in \Sigma_Z} \Pi_s^Z R_Z(s)^\dagger \bar{Z} R_Z(s) \Pi_s^Z.
\end{equation}

\end{definition}

Measuring the decoded logical observables, namely taking $\mathrm{Tr}(T_X\rho)$ corresponds to first applying the decoding channel $\mathcal{E}_{\mathrm{dec}}$ on $\rho$ and then measuring the bare logical observable $\bar{X}$ (likewise for $Z$). The code-space autocorrelations of $T_X$ and $T_Z$ will be used below to
quantify the quantum memory; their relation to the classical sector
processes is established in Proposition~\ref{prop:reduction-async-formal}.

\subsubsection{Mixed code state and code-space inner product}
\begin{definition}[Mixed code state and code-space inner product] 
\label{def:reference-state}
The mixed code state $\rho_{\mathrm{code}} \in \mathcal{T}(\mathcal{H})$ is defined as the physical state corresponding to the maximally mixed state on the logical subsystem tensored with the remainder vacuum state:
\begin{equation}
    \rho_{\mathrm{code}} := U \left( \frac{1}{2} I_Q \otimes |0\rangle\langle 0|_{\mathrm{rest}} \right) U^\dagger.
\end{equation}
This reference state defines a code-space inner product for observables $A, B \in \mathcal{B}(\mathcal{H})$:
\begin{equation}
    \langle A, B \rangle_{\mathrm{code}} := \operatorname{Tr}\!\left(\rho_{\mathrm{code}} A^\dagger B\right).
\end{equation}
\end{definition}

\subsubsection{Effective logical channel and memory}

The encoding and final decoding maps are independent of the physical
dynamics acting between them.  We therefore define the effective logical
channel for an arbitrary physical CPTP map.

\begin{definition}[Effective logical channel]
\label{def:effective-logical-channel}
Let
\(
\mathcal W_*:\mathcal T(\mathcal H)\to\mathcal T(\mathcal H)
\)
be a CPTP map describing the physical evolution between encoding and
readout time, and let \(\mathcal W^*\) denote its Heisenberg adjoint.
The corresponding effective logical channel and its adjoint are
\begin{equation}\label{eq:effective-channel}
    \mathcal E_L[\mathcal W_*]
    :=
    \mathcal D_{\mathrm{dec}}
    \circ
    \mathcal W_*
    \circ
    \mathcal E_{\mathrm{enc}},
    \qquad
    \mathcal E_L^*[\mathcal W_*]
    :=
    \mathcal E_{\mathrm{enc}}^*
    \circ
    \mathcal W^*
    \circ
    \mathcal D_{\mathrm{dec}}^* .
\end{equation}
Here
\begin{align}
    \mathcal E_{\mathrm{enc}}(\rho_Q)
    &:=
    U\bigl(
        \rho_Q\otimes |0\rangle\langle0|_{\mathrm{rest}}
    \bigr)U^\dagger,
    &
    \mathcal E_{\mathrm{enc}}^*(O)
    &:=
    \langle0|_{\mathrm{rest}}
    U^\dagger O U
    |0\rangle_{\mathrm{rest}},
    \\
    \mathcal D_{\mathrm{dec}}(\rho)
    &:=
    \operatorname{Tr}_{\mathrm{rest}}
    \!\left[
        U^\dagger \mathcal E_{\mathrm{dec}}(\rho) U
    \right],
    &
    \mathcal D_{\mathrm{dec}}^*(O_Q)
    &:=
    \mathcal E_{\mathrm{dec}}^*
    \!\left(
        U(O_Q\otimes I_{\mathrm{rest}})U^\dagger
    \right).
\end{align}
\end{definition}

We use worst case fidelity to quantify memory. A system is said to have memory if its fidelity stays close to 1 over time.
\begin{definition}[Worst-case encoded fidelity]
\label{def:fidelity}
For a physical channel \(\mathcal W_*\), define
\begin{equation}
    F(\psi;\mathcal W_*)
    :=
    \left\langle\psi\right|
    \mathcal E_L[\mathcal W_*]
    \bigl(|\psi\rangle\langle\psi|\bigr)
    \left|\psi\right\rangle
\end{equation}
and
\begin{equation}
    F_{\min}[\mathcal W_*]
    :=
    \min_{|\psi\rangle\in\mathcal H_Q}
    F(\psi;\mathcal W_*).
\end{equation}
\end{definition}

\begin{definition}[Quantum memory]
\label{def:quantum-memory}
\label{def:quantum-memory-cycles}
Let
\(
\{\mathcal W_{*,a}\}_{a\in\mathcal I}
\)
be a family of physical evolution channels indexed by the allowed
readout points \(\mathcal I\).  The encoded system is a
\emph{quantum \(\memerr\)-memory on \(\mathcal I\)} if
\begin{equation}
    \inf_{a\in\mathcal I}
    F_{\min}[\mathcal W_{*,a}]
    \ge
    1-\memerr .
\end{equation}

For continuous-time dynamics up to time \(t\), we take
\(\mathcal I=[0,t]\).  For synchronous dynamics through round \(m\),
we take
\(\mathcal I=\{0,\ldots,m\}\).  For a periodic continuous-time dynamics
sampled at cycle boundaries \(t_r=r\Delta\), we take
\(\mathcal I=\{t_0,\ldots,t_m\}\), and speak of memory \emph{at cycle boundaries through cycle \(m\)}.
\end{definition}

The classical analogue of memory requires the logical bit to retain its value over time.
\begin{definition}[Classically \(q\)-remembering a function]
\label{def:classical-memory}
Let
\(
(\xi_a)_{a\in\mathcal I}
\)
be a classical stochastic process on a configuration space \(\Omega\),
indexed by a set of allowed readout points \(\mathcal I\) containing
the initial point \(a_0\).
We say that the process \emph{\(q\)-remembers}
\(b:\Omega\to S\) on \(\mathcal I\) if
\begin{equation}
    \sup_{a\in\mathcal I}
    \Pr\!\left[
        b(\xi_a)\ne b(\xi_{a_0})
    \right]
    \le q.
\end{equation}
We refer to the cases
\(\mathcal I=[0,t]\) and
\(\mathcal I=\{0,\ldots,m\}\) as remembering
\emph{up to time \(t\)} and \emph{through round \(m\)}, respectively.
\end{definition}

\subsection{Asynchronous (continuous-time) setup}
\label{subsec:async-setup}

We turn to discuss the specific dynamics of the system. We start with the more physically-motivated case of continuous time evolution.
\subsubsection{Syndrome-based Lindbladians}
We use the standard Lindblad form~\cite{alicki_thermal_2010} acting on an $n$-qubit system, in which the dissipative
part of the generator is a sum of \emph{dissipators}
\begin{equation}
    \mathcal D[L](\rho)
    =
    L\rho L^\dagger-\frac12\{L^\dagger L,\rho\},
\end{equation}
where $\rho$ is a density matrix and the \emph{jump operator}
$L\in\mathcal B(\mathcal H)$ is an arbitrary linear operator. In the generators we
consider, every jump operator is a single-qubit Pauli operator multiplied by an
element of the full stabilizer algebra $\mathcal A$. We write $P_i$ for the
single-qubit operator $P\in\{I,X,Y,Z\}$ acting on qubit $i$, with $I_i=\mathbb I$.

We present a class of Lindbladians which are in the heart of this framework. Informally, a Lindbladian is \emph{syndrome-based} when the total rate of every
physical flip in each CSS sector depends only on that sector's syndrome.

\begin{definition}[Syndrome-based Lindbladian]
\label{def:syndrome-based-lindbladian}
A Lindbladian on $\mathcal H$ is \emph{syndrome-based} for the CSS code of
Definition~\ref{def:css-code} if it admits a Lindblad representation of the form
\begin{equation}\label{eq:lindbladian}
    \mathcal L(t)(\rho)
    =
    -i[H(t),\rho]
    +
    \sum_{i=1}^n\sum_{P\in\{I,X,Y,Z\}}\sum_{\alpha}
    \mathcal D\bigl[P_i\,f_{P,i,\alpha}(t)\bigr](\rho),
\end{equation}
where $\alpha$ ranges over a finite index set, the Hamiltonian $H(t)$ is a Hermitian
element of $\mathcal A$, each \emph{syndrome function} $f_{P,i,\alpha}(t)$ is an
element of $\mathcal A$, and, for every qubit $i$ and every $t\ge0$, the
\emph{flip-rate operators}
\begin{equation}\label{eq:flip-rates}
\begin{aligned}
    \Gamma_i^{(X)}(t)
    &:=
    \sum_{\alpha}
    \Bigl(
        f_{Z,i,\alpha}(t)^\dagger f_{Z,i,\alpha}(t)
        +
        f_{Y,i,\alpha}(t)^\dagger f_{Y,i,\alpha}(t)
    \Bigr),
    \\
    \Gamma_i^{(Z)}(t)
    &:=
    \sum_{\alpha}
    \Bigl(
        f_{X,i,\alpha}(t)^\dagger f_{X,i,\alpha}(t)
        +
        f_{Y,i,\alpha}(t)^\dagger f_{Y,i,\alpha}(t)
    \Bigr),
\end{aligned}
\end{equation}
satisfy the \emph{rate condition}
\begin{equation}\label{eq:rate-condition}
    \Gamma_i^{(X)}(t)\in\mathcal A_X,
    \qquad
    \Gamma_i^{(Z)}(t)\in\mathcal A_Z.
\end{equation}
\end{definition}

The subscript $P$ of a syndrome function labels the Pauli type of its jump, while the
superscript of a flip-rate operator labels the sector, in line with the convention
of the table above. Thus $\Gamma_i^{(X)}(t)$ collects the jumps that flip qubit $i$
in the $X$-sector, those with Pauli factor $Z_i$ or $Y_i$, and $\Gamma_i^{(Z)}(t)$
those with Pauli factor $X_i$ or $Y_i$.

 For instance, if the only
jumps at qubit $i$ are $Y_i$ jumps, their total rate operator enters both
$\Gamma_i^{(X)}(t)$ and $\Gamma_i^{(Z)}(t)$; since the only Pauli string that is both
$X$-type and $Z$-type is the identity, $\mathcal A_X\cap\mathcal A_Z=\mathbb C\,\mathbb I$,
and this rate operator must be a scalar.

Note that the class of syndrome-based Lindbladians is closed under sums: the Hamiltonians add,
the jump families are joined, and the flip-rate operators add, so the rate condition
is preserved. Our motivating example is the sum of the stabilizer Hamiltonian, a
correction Lindbladian and local Pauli noise, which we now define.

\begin{definition}[Local Lindbladian]
\label{def:local-lindbladian}
A syndrome-based Lindbladian is \emph{\(R\)-local} if it admits a
representation as in Equation~\eqref{eq:lindbladian} in which, for all
\(t\), every jump operator \(P_i\,f_{P,i,\alpha}(t)\) is supported in
\(N_R(i)\), and \(H(t)=\sum_iH_i(t)\) with each \(H_i(t)\in\mathcal A\)
Hermitian and supported in \(N_R(i)\).
\end{definition}

Thus in an infinitesimal time step each qubit is influenced only by its
\(R\)-neighborhood. For Lindbladians with auxiliary registers
(Appendix~\ref{app:registers}), the registers attached to a site \(i\)
are counted as located at \(i\).

\begin{definition}[Correction Lindbladian]
\label{def:correction-lindbladian}
A \emph{correction Lindbladian} is a syndrome-based Lindbladian with $H=0$ whose jumps
are all of the form
\begin{equation}\label{eq:jump-operators}
    Z_i\,f_{Z,i,\alpha}(t)
    \quad\text{with}\quad
    f_{Z,i,\alpha}(t)\in\mathcal A_X,
    \qquad
    X_i\,f_{X,i,\alpha}(t)
    \quad\text{with}\quad
    f_{X,i,\alpha}(t)\in\mathcal A_Z,
\end{equation}
that is,
\begin{equation}
    \mathcal L_{\mathrm c}(t)
    =
    \sum_{i=1}^n\sum_{\alpha}\mathcal D\bigl[Z_i\,f_{Z,i,\alpha}(t)\bigr]
    +
    \sum_{i=1}^n\sum_{\alpha}\mathcal D\bigl[X_i\,f_{X,i,\alpha}(t)\bigr].
\end{equation}
\end{definition}

Each syndrome function of a correction Lindbladian lies in the algebra of the sector
its jump flips, so the rate condition holds jump by jump. Typically the syndrome
functions are local scaled projectors selecting prescribed syndrome patterns in the
vicinity of qubit~$i$: the jump $Z_i\,f_{Z,i,\alpha}(t)$ applies the correction $Z_i$
exactly when the nearby $X$-syndrome matches the pattern selected by
$f_{Z,i,\alpha}(t)$, and the auxiliary index $\alpha$ allows several patterns to
trigger the same correction. In the case of Pauli noise, $f$ would be proportional to identity, since it is oblivious to the system's state.

The second instance is local Pauli noise.
\begin{definition}[Pauli channel noise generator]\label{def:pauli-noise} 
The noise generator is a local Pauli noise channel,
\begin{equation}
    \mathcal{L}_{\mathrm{noise}}(\rho)
    =
    \sum_{i=1}^n\sum_{P\in\{X,Y,Z\}}
    \gamma_{P,i}\,\mathcal{D}[P_i](\rho),
\end{equation}
with constant rates $\gamma_{P,i}\ge0$ (see Remark~\ref{rem:time-dependent-noise} for time-dependent rates). We define the \emph{noise rate} as
\begin{equation}
    \physrate := \max_{i,P} \, \gamma_{P,i}.
\end{equation}
For the specific case of depolarizing noise, $\gamma_{X,i}=\gamma_{Y,i}=\gamma_{Z,i}=\physrate$.
\end{definition}

Pauli noise is syndrome-based, with scalar syndrome functions
$f_{P,i}=\sqrt{\gamma_{P,i}}\,\mathbb I$ for $P\in\{X,Y,Z\}$. Since a $Y_i$ error flips
qubit $i$ in both sectors, its flip-rate operators are $\eta_i^{(X)}\mathbb I$ and
$\eta_i^{(Z)}\mathbb I$, where
\begin{equation}\label{eq:noise-flip-rates}
    \eta_i^{(X)} := \gamma_{Z,i}+\gamma_{Y,i},
    \qquad
    \eta_i^{(Z)} := \gamma_{X,i}+\gamma_{Y,i};
\end{equation}
for depolarizing noise, $\eta_i^{(X)}=\eta_i^{(Z)}=2\physrate$. Consequently, with $H$
the stabilizer Hamiltonian, $-i[H,\cdot\,]+\mathcal L_{\mathrm c}(t)+\mathcal L_{\mathrm{noise}}$
is syndrome-based, and in its flip-rate operators the correction and the noise
contribute additively to the rate of the same flip,
\begin{equation}
    \Gamma_i^{(X)}(t)
    =
    \sum_{\alpha}f_{Z,i,\alpha}(t)^\dagger f_{Z,i,\alpha}(t)+\eta_i^{(X)}\mathbb I,
    \qquad
    \Gamma_i^{(Z)}(t)
    =
    \sum_{\alpha}f_{X,i,\alpha}(t)^\dagger f_{X,i,\alpha}(t)+\eta_i^{(Z)}\mathbb I.
\end{equation}

\begin{remark}[Time-dependent noise rates]
\label{rem:time-dependent-noise}
We take the noise rates to be constant only to lighten the notation.
Rates $\gamma_{P,i}(t)$ that are measurable and locally bounded in $t$ are
covered verbatim: Definition~\ref{def:syndrome-based-lindbladian} already
allows time-dependent syndrome functions, and Section~\ref{sec:periodic} uses the
rates only through the bound
$\physrate:=\sup_{t\ge0}\max_{i,P}\gamma_{P,i}(t)$.
\end{remark}

This completes the definition of the Lindbladian we will be using. We recall \cite{breuer_concepts_2003} that the Schr\"odinger-picture propagator $\mathcal W_*(t,t')$ generated by $\mathcal L(t)$ is defined by
\begin{equation}\label{eq:s-propagator}
    \frac{d}{dt}\mathcal W_*(t,t')
    =
    \mathcal L(t)\,\mathcal W_*(t,t'),
    \qquad
    \mathcal W_*(t',t')=\mathrm{id},
\end{equation}
so that $\rho(t)=\mathcal W_*(t,t')\rho(t')$.

Our analysis is carried out in the Heisenberg picture, which we now obtain by duality. For a superoperator $\mathcal K$ on $\mathcal B(\mathcal H)$, its adjoint $\mathcal K^*$ is defined by $\operatorname{Tr}\!\big(O\,\mathcal K(\rho)\big)=\operatorname{Tr}\!\big(\mathcal K^*(O)\,\rho\big)$ for all $\rho,O$. The adjoint of the dissipator is
\begin{equation}\label{eq:dissipator-adjoint}
    \mathcal D[L]^*(O)
    =
    L^\dagger O L-\frac12\{L^\dagger L,O\},
\end{equation}
and the Heisenberg-picture generator is
\begin{equation}\label{eq:lindbladian-heisenberg}
    \mathcal{L}^*(t)(O) = i[H(t), O]
    + \sum_{i=1}^n\sum_{P\in\{I,X,Y,Z\}}\sum_{\alpha}
    \mathcal D\bigl[P_i\,f_{P,i,\alpha}(t)\bigr]^*(O).
\end{equation}
Taking the adjoint of Eq.~\eqref{eq:s-propagator}, the Heisenberg-picture propagator $\mathcal W^*(t,t'):=\mathcal W_*(t,t')^*$ satisfies
\begin{equation}\label{eq:q-propgator}
    \frac{d}{dt}\mathcal W^*(t,t')
    =
    \mathcal W^*(t,t')\,\mathcal L^*(t),
    \qquad
    \mathcal W^*(t',t')=\mathrm{id},
\end{equation}
so that $\operatorname{Tr}\!\big(O\,\rho(t)\big)=\operatorname{Tr}\!\big(\mathcal W^*(t,t')(O)\,\rho(t')\big)$.

In Subsection~\ref{subsec:register-extension} we extend the definition of the syndrome-based Lindbladian to include auxiliary classical registers.

\begin{remark}[Thermal baths as a special case]\label{rem:thermal}
For a commuting stabilizer Hamiltonian, the standard Davies dynamics is
syndrome based. Suppose each qubit is coupled locally through $X_i$ and
$Z_i$ to independent thermal bath channels at inverse temperature $\beta$.
In the weak-coupling limit, the dissipative part of the Davies generator
has jump operators~\cite{alicki_thermal_2010}
\begin{equation}
    L_{\sigma,i}(\omega)
    =
    \sum_{\varepsilon'-\varepsilon=\omega}
    \Pi(\varepsilon)\sigma_i\Pi(\varepsilon'),
\end{equation}
with rates $g_\sigma(\omega)\ge0$ satisfying
\begin{equation}
    g_\sigma(-\omega)
    =
    e^{-\beta\omega}g_\sigma(\omega).
\end{equation}
Here $\omega$ is the energy change associated with the transition.
For a CSS stabilizer Hamiltonian, this energy change is determined by the
local syndrome. Hence
\begin{equation}
    L_{X,i}(\omega)
    =
    X_iP_{i,\omega}^{(Z)},
    \qquad
    L_{Z,i}(\omega)
    =
    Z_iP_{i,\omega}^{(X)},
\end{equation}
where
$P_{i,\omega}^{(Z)}\in\mathcal A_Z$ and
$P_{i,\omega}^{(X)}\in\mathcal A_X$
project onto the syndrome patterns with Bohr frequency $\omega$.
The Davies generator is therefore syndrome-based
(Definition~\ref{def:syndrome-based-lindbladian}), with syndrome functions
$f_{X,i,\omega}=\sqrt{g_X(\omega)}\,P_{i,\omega}^{(Z)}\in\mathcal A_Z$ and
$f_{Z,i,\omega}=\sqrt{g_Z(\omega)}\,P_{i,\omega}^{(X)}\in\mathcal A_X$, the Bohr
frequency $\omega$ playing the role of the auxiliary label $\alpha$. For independent
local bath channels the Lamb shift is syndrome diagonal and enters the Hamiltonian
$H\in\mathcal A$.
\end{remark}

\subsubsection{Induced continuous-time sector processes}
We now construct continuous-time Markov chains (CTMCs) on the sector configuration
spaces $\Omega_X$ and $\Omega_Z$, describing the evolution of the physical
sector strings.

\begin{definition}[Classical Markov sector processes]
\label{def:string-process}
We consider continuous time Markov processes (CTMCs) on the strings
in each sector separately.
The $X$-sector Markov process will be denoted $(\chi_t^{(X)})_{t \ge 0}$ and takes values $\chi_t^{(X)}\in\Omega_X$, and $(\chi_t^{(Z)})_{t \ge 0}$ is defined similarly for $Z$. Their transitions are the single-bit flips specified in Definition~\ref{def:ctmc}. The process is initialized at a codeword $\mathbf w^X_0\in\mathcal C_X$,
so that $\chi_0^{(X)}=\mathbf w^X_0$ and similarly $\chi_0^{(Z)}=\mathbf w^Z_0\in\mathcal C_Z$.
\end{definition}

In Proposition~\ref{prop:reduction-async-formal} we will prove the connection of these processes to the quantum evolution generated by the Lindbladian in Eq.~\eqref{eq:lindbladian}.
We follow the standard construction of continuous-time Markov chains from time-dependent transition rates, see e.g.~\cite[\S2.3--2.4]{yin_continuous-time_2013}. We start by defining the transition rates for these processes:

\begin{definition}[Classical Markov sector processes - transition rates]
\label{def:ctmc}
Let $\mathcal L(t)$ be a syndrome-based Lindbladian , with a chosen syndrome-based representation 
(Definition~\ref{def:syndrome-based-lindbladian}). The transitions of its induced
sector processes are the single-bit flips $\mathbf x\mapsto\mathbf x\oplus \hat{e}_i$, at
rates given by the syndrome eigenvalue functions
(Equation~\eqref{def:eigenvalue-functions}) of the flip-rate operators of
Equation~\eqref{eq:flip-rates}: for valid syndromes $s_X\in\Sigma_X$ and
$s_Z\in\Sigma_Z$,
\begin{equation}
    r_i^{(X)}(s_X,t):=\Gamma_i^{(X)}(t)(s_X),
    \qquad
    r_i^{(Z)}(s_Z,t):=\Gamma_i^{(Z)}(t)(s_Z).
\end{equation}
\end{definition}

The rates are well defined by the rate condition~\eqref{eq:rate-condition}, and
nonnegative since the flip-rate operators are positive semidefinite. Every jump that
flips qubit $i$ in a sector contributes additively to the rate of the single
transition $\mathbf x\mapsto\mathbf x\oplus \hat{e}_i$; for the sum of a correction
Lindbladian and Pauli noise,
$r_i^{(X)}(s_X,t)=\sum_\alpha\bigl(f_{Z,i,\alpha}(t)^\dagger f_{Z,i,\alpha}(t)\bigr)(s_X)+\eta_i^{(X)}$.

Using these rates, we can now define in a standard way the generators of the classical Markov sector processes. We will take here the analogue of the Heisenberg picture approach, and define the analogue of the Lindbladian in this picture 
 (see Equation \eqref{eq:lindbladian-heisenberg}). 
This turns out to be the generators for the {\it backward} propagation of {\it functions} $h_X:\Omega_X\to \mathbb{R}$ on the Markov processes, where  
these functions should be thought of as the classical analogue of quantum observables in the Heisenberg picture.  

\begin{definition}[Classical Markov sector processes - Generators and Propagators]\label{def:ctmc-generators}
Let $h_X : \Omega_X \to \mathbb{R}$ and $h_Z : \Omega_Z \to \mathbb{R}$ be 
functions of the ($X$ or $Z$) error vectors to the real numbers. We define   the infinitesimal backward generators $G_X(t)$ and $G_Z(t)$ for the time evolution of $h_X(t)$ and $h_Z(t)$ as follows 
(note that $G_X(t)h_X,G_Z(t)h_Z$ are also functions from the error vectors to scalars):  
\begin{align}\label{eq:ctmc-generator}
    (G_X(t)h_X)(\mathbf{e}) &:= \sum_{i=1}^n r_i^{(X)}(H_X\mathbf{e},t)\, \bigl(h_X(\mathbf{e}\oplus \hat{e}_i)-h_X(\mathbf{e})\bigr), \\
    (G_Z(t)h_Z)(\mathbf{e}) &:= \sum_{i=1}^n r_i^{(Z)}(H_Z\mathbf{e},t)\, \bigl(h_Z(\mathbf{e}\oplus \hat{e}_i)-h_Z(\mathbf{e})\bigr).
\end{align}

The backward transition operator $M^{(X)}(t,t')$ (the classical analogue of the Heisenberg picture propagator $\mathcal{W}^*(t,t')$ defined in Eq.~\eqref{eq:q-propgator}) is defined by
\begin{equation}\label{eq:classical-propagator}
    (M^{(X)}(t,t')h_X)(\mathbf{e})
    :=
    \mathbb E\!\left[
        h_X(\chi_t^{(X)})
        \mid
        \chi_{t'}^{(X)}=\mathbf{e}
    \right].
\end{equation}
It satisfies, for $t\ge t'$,
\begin{equation}
    \frac{d}{dt}M^{(X)}(t,t')=M^{(X)}(t,t')G_X(t),
    \qquad
    M^{(X)}(t',t')=\mathbb{I},
\end{equation}
which is Kolmogorov's forward equation for the transition matrix of a Markov chain with time-dependent generator, see~\cite[Thm.~2.5(c)]{yin_continuous-time_2013}. Similarly for $Z$.
\end{definition}

We can now formulate what is meant by storing a classical logical value.
By Definition~\ref{def:string-process}, the $X$-sector process starts from
a codeword $\mathbf w_0^X\in\mathcal C_X$. The initially stored value is
therefore
   $ b_X(\chi_0^{(X)})=\mu_X([\mathbf w_0^X]_X)$.
A readout at time $t$ returns
$b_X(\chi_t^{(X)})$, so retaining the stored logical value means that
$b_X(\chi_t^{(X)})=b_X(\chi_0^{(X)})$ with high probability, as formalized in
Definition~\ref{def:classical-memory}.

For the analysis it is convenient to measure the evolving string relative to
the initially stored codeword.

\begin{definition}[Sector error process]
\label{def:error-process}
The \emph{$X$-sector error process} is the difference between the initial codeword and the current string of the sector process
\begin{equation}\label{eq:error-process-def}
    \xi_t^{(X)}
    :=
    \chi_t^{(X)}\oplus\mathbf w^X_0 ,
    \qquad t\ge0 .
\end{equation}
Since $\chi_0^{(X)}=\mathbf w^X_0$, the error process starts at
$\xi_0^{(X)}=\mathbf 0$. In that context, we refer to a vector $\mathbf{e}_X\in\Omega_X$ as \emph{error vector}.
\end{definition}

Applying Eq.~\eqref{eq:observable-covariance} to
$\chi_t^{(X)}=\xi_t^{(X)}\oplus\mathbf w^X_0$ gives
\begin{equation}\label{eq:normalization}
    b_X(\chi_t^{(X)})\oplus b_X(\chi_0^{(X)})
    =
    b_X(\xi_t^{(X)})\oplus b_X(\xi_0^{(X)}),
\end{equation}
so the retention condition is equivalent to
$b_X(\xi_t^{(X)})=b_X(\xi_0^{(X)})=0$, whatever codeword was stored. 
\begin{definition}[logical-error indicator]
    $b_X$ evaluated along the error process, $b_X(\xi_t^{(X)})$, is called \emph{logical-error indicator}. It is evaluated to 1 when a logical flip has occurred.
\end{definition}

The error process was obtained from the string process by a fixed shift. It
remains to check that this shift leaves the stochastic dynamics unchanged,
namely that the error process is itself a sector process with the same
generator $G_X(t)$ of Definition~\ref{def:ctmc-generators}, now started at
$\mathbf 0$.

\begin{claim}[Codeword-translation invariance]
\label{clm:translation-invariance}
Let $\mathbf w\in\mathcal C_X$. Define the map \begin{equation}
    T_{\mathbf w}:\Omega_X\to\Omega_X,
    \qquad
    T_{\mathbf w}(\mathbf x):=\mathbf x\oplus\mathbf w .
\end{equation}
Then for every $h_X:\Omega_X\to\mathbb R$ and
every $t\ge0$,
\begin{equation}\label{eq:generator-intertwining}
    G_X(t)\bigl(h_X\circ T_{\mathbf w}\bigr)
    =
    \bigl(G_X(t)h_X\bigr)\circ T_{\mathbf w}.
\end{equation}
Where $G_X$ is the generator from Definition~\ref{def:ctmc-generators}. Consequently the error process of Eq.~\eqref{eq:error-process-def} is an
$X$-sector process with the same generator $G_X(t)$ and hence the same
backward transition operators $M^{(X)}(t,t')$, initialized at
$\xi_0^{(X)}=\mathbf 0$. In particular, its law is independent of the
initially stored codeword $\mathbf w^X_0$.
\end{claim}

The proof is deferred to Appendix~\ref{app:setup}.

The $Z$ sector is obtained by the exchange $X\leftrightarrow Z$, which
replaces $H_X$ by $H_Z$, $\mathcal C_X$ by $\mathcal C_Z$ and
$\operatorname{row}(H_Z)$ by $\operatorname{row}(H_X)$. It yields the logical
values $\mu_Z$ on $\mathcal C_Z/\operatorname{row}(H_X)$, a readout decoder
$\mathsf{Dec}_Z$, the decoded logical value $b_Z$, and the error process
$\xi_t^{(Z)}$.

\begin{example}[The 2D toric code]
For intuition, consider the $L\times L$ toric code as a classical code on edge strings and focus on the $X$ sector. The codewords $\mathcal C_X=\ker H_X$ are closed edge loops, modulo plaquette boundaries, and we may choose $\mu_X$ to be the mod-$2$ winding number in one fixed noncontractible direction. By linearity, the trivial logical class has value $0$. We choose
$\mu_X$ to detect winding in one fixed noncontractible direction, and take
$\mathbf w_0^X$ to be a noncontractible loop in that direction, so that
$\mu_X([\mathbf w_0^X]_X)=1$. Suppose, for example, that during the evolution a single edge of this loop is flipped from $1$ to $0$, creating a small gap in the loop. The resulting string $\chi_t^{(X)}$ has two syndrome defects at the endpoints of the missing edge. A natural readout decoder fills this gap by returning that edge, so that $\chi_t^{(X)}\oplus\mathsf{Dec}_X(H_X\chi_t^{(X)})=\mathbf w^X_0$ and hence $b_X(\chi_t^{(X)})=1$: the stored logical value is unchanged. In relative coordinates,

$$
    \xi_t^{(X)}=\chi_t^{(X)}\oplus\mathbf w^X_0,
    \qquad
    \xi_0^{(X)}=\mathbf 0,
$$
the same event is represented simply by the single flipped edge, and Eq.~\eqref{eq:observable-covariance} gives

$$
    b_X(\chi_t^{(X)})
    =
    b_X(\xi_t^{(X)})\oplus\mu_X([\mathbf w^X_0]_X)
    =
    b_X(\xi_t^{(X)})\oplus 1.
$$
Thus in this example $b_X(\xi_t^{(X)})=0$.
\end{example}

From this point on we work exclusively with the error processes
$(\xi_t^{(X)})_{t\ge0}$ and $(\xi_t^{(Z)})_{t\ge0}$, with generators $G_X(t)$
and $G_Z(t)$, propagators $M^{(X)}$ and $M^{(Z)}$, and initial conditions
$\xi_0^{(X)}=\xi_0^{(Z)}=\mathbf 0$; expectations over trajectories initiated
at $\mathbf 0$ are denoted by $\mathbb E_0$. The physical dynamics acts on the stored string and has no access to the
error vector, which is defined only relative to a codeword the system does
not know. Working with the error process nevertheless loses no physical
content. The syndrome map is linear and vanishes on codewords, so shifting by
the stored codeword changes no syndrome, and hence no transition rate and no
decoder output. By Claim~\ref{clm:translation-invariance}, the error process
is therefore distributed exactly as the physical string process started at the
codeword $\mathbf 0$, and by Eq.~\eqref{eq:normalization} the value a readout
returns is recovered from it by adding the stored value
$\mu_X([\mathbf w^X_0]_X)$.

For continuous-time dynamics, the physical channel at readout time \(t\)
is the Lindblad propagator
\begin{equation}
    \mathcal W_{*,t}
    :=
    \mathcal W_*(t,0).
\end{equation}
Accordingly, we use the shorthand
\begin{equation}
    \mathcal E_L(t)
    :=
    \mathcal E_L[\mathcal W_*(t,0)],
    \qquad
    F_{\min}(t)
    :=
    F_{\min}[\mathcal W_*(t,0)].
\end{equation}

\subsection{Synchronous (discrete-time) setup}
\label{subsec:sync-setup}

We now introduce the synchronous PCA-type discrete-time counterpart of the
continuous-time dynamics of Subsection~\ref{subsec:async-setup}.
The reduction below requires only that the probability of each physical
Pauli flip in a given CSS sector depend on the state through that
sector's syndrome (Definition~\ref{def:sync-syndrome-channel}).
Locality is a separate property, introduced in
Subsection~\ref{subsubsec:sync-locality} and used only when channels are constructed from SPCAs.

\subsubsection{Syndrome-based quantum channels}

\begin{definition}[Syndrome-based quantum channel]
\label{def:sync-syndrome-channel}
A CPTP map
    $\mathcal Q_*:
    \mathcal T(\mathcal H)
    \longrightarrow
    \mathcal T(\mathcal H)$
is \emph{syndrome-based} for the CSS code if it admits a Kraus
representation 
\begin{equation}\label{eq:sync-channel}
    \mathcal Q_*(\rho)
    =
    \sum_{\mathbf u^X,\mathbf u^Z,\alpha}
    K_{\mathbf u^X,\mathbf u^Z,\alpha}
    \rho
    K_{\mathbf u^X,\mathbf u^Z,\alpha}^\dagger ,
\end{equation}
where
\begin{equation}\label{eq:sync-kraus}
    K_{\mathbf u^X,\mathbf u^Z,\alpha}
    :=
    X^{\mathbf u^Z}Z^{\mathbf u^X}
    f_{\mathbf u^X,\mathbf u^Z,\alpha},
    \qquad
    f_{\mathbf u^X,\mathbf u^Z,\alpha}\in\mathcal A,
\end{equation}
with
\(
\mathbf u^X\in\Omega_X
\)
and
\(
\mathbf u^Z\in\Omega_Z
\).
The Kraus completeness relation is
    $\sum_{\mathbf u^X,\mathbf u^Z,\alpha}
    f_{\mathbf u^X,\mathbf u^Z,\alpha}^\dagger
    f_{\mathbf u^X,\mathbf u^Z,\alpha}
    =
    \mathbb I.$

For every sector flip, define the corresponding 
\emph{sector probability operators}
\begin{equation}
    \mathsf P_{\mathbf u^X}^{(X)}
    :=
    \sum_{\mathbf u^Z,\alpha}
    f_{\mathbf u^X,\mathbf u^Z,\alpha}^\dagger
    f_{\mathbf u^X,\mathbf u^Z,\alpha},\qquad  \mathsf P_{\mathbf u^Z}^{(Z)}
    :=
    \sum_{\mathbf u^X,\alpha}
    f_{\mathbf u^X,\mathbf u^Z,\alpha}^\dagger
    f_{\mathbf u^X,\mathbf u^Z,\alpha}. 
\end{equation}
We require the \emph{syndrome based condition} per sector
\begin{equation}\label{eq:sync-sector-condition}
    \mathsf P_{\mathbf u^X}^{(X)}
    \in\mathcal A_X,
    \qquad
    \mathsf P_{\mathbf u^Z}^{(Z)}
    \in\mathcal A_Z
\end{equation}
for every
\(\mathbf u^X,\mathbf u^Z\).
We denote the Heisenberg adjoint of \(\mathcal Q_*\) by
\(\mathcal Q^*\).
\end{definition}

The superscript on \(\mathbf u^\sigma\) labels the sector whose
configuration is flipped.  Thus the physical Pauli
\(Z^{\mathbf u^X}\) produces the \(X\)-sector flip
\(\mathbf u^X\), while
\(X^{\mathbf u^Z}\) produces the \(Z\)-sector flip
\(\mathbf u^Z\). 

By the Kraus completeness,
\begin{equation}
    \sum_{\mathbf u^X}
    \mathsf P_{\mathbf u^X}^{(X)}
    =
    \mathbb I,
    \qquad
    \sum_{\mathbf u^Z}
    \mathsf P_{\mathbf u^Z}^{(Z)}
    =
    \mathbb I.
\end{equation}
Since the sector probability operators are positive and belong to the
corresponding stabilizer algebras, their syndrome eigenvalue functions
define probability distributions.

\subsubsection{Induced discrete-time sector processes}

\begin{definition}[Induced discrete sector processes]
\label{def:sync-sector-process}
For
\(s_X\in\Sigma_X\) and \(s_Z\in\Sigma_Z\), define
\begin{align}
    p_X(\mathbf u^X\mid s_X)
    &:=
    \mathsf P_{\mathbf u^X}^{(X)}(s_X),
    \\
    p_Z(\mathbf u^Z\mid s_Z)
    &:=
    \mathsf P_{\mathbf u^Z}^{(Z)}(s_Z).
\end{align}

The induced \(X\)-sector process
\((\chi_r^{(X)})_{r\ge0}\) is the discrete-time Markov chain on
\(\Omega_X\) with transition probabilities
\begin{equation}\label{eq:sync-sector-transition-X}
    \Pr\!\left[
        \chi_{r+1}^{(X)}
        =
        \mathbf e\oplus\mathbf u^X
        \,\middle|\,
        \chi_r^{(X)}=\mathbf e
    \right]
    =
    p_X(
        \mathbf u^X
        \mid
        H_X\mathbf e
    ).
\end{equation}
The \(Z\)-sector process is defined analogously.

The corresponding one-step backward transition operators are
\begin{align} 
    (\mathsf B_X h)(\mathbf e)
    &:=
    \sum_{\mathbf u^X\in\Omega_X}
    p_X(
        \mathbf u^X
        \mid
        H_X\mathbf e
    )
    h(\mathbf e\oplus\mathbf u^X),
    \label{eq:sync-backward-X}
    \\
    (\mathsf B_Z h)(\mathbf e)
    &:=
    \sum_{\mathbf u^Z\in\Omega_Z}
    p_Z(
        \mathbf u^Z
        \mid
        H_Z\mathbf e
    )
    h(\mathbf e\oplus\mathbf u^Z).
    \label{eq:sync-backward-Z}
\end{align}
Thus
\begin{equation}\label{eq:sync-backward-r}
    (\mathsf B_\sigma^r h)(\mathbf e)
    =
    \mathbb E\!\left[
        h(\chi_r^{(\sigma)})
        \,\middle|\,
        \chi_0^{(\sigma)}=\mathbf e
    \right].
\end{equation}
\end{definition}

As in the continuous-time setting, let
\(\mathbf w_0^\sigma\in\mathcal C_\sigma\) be the initially stored
sector codeword and define the corresponding error process by
\begin{equation}\label{eq:sync-error-process}
    \xi_r^{(\sigma)}
    :=
    \chi_r^{(\sigma)}
    \oplus
    \mathbf w_0^\sigma .
\end{equation}
Then
\(
\xi_0^{(\sigma)}=\mathbf0
\).
Moreover,
\(H_\sigma(\mathbf e\oplus\mathbf w_0^\sigma)=H_\sigma\mathbf e\),
so the transition kernel is unchanged by this translation.  Hence
\((\xi_r^{(\sigma)})_{r\ge0}\) is distributed as the same sector
Markov chain, exactly as in the
continuous-time translation argument of
Claim~\ref{clm:translation-invariance}.

For synchronous quantum dynamics, the physical evolution after
\(r\) rounds is
\begin{equation}
    \mathcal W_{*,r}^{\mathrm{sync}}
    :=
    \mathcal Q_*^r.
\end{equation}
Accordingly, we use the shorthand
\begin{equation}
    \mathcal E_L(r)
    :=
    \mathcal E_L[\mathcal Q_*^r],
    \qquad
    F_{\min}(r)
    :=
    F_{\min}[\mathcal Q_*^r].
\end{equation}

\subsubsection{Synchronous Pauli noise}

\begin{definition}[Synchronous Pauli noise layer]
\label{def:sync-pauli-noise}
A \emph{synchronous Pauli noise layer} is a product channel
\(\mathcal N_*=\mathcal N_{1*}\otimes\cdots\otimes\mathcal N_{n*}\) with
\begin{equation}
    \mathcal N_{i*}(\rho)
    =
    \bigl(1-\theta_{X,i}-\theta_{Y,i}-\theta_{Z,i}\bigr)\rho
    +\theta_{X,i}\,X_i\rho X_i
    +\theta_{Y,i}\,Y_i\rho Y_i
    +\theta_{Z,i}\,Z_i\rho Z_i ,
\end{equation}
where \(\theta_{P,i}\ge0\) and
\(\theta_{X,i}+\theta_{Y,i}+\theta_{Z,i}\le1\). Its
\emph{sector flip probabilities} are
\begin{equation}
    \theta_i^{(X)}:=\theta_{Z,i}+\theta_{Y,i},
    \qquad
    \theta_i^{(Z)}:=\theta_{X,i}+\theta_{Y,i}.
\end{equation}
We say that the noise layer has \emph{rate} $\eta$ if
$\theta_{P,i}\le\eta$ for every qubit $i$ and every $P\in\{X,Y,Z\}$.
\end{definition}

Writing \(Y_i=\mathrm iX_iZ_i\), the Kraus operators of
\(\mathcal N_{i*}\) are scalar multiples of \(\mathbb I\), \(Z_i\),
\(X_i\) and \(X_iZ_i\), so \(\mathcal N_*\) is syndrome-based with
scalar syndrome functions. Its induced \(X\)-sector process flips each
bit \(i\) independently with probability \(\theta_i^{(X)}\),
irrespective of the syndrome, and similarly for the \(Z\) sector.

\subsubsection{Locality}
\label{subsubsec:sync-locality}

In a local Lindbladian, each jump acts near one site with a rate read
from its neighborhood, so that in an infinitesimal time step a site is
influenced only by its neighborhood. The synchronous counterpart is a
condition on one complete round. 

\begin{definition}[Local quantum channel]
\label{def:local-sync-channel}
A CPTP map
\(\mathcal Q_*:\mathcal T(\mathcal H)\to\mathcal T(\mathcal H)\) is
\emph{\(R\)-local} if for every set of qubits $A$ there is a CPTP map
\(
    \mathcal Q_{A*}:
    \mathcal T(\mathcal H_{N_R(A)})\to\mathcal T(\mathcal H_A)
\)
such that, for every state \(\rho\),
\begin{equation}\label{eq:sync-locality-S}
    \operatorname{Tr}_{\bar A}\,\mathcal Q_*(\rho)
    =
    \mathcal Q_{A*}\!\left(
        \operatorname{Tr}_{\overline{N_R(A)}}\,\rho
    \right).
\end{equation}
\end{definition}

Thus the state of any set of qubits \(A\) after one round depends only
on the input state in the bounded neighborhood \(N_R(A)\).  Equivalently,
information can propagate by at most distance \(R\) in one synchronous
round.  This is the finite-propagation causality condition familiar
from quantum cellular automata, here imposed on a general CPTP map
without assuming reversibility, unitarity, or translation invariance.

\begin{claim}[Heisenberg form of locality]
\label{clm:sync-locality-H}
A CPTP map \(\mathcal Q_*\) is \(R\)-local if and only if, for every
set of qubits $A$ and every operator \(O\) supported on \(A\), the operator
\(\mathcal Q^*(O)\) is supported on \(N_R(A)\).
\end{claim}
Proof deferred to Appendix~\ref{app:setup}

If \(\mathcal Q_{1*}\) is \(R_1\)-local and \(\mathcal Q_{2*}\) is
\(R_2\)-local, then \(\mathcal Q_{2*}\circ\mathcal Q_{1*}\) is
\((R_1+R_2)\)-local, since
\(N_{R_1}(N_{R_2}(A))\subseteq N_{R_1+R_2}(A)\). A synchronous Pauli
noise layer is \(0\)-local. Locality and the syndrome-based condition
of Definition~\ref{def:sync-syndrome-channel} are independent: the
reduction of Section~\ref{sec:sync-reduction} uses only the latter.

\subsubsection{Syndrome-based probabilistic cellular automata}

Lastly we define the related classical object, the syndrome-based probabilistic cellular automaton.

\begin{definition}[Syndrome-based Probabilistic Cellular Automaton (SPCA)]
\label{def:syndrome-pca-classical}
Let $C=\ker H\subseteq\Omega:=\mathbb F_2^n$ be a classical linear
code with parity-check matrix $H$, and equip
measure distances in the interaction graph of $C$, in which two bits are adjacent if
some parity check acts on both. An
$R$-\emph{local syndrome-based PCA (SPCA) for $C$ at
stochastic-adversarial fault rate $p$} is a random map of the form
\begin{equation}
\Psi:\Omega\to\Omega,\qquad
\mathbf e_{r+1}=\Psi(\mathbf e_r)
=\mathbf e_r\oplus\mathbf a(H\mathbf e_r)\oplus\zeta_r
\end{equation}
where $\mathbf a=(a_1,\dots,a_n)$ is a vector of Boolean functions
$a_i:\operatorname{Im}H\to\mathbb F_2$, and $\zeta_r\in\Omega$ is a
random fault vector. The following conditions are imposed:
\begin{enumerate}
    \item \textbf{Syndrome-based update:} Each $a_i$ depends only on
    syndrome bits $s_j$ whose corresponding parity check has support
    within graph distance $R$ of bit $i$.
        \item \textbf{Stochastic-adversarial noise:} At each round $r$, a
    random fault-location set $F_r\subseteq\{1,\dots,n\}$ is drawn by
    including each site independently with probability $p$,
    independently of $\mathbf e_0,\dots,\mathbf e_r$ and of
    $F_0,\dots,F_{r-1}$. The fault vector $\zeta_r$ is otherwise
    unrestricted: it may be random and may depend on the entire history
    and on $F_r$, subject only to
    \begin{equation}
        \operatorname{supp}(\zeta_r)\subseteq F_r.
    \end{equation}
    Thus the deterministic rule is followed exactly away from faulty
    sites, while arbitrary deviations from it are allowed on $F_r$.
\end{enumerate}
A \emph{run} of $\Psi$ is a process $(\mathbf e_r)_{r\ge0}$ generated by this map for some choice of fault vectors $\zeta_r$ satisfying the conditions above. For $b:\Omega\to S$, we say that $\Psi$, started at $\mathbf e_0$, \emph{$q$-remembers $b$ through round $m$} if every run of $\Psi$ with initial configuration $\mathbf e_0$ $q$-remembers $b$ through round $m$ in the sense of Definition~\ref{def:classical-memory}.
\end{definition}

\section{Main results}
\label{sec:results}
Our results have three layers.

First, we establish an exact quantum--classical reduction for
continuous-time syndrome-based Lindbladian dynamics.  Each CSS sector
induces a classical continuous-time Markov process, and the decoded
quantum logical autocorrelations are exactly the logical-retention
correlations of these sector processes.  This gives 
Proposition~\ref{prop:reduction-async-formal} and the corresponding
classical--quantum memory equivalence, 
Theorem~\ref{thm:main-async-formal}.

Second, we prove the PCA-type analogue.  A
syndrome-based quantum channel induces discrete-time classical sector
processes, and the same algebraic reduction gives 
Proposition~\ref{prop:reduction-sync-formal} and 
Theorem~\ref{thm:main-sync}.  In this setting, local
syndrome-based probabilistic cellular automata can be lifted directly
to finite-range syndrome-based quantum channels.  This yields
Corollary~\ref{cor:sync-spca-memory}, which we apply to the 4D Toom
rule in Theorem~\ref{thm:toom-sync-formal}.

Finally, we bridge the synchronous classical construction back to
continuous time.  Theorem~\ref{thm:periodic-formal} shows that
any pair of local SPCAs with a nonzero memory threshold can be
implemented robustly by a geometrically local Floquet Lindbladian.
The additional work is to show that the continuous-time sector
processes, sampled at correction-cycle boundaries, remain admissible
trajectories of the original SPCAs.

\subsection{Asynchronous quantum--classical equivalence}
\label{subsec:async-results}

\begin{formalas}{prop:reduction-async-informal}
\begin{mainproposition}[Asynchronous quantum-to-classical reduction]
\label{prop:reduction-async-formal}
Let $\mathcal L(t)$ be a syndrome-based Lindbladian for a CSS code,
with Heisenberg-picture propagator $\mathcal W^*(t,0)$.  Let
$T_X,T_Z$ be the decoded logical observables associated with a
designated logical qubit, let $b_X,b_Z$ be the corresponding
logical-error indicators, and let
$(\xi_t^{(X)})_{t\ge0}$ and $(\xi_t^{(Z)})_{t\ge0}$ be the induced
sector error processes.

Then, for every $\sigma\in\{X,Z\}$ and every $t\ge0$,
\begin{equation}
    \left\langle
        T_\sigma,\mathcal W^*(t,0)(T_\sigma)
    \right\rangle_{\mathrm{code}}
    =
    \mathbb E\!\left[
        (-1)^{
            b_\sigma(\xi_t^{(\sigma)})
            \oplus
            b_\sigma(\xi_0^{(\sigma)})
        }
    \right]
    =
    1-2
    \Pr\!\left[
        b_\sigma(\xi_t^{(\sigma)})
        \ne
        b_\sigma(\xi_0^{(\sigma)})
    \right].
\end{equation}
\end{mainproposition}
\end{formalas}

Since $\xi_0^{(X)}=\mathbf 0$ and $b_X(\mathbf 0)=0$, this equals
$1-2\Pr\bigl[b_X(\xi_t^{(X)})=1\bigr]$, and similarly for $Z$.

The proof of Proposition~\ref{prop:reduction-async-formal} is given in
Section~\ref{sec:async}. The proposition directly identifies each
quantum decoded logical observable autocorrelation with the classical probability
of an incorrect classical logical readout at the same time. Thus controlling the
classical logical-error indicators at every chosen readout time controls the
corresponding quantum logical correlations pointwise in time.

Proposition~\ref{prop:reduction-async-formal} therefore couples the quantum
autocorrelation to the classical label-retention problem of
Definition~\ref{def:classical-memory}, which we now use to lower
bound the worst-case fidelity $F_{\mathrm{min}}(t)$ and establish
quantum memory:

\begin{formalas}{thm:main-async-informal}
\begin{maintheorem}[Fully Autonomous quantum--classical memory equivalence]
\label{thm:main-async-formal}
Consider a CSS code with a designated logical qubit evolving under a
syndrome-based Lindbladian \(\mathcal L(t)\), and let
\((\xi_t^{(X)})_{t\ge0}\) and \((\xi_t^{(Z)})_{t\ge0}\) be its induced
sector error processes.

If the \(X\)- and \(Z\)-sector processes respectively \(q_X\)- and
\(q_Z\)-remember their logical-error indicators \(b_X\) and \(b_Z\) up to
time \(t\), then the encoded system is a quantum
\((q_X+q_Z)\)-memory up to time \(t\).

Conversely, if the encoded system is a quantum
\(\memerr\)-memory up to time \(t\), then each induced sector
process \(\memerr\)-remembers its corresponding logical-error
indicator up to time \(t\).
\end{maintheorem}
\end{formalas}

Thus, for syndrome-based Lindbladian dynamics, quantum memory of the
designated logical qubit is equivalent, up to a constant factor, to
memory of the two induced classical sector processes.  
Proposition~\ref{prop:reduction-async-formal} is proved in
Section~\ref{sec:async}, and 
Theorem~\ref{thm:main-async-formal} in Subsection~\ref{subsec:memory}.

\subsection{Synchronous quantum--classical equivalence}
\label{subsec:sync-results}

The same reduction has a synchronous discrete-time counterpart.  One
round is now a syndrome-based quantum channel $\mathcal Q_*$, and its
sector marginals define discrete-time Markov chains on the two error
configuration spaces.  The continuous-time generator intertwining is
replaced by a one-step intertwining between $\mathcal Q^*$ and the
backward transition operators of these chains.

\begin{formalas}{prop:reduction-sync-informal}
\begin{mainproposition}[Quantum-to-classical reduction for syndrome-based channels]
\label{prop:reduction-sync-formal}
Let \(\mathcal Q_*\) be a syndrome-based quantum channel for
a CSS code (Definition~\ref{def:sync-syndrome-channel}), let \(T_X,T_Z\)
be the decoded logical observables, \(b_X,b_Z\) the logical-error
indicators, and \((\xi_r^{(X)})_{r\ge0}\), \((\xi_r^{(Z)})_{r\ge0}\) the
induced sector error processes. Then, for every \(\sigma\in\{X,Z\}\) and every integer \(r\ge0\),
\begin{equation}
    \left\langle
        T_\sigma,
        (\mathcal Q^*)^r(T_\sigma)
    \right\rangle_{\mathrm{code}}
    =
    \mathbb E\!\left[
        (-1)^{
            b_\sigma(\xi_r^{(\sigma)})
            \oplus
            b_\sigma(\xi_0^{(\sigma)})
        }
    \right]
    =
    1-2
    \Pr\!\left[
        b_\sigma(\xi_r^{(\sigma)})
        \ne
        b_\sigma(\xi_0^{(\sigma)})
    \right].
\end{equation}
\end{mainproposition}
\end{formalas}

Proposition~\ref{prop:reduction-sync-formal} is the exact
discrete-time analogue of 
Proposition~\ref{prop:reduction-async-formal}: one synchronous quantum round
is represented by one step of the induced classical Markov chain, and
iteration gives the identity at every round.

\begin{maintheorem}[Quantum--classical memory equivalence for syndrome-based channels]
\label{thm:main-sync}
Consider a CSS code with a designated logical qubit evolving under a syndrome-based quantum channel $\mathcal Q_*$, and let
$(\xi_r^{(X)})_{r\ge0}$ and $(\xi_r^{(Z)})_{r\ge0}$ be its induced
sector error processes.

If the $X$- and $Z$-sector processes respectively $q_X$- and
$q_Z$-remember their logical-error indicators $b_X$ and $b_Z$ through round $m$, then the encoded
system is a quantum $(q_X+q_Z)$-memory through round $m$.

Conversely, if the encoded system is a quantum \(\memerr\)-memory through
round \(m\), then each induced sector process
\(\memerr\)-remembers its corresponding logical-error indicator through
round \(m\).
\end{maintheorem}

Both synchronous statements are proved in
Section~\ref{sec:sync-reduction}.

\paragraph{Local PCA-type construction from SPCAs.}
The synchronous correspondence becomes constructive when the
classical sector processes are chosen to be local SPCAs.  Given an
SPCA rule in each sector, its quantum lift defines a synchronous
syndrome-based channel whose causal range is bounded by the locality
radius of the classical rules.  After a local Pauli-noise layer, the
induced sector processes are themselves admissible runs of the input
SPCAs.  Theorem~\ref{thm:main-sync} therefore immediately
transfers their classical memory to the quantum system.

\begin{maincorollary}[Quantum memory from local SPCAs]
\label{cor:sync-spca-memory}
Let $\mathcal{C}_X,\mathcal{C}_Z$ be a pair of CSS-compatible classical
codes with bounded-degree interaction graphs and linear size $L$. Fix a
designated logical qubit and readout decoders, with logical-error
indicators $b_X,b_Z$.

Suppose that, for each sector $\sigma\in\{X,Z\}$, the component code
$\mathcal C_\sigma$ admits an $R_\sigma$-local SPCA $\Psi_\sigma$, with
update rule $\mathbf a^{(\sigma)}$, where $R_\sigma$ is independent of
$L$. Assume that there exists $\pcth>0$, independent of $L$, such that
for every $p<\pcth$, $\Psi_\sigma$ at fault rate $p$, started at $\mathbf 0$,
$q_\sigma(m,L)$-remembers $b_\sigma$ through round $m$.

Let $\mathcal N_*$ be a synchronous Pauli noise layer of rate $\eta<\pcth/2$
(Definition~\ref{def:sync-pauli-noise}). Then there exists a syndrome-based
 quantum channel $\mathcal Q_*$, $R$-local for some $R$ independent
of $L$, such that the induced sector processes under $\mathcal Q_*$ are runs of
$\Psi_X$ and $\Psi_Z$ at fault rate $2\eta$, and the encoded system evolving
under $\mathcal Q_*$ is a quantum $q(m,L)$-memory through round $m$, with
$q(m,L)=q_X(m,L)+q_Z(m,L)$. The channel is the quantum lift of the update rules
$\mathbf a^{(X)},\mathbf a^{(Z)}$ followed by $\mathcal N_*$
(Definition~\ref{def:spca-lift}), and is translationally invariant whenever the
code pair, both SPCAs and $\mathcal N_*$ are.
\end{maincorollary}

The proof is given in
Subsection~\ref{subsec:sync-spca-realization}. Combining Theorem~\ref{thm:main-sync} with this Corollary, gives the informal Theorem~\ref{thm:main-spca-sync-informal} in Section~\ref{sec:intro}.

\paragraph{Application to the 4D Toom rule.}
For the 4D toric code, the local Toom rules provide precisely the
classical SPCAs required by Corollary~\ref{cor:sync-spca-memory}.
Section~\ref{sec:toom} verifies the SPCA structure and imports Ahn's
classical threshold estimate, which gives
\[
    q_X(m,L)=q_Z(m,L)=Cme^{-cL}.
\]
The synchronous quantum-memory result is therefore an immediate
application of the preceding corollary.

\begin{formalas}{thm:toom-sync-informal}
\begin{maintheorem}[Synchronous 4D Toom quantum memory]
\label{thm:toom-sync-formal}
For the 4D toric code of linear size $L$, let $\mathcal N_*$ be a synchronous
Pauli noise layer of rate $\eta$ (Definition~\ref{def:sync-pauli-noise}). There
exist constants $\eta_c,C,c>0$, independent of $L$, and a syndrome-based quantum channel $\mathcal Q_{\mathrm{Toom}*}$, $R$-local for some
$R$ independent of $L$, such that if $\eta<\eta_c$, the encoded system evolving
under $\mathcal Q_{\mathrm{Toom}*}$ is a quantum $q(m,L)$-memory through round
$m$, with $q(m,L)=2Cme^{-cL}$. The channel is the quantum lift of the two 4D
Toom update rules followed by $\mathcal N_*$ (Definition~\ref{def:spca-lift}),
and is translationally invariant whenever $\mathcal N_*$ is.
\end{maintheorem}
\end{formalas}

Hence, for every fixed target error $\epsilon>0$, the number of
synchronous rounds for which the encoded system remains a quantum
$\epsilon$-memory grows exponentially with $L$.
The proof, in Section~\ref{sec:toom}, verifies that the 4D Toom rule of
Ahn~\cite[\S 7]{ahn_extending_2004} is a valid SPCA satisfying the
hypotheses of Corollary~\ref{cor:sync-spca-memory}, importing
Ahn's classical threshold result.

\subsection{Bridging PCA-type and Fully autonomous dynamics}
\label{subsec:bridge-results}

The preceding corollaries assume an ideal synchronous quantum update.
We now return to continuous-time local dynamics and ask whether the
same classical SPCAs can be used without assuming that their entire
update is executed as a single quantum channel.

This is not automatic. Local Lindblad jumps occur at stochastic
times. If each correction were conditioned directly on the current
syndrome, a correction occurring early in a round could change the
syndrome seen by later corrections, so the resulting dynamics would
not implement the synchronous SPCA rule.

Our solution is a periodically driven local Lindbladian with local
decision and execution registers. Each cycle is divided into
Reset, Read, and Execute stages. The Read stage records all local
decisions from a common syndrome configuration, and the Execute stage
subsequently applies those stored decisions. Missed operations and
physical noise generate additional local faults. The main content of
the theorem is that these faults can still be dominated by the
stochastic-adversarial SPCA fault model.

\begin{formalas}{thm:periodic-informal}
\begin{maintheorem}[Floquet autonomous quantum memory from SPCAs]
\label{thm:periodic-formal}
Let $\mathcal C_X,\mathcal C_Z$ be a pair of CSS-compatible classical codes
with bounded-degree interaction graphs and linear size $L$. Fix a designated
logical qubit and readout decoders, with logical-error indicators $b_X,b_Z$.

Suppose that, for each sector $\sigma\in\{X,Z\}$, the component code
$\mathcal C_\sigma$ admits an $R_\sigma$-local SPCA $\Psi_\sigma$, with update
rule $\mathbf a^{(\sigma)}$, where $R_\sigma$ is independent of $L$. Assume that
there exists $\pcth>0$, independent of $L$, such that for every $p<\pcth$,
$\Psi_\sigma$ at fault rate $p$, started at $\mathbf 0$, $q_\sigma(m,L)$-remembers
$b_\sigma$ through round $m$.

Then there exist a threshold $\eta_c>0$, independent of $L$, and a geometrically
local periodic Lindbladian $\mathcal L_{\mathrm{PCA}}(t)$ with period $\Delta$
independent of $L$, acting on the code qubits together with local classical
registers, such that under local Pauli noise of rate $\eta<\eta_c$ on the code
qubits and the registers, the encoded system evolving under
$\mathcal L_{\mathrm{PCA}}(t)$ is a quantum $q(m,L)$-memory at cycle boundaries
through cycle $m$, with $q(m,L)=q_X(m,L)+q_Z(m,L)$. The Lindbladian
$\mathcal L_{\mathrm{PCA}}(t)$ is translationally invariant whenever the code
pair and both SPCAs are.
\end{maintheorem}
\end{formalas}

The proof is given in Section~\ref{sec:periodic}.

\section{Asynchronous quantum-to-classical reduction}
\label{sec:async}

In this section we prove Proposition~\ref{prop:reduction-async-formal}.  The basic idea is that, although the full Heisenberg dynamics acts on a very large operator space, the evolution of a decoded logical observable is much more constrained.

Consider the decoded \(X\)-logical observable.  It consists of the bare logical operator \(\bar X\), together with a coefficient that depends on the measured \(X\)-syndrome.  The special form of a syndrome-based Lindbladian preserves this structure.  Its jumps may change the syndrome and may change the sign with which \(\bar X\) appears, but they do not turn \(\bar X\) into a different logical operator.  Starting from the decoded logical observable, the Heisenberg evolution therefore remains inside a fixed logical-sector subspace \(\mathcal V_X\).  The same statement holds for the \(Z\) sector.

This observation reduces the problem considerably.  Once the logical part is fixed, the only nontrivial information left in an operator in \(\mathcal V_X\) is how its coefficient depends on the error configuration. We can therefore encode every such operator by a function on the classical space of \(X\)-error configurations. As the quantum observable evolves inside \(\mathcal V_X\), this function evolves as well.  The central point of the reduction is that its evolution is not merely analogous to a classical process: it is exactly the backward evolution generated by the induced classical Markov chain.  In this sense, the classical process keeps track of the entire Heisenberg evolution relevant to the logical observable.

The correspondence can be summarized by the commutative picture
\[
\begin{array}{ccc}
\text{decoded logical observable}
&\xrightarrow{\text{quantum Heisenberg evolution}}&
\text{evolved logical-sector observable}
\\[1mm]
\downarrow\text{ classical representation}
&&
\downarrow\text{ classical representation}
\\[1mm]
\text{function on error configurations}
&\xrightarrow{\text{classical Markov evolution}}&
\text{evolved function}.
\end{array}
\]

For the decoded logical observable, the classical representative has a direct coding-theoretic meaning.  The associated classical decoder maps each error configuration to one of the two logical classes of the component code, which we encode as \(+1\) or \(-1\).  The induced Markov process therefore tracks this classical logical value, and the intertwining identifies its retention with the quantum decoded-logical autocorrelation.

The remainder of the section makes this picture precise. First, we identify the algebraic form of the decoded logical observables and the
corresponding logical sector subspaces, and show that the full Heisenberg
generator preserves these subspaces independently
(Subsections~\ref{subsec:algebraic-structure}
and~\ref{subsec:sector-decoupling}).

Second, we construct representation maps
$\Phi_X,\Phi_Z$ from the invariant quantum subspaces to scalar functions on
the classical error spaces and prove that these maps intertwine the
Heisenberg evolution with the classical Markov propagators
(Subsection~\ref{subsec:classical-repr}). In this step the classical image
of each decoded logical observable is identified with the corresponding
binary logical-error indicator.

Finally, we show that the code-space inner product is mapped to evaluation
at the vacuum error vector. Combining these facts yields the
classical expression for the decoded-logical autocorrelations
(Subsections~\ref{subsec:expectation-values}
and~\ref{subsec:main-reduction}).

\subsection{Algebraic structure of the logical sector subspaces}\label{subsec:algebraic-structure}

The first step is to expose the algebraic structure of the decoded
observables. Since each readout correction is Pauli, its conjugation of the
corresponding bare logical operator contributes only a phase. We first
isolate this syndrome-dependent phase and then use it to factor the decoded
observables.

\begin{definition}[Readout decoder phase]
\label{def:decoder-signs}
For valid syndromes $s_X\in\Sigma_X$ and $s_Z\in\Sigma_Z$, define
$\varphi_X(s_X),\varphi_Z(s_Z)\in\{\pm1\}$ by
\begin{equation}
    R_X(s_X)^\dagger \bar X R_X(s_X)
    =
    \varphi_X(s_X)\bar X,
    \qquad
    R_Z(s_Z)^\dagger \bar Z R_Z(s_Z)
    =
    \varphi_Z(s_Z)\bar Z.
\end{equation}
The corresponding syndrome-diagonal operators are
\begin{equation}\label{eq:syndrome-diagonal-operators}
    \varphi_X
    :=
    \sum_{s_X\in\Sigma_X}\varphi_X(s_X)\Pi_{s_X}^X,
    \qquad
    \varphi_Z
    :=
    \sum_{s_Z\in\Sigma_Z}\varphi_Z(s_Z)\Pi_{s_Z}^Z,
\end{equation}
with $\Pi_{s_X}^X$ the syndrome projectors defined in Definition~\ref{def:syndrome-projectors}
\end{definition}

As in Definition~\ref{def:syndrome-projectors},
$\varphi_X(s_X)$ and $\varphi_Z(s_Z)$ are the syndrome eigenvalue functions
of the corresponding operators. The vacuum normalization of the readout
decoder gives $\varphi_X(\mathbf0)=\varphi_Z(\mathbf0)=1$.
Plugging the decoder from Definitions~\ref{def:pauli-readout-corrections}
and~\ref{def:decoder-signs}, together with the Pauli commutation relations gives the form
\begin{equation}\label{eq:decoder-phase}
    \varphi_X(s)
    =
    (-1)^{
        \ell_X\cdot\mathsf{Dec}_X(s)
    }.
\end{equation}

The decoded logical observables therefore have a particularly simple algebraic form.

\begin{lemma}[Algebraic reduction of decoded logical observables]
\label{lem:algebraic-form}
The decoded logical observables (Definition~\ref{def:decoded-observables}) reduce to the algebraic form 
\begin{equation}
    T_X=\bar X\varphi_X,
    \qquad
    T_Z=\bar Z\varphi_Z.
\end{equation}
where $\varphi_X \in \mathcal{A}_X$ and $\varphi_Z \in \mathcal{A}_Z$ are the diagonal syndrome operators  defined in Definition~\ref{def:decoder-signs}.
\end{lemma}

\begin{proof}[Proof idea]
 The decoded logical observable $T_X$ applies the bare logical conjugated by a syndrome conditioned Pauli correction. The bare logical $\bar X$, being $X$-type, commutes with every element of $\mathcal{A}_X$, specifically with the syndrome projectors, and thus factors out of the syndrome sum, leaving $\varphi_X$. Proof deferred to Appendix~\ref{app:reduction}.
\end{proof}

To formalize the sector decoupling, we introduce the \emph{logical sector subspaces} 
underlying the reduction. Each subspace consists of all operators of the 
form $\bar{X}A$, where $A$ ranges over the corresponding 
stabilizer algebra. The decoded logical observables $T_X$ and $T_Z$ are 
particular elements of these spaces, and the key structural fact,
established in the next subsection, is that these subspaces are invariant 
under the full Heisenberg Lindbladian dynamics.

\begin{definition}[Logical sector subspaces]
\label{def:sector-subspaces}
The logical sector subspaces  $\mathcal{V}_X, \mathcal{V}_Z \subset \mathcal{B}(\mathcal{H})$ are the vector spaces spanned by observables of the following form:
\begin{equation}
    \mathcal{V}_X := \mathrm{span} \{ \bar{X} A \mid A \in \mathcal{A}_X \}, 
    \qquad 
    \mathcal{V}_Z := \mathrm{span} \{ \bar{Z} B \mid B \in \mathcal{A}_Z \}.
\end{equation}
By Lemma~\ref{lem:algebraic-form}, the decoded logical observables belong to these subspaces.
\end{definition}
We note that $\bar{X}$ is not in $\mathcal{A}_X$, and $\bar{Z}$ not in $\mathcal{A}_Z$.

An observable $O_X=\bar X A\in\mathcal V_X$ consists of the fixed bare
logical observable $\bar X$, multiplied by a syndrome-dependent
coefficient $A\in\mathcal A_X$.  Indeed, since $A$ is diagonal in the
$X$-syndrome decomposition, on the syndrome sector $s_X$ one has
\[
    O_X|_{s_X}=A(s_X)\bar X.
\]
Thus all observables in $\mathcal V_X$ have the same logical operator
structure; only their syndrome-dependent coefficient varies.  The same
interpretation applies to $\mathcal V_Z$.

\subsection{Closed sector dynamics under the Lindbladian}
\label{subsec:sector-decoupling}
With the logical sector subspaces identified, we show that the Heisenberg generator
of a syndrome-based Lindbladian maps each subspace into itself. The Hamiltonian
term $i[H(t),\cdot\,]$ vanishes identically on $\mathcal V_X$ and $\mathcal V_Z$,
since every element of $\mathcal A$ commutes with both subspaces
(Lemma~\ref{lem:hamiltonian-invariance}, Appendix~\ref{app:reduction}). For the
dissipative part, the key fact is that a single jump acts on a sector subspace only
through conjugation by its Pauli factor, weighted by its rate operator.

\begin{lemma}[Action of a single jump on the sector subspaces]
\label{lem:jump-action}
Let $i\in\{1,\dots,n\}$, $P\in\{I,X,Y,Z\}$ and $f\in\mathcal A$. For every
$O\in\mathcal V_X$,
\begin{equation}
    \mathcal D[P_if]^*(O)
    =
    \mathbf 1_{\{P\in\{Z,Y\}\}}\;
    f^\dagger f\,\bigl(Z_iO_XZ_i-O\bigr)
\end{equation}
and for every $O_Z\in\mathcal V_Z$,
\begin{equation}
    \mathcal D[P_if]^*(O_Z)
    =
    \mathbf 1_{\{P\in\{X,Y\}\}}\;
    f^\dagger f\,\bigl(X_iO_ZX_i-O_Z\bigr)
\end{equation}
\end{lemma}

\begin{proof}[Proof idea]
Every element of $\mathcal A$ commutes with the sector subspaces
(Lemma~\ref{lem:hamiltonian-invariance}), so the adjoint dissipator collapses to
$f^\dagger f\,(P_iO_XP_i-O)$; conjugation by $P_i$ acts trivially on $O$ for $P\in\{I,X\}$ and
acts as conjugation by $Z_i$ for $P\in\{Z,Y\}$. Proof deferred to
Appendix~\ref{app:reduction}.
\end{proof}

Summing Lemma~\ref{lem:jump-action} over all jumps gives a single closed generator
for the sector-restricted dynamics using the flip rate operator from Equation~\eqref{eq:flip-rates}:

\begin{corollary}[Sector-restricted instantaneous dynamics]
\label{cor:sector-dynamics-general}
Let $\mathcal L(t)$ be a syndrome-based Lindbladian and $O_X\in\mathcal V_X$. Then
\begin{equation}\label{eq:lindblad-Z}
    \mathcal L^*(t)(O_X)
    =
    \sum_{i=1}^n
    \Gamma_i^{(X)}(t)
    \bigl(
        Z_iO_XZ_i-O_X
    \bigr).
\end{equation}
Similarly, for $O_Z\in\mathcal V_Z$,
\begin{equation}
    \mathcal L^*(t)(O_Z)
    =
    \sum_{i=1}^n
    \Gamma_i^{(Z)}(t)
    \bigl(
        X_iO_ZX_i-O_Z
    \bigr).
\end{equation}
\end{corollary}

\begin{proof}
We prove the $X$-sector statement. The Hamiltonian term vanishes on $O_X$ by
Lemma~\ref{lem:hamiltonian-invariance}. Applying Lemma~\ref{lem:jump-action} to each
jump in Equation~\eqref{eq:lindbladian-heisenberg}, the jumps with $P\in\{I,X\}$
vanish, and each jump $P_if_{P,i,\alpha}(t)$ with $P\in\{Z,Y\}$ contributes
$f_{P,i,\alpha}(t)^\dagger f_{P,i,\alpha}(t)\bigl(Z_iO_XZ_i-O_X\bigr)$. All these jumps
at qubit $i$ induce the same conjugation, so summing over $P\in\{Z,Y\}$ and $\alpha$
collects their coefficients into the flip-rate operator $\Gamma_i^{(X)}(t)$ of
Equation~\eqref{eq:flip-rates}. The $Z$-sector statement follows by exchanging $X$
and $Z$.
\end{proof}

For the quantum-to-classical reduction, this is the only point at which the
microscopic Lindblad representation enters; from here onward the sector dynamics
depend only on the total flip-rate operators.
The rate condition~\eqref{eq:rate-condition} has not been used so far. It enters in
the next subsection: since $\Gamma_i^{(X)}(t)\in\mathcal A_X$, the right-hand side of
Corollary~\ref{cor:sector-dynamics-general} lies again in $\mathcal V_X$
(Lemma~\ref{lem:stab-alg-factor}).

We can now state the following central proposition. It shows that the logical sector subspaces are closed under the syndrome based Lindbladian. In the next subsection, we will show that each member of this $X(Z)$-subspace has a classical description through the $X(Z)$-sector error process.

\begin{proposition}[Syndrome-based Lindbladians preserve the logical sector subspaces]
\label{prop:sector-closure}
Let $\mathcal L(t)$ be a syndrome-based Lindbladian. Then, for every $t\ge0$,
\begin{equation}\label{eq:sector-closure}
    \mathcal L^*(t)(\mathcal V_X)\subseteq\mathcal V_X,
    \qquad
    \mathcal L^*(t)(\mathcal V_Z)\subseteq\mathcal V_Z .
\end{equation}
\end{proposition}

\begin{proof}
Let $O_X=\bar XF$ with $F\in\mathcal A_X$. Conjugation by $Z_i$ maps $\bar X$ to
$\pm\bar X$ and $\mathcal A_X$ into itself, so $Z_iO_XZ_i\in\mathcal V_X$. By the rate
condition~\eqref{eq:rate-condition}, $\Gamma_i^{(X)}(t)\in\mathcal A_X$, and since
$\mathcal A_X$ commutes with $\bar X$, $A\bar XF=\bar X(AF)\in\mathcal V_X$ for every
$A\in\mathcal A_X$. Each term on the right-hand side of
Corollary~\ref{cor:sector-dynamics-general} therefore lies in $\mathcal V_X$. The
$Z$-sector statement follows by exchanging $X$ and $Z$.
\end{proof}

The significance of the following invariance statement is therefore that
Heisenberg evolution preserves this form: starting from a decoded logical
observable in $\mathcal V_X$, the evolution changes only its
syndrome-dependent coefficient.

Closure is a statement about each marginal sector on its own and does not mean
that the two sectors are independent: a $Y_i$ jump flips qubit $i$ in both
sectors and enters both $\Gamma_i^{(X)}(t)$ and $\Gamma_i^{(Z)}(t)$. Our
results use each sector separately and never the joint statistics of the two
sector processes.

\subsection{Classical representation and propagator intertwining}
\label{subsec:classical-repr}
With the sector-restricted dynamics established, we now construct the 
bridge to the classical Markov processes. We define linear maps $\Phi_X$ 
and $\Phi_Z$ that translate operators in $\mathcal{V}_X$ and $\mathcal{V}_Z$ 
into scalar functions on the error vector spaces $\Omega_X$ and 
$\Omega_Z$. The main result of this subsection 
(Proposition~\ref{prop:classical-repr}) shows that the quantum Heisenberg 
evolution of a decoded logical observable, when translated by $\Phi_X$, is 
identical to evolving its classical image under the Markov propagator (Def.~\ref{def:ctmc-generators}). 
As a consequence, the quantum code-space inner product (Def.~\ref{def:reference-state}) of the decoded logical 
observables (Def.~\ref{def:decoded-observables}), which controls the quantum memory lifetime, reduces to a 
classical two-point expectation over the sector error processes (Def.~\ref{def:error-process}), 
completing the proof of Proposition~\ref{prop:reduction-async-formal}.

To represent an observable $\bar X A\in\mathcal V_X$ as a function of a
classical error vector $\mathbf{e}\in \Omega_X$ or $\Omega_Z$, we must keep track of two pieces of
information. The first is the syndrome seen by $A\in \mathcal{A}_X$ (recall from Def.~\ref{def:syndrome-projectors} that $A\in \mathcal{A}_X$ acts on any state with a well defined syndrome $s$ as a scalar). The second is the phase acquired by the bare
logical operator $\bar{X}$ under the Heisenberg dynamics (Eq.~\eqref{eq:lindbladian-heisenberg}), that applies some physical $Z$-Pauli error (Eq.~\eqref{eq:lindblad-Z}). 

\begin{definition}[Bare logical function]
For error vectors in $\mathbf e_X\in\Omega_X$ and $\mathbf e_Z\in\Omega_Z$
define
\begin{equation}
    x_{\mathrm{bare}}(\mathbf e_X)
    :=
    (-1)^{\ell_X\cdot\mathbf e_X},
    \qquad
    z_{\mathrm{bare}}(\mathbf e_Z)
    :=
    (-1)^{\ell_Z\cdot\mathbf e_Z},
\end{equation}
where $\ell_X,\ell_Z$ are the support vectors of the bare logical operators (Definition~\ref{def:bare-logicals})
\end{definition}
We will see below that these are exactly the phases acquired by $\bar X$ and $\bar Z$, respectively, under
the corresponding physical Pauli errors.

\begin{definition}[Classical representation maps]
\label{def:identification-map}
We define the maps
$\Phi_X:\mathcal V_X\to\mathbb R^{\Omega_X}$ and
$\Phi_Z:\mathcal V_Z\to\mathbb R^{\Omega_Z}$ as follows. For $A\in\mathcal A_X, B\in\mathcal A_Z$, 
\begin{align}
    [\Phi_X(\bar X A)](\mathbf e_X)
    &=
    x_{\mathrm{bare}}(\mathbf e_X)\,
    A(H_X\mathbf e_X),
    \qquad A\in\mathcal A_X,
    \\
    [\Phi_Z(\bar Z B)](\mathbf e_Z)
    &=
    z_{\mathrm{bare}}(\mathbf e_Z)\,
    B(H_Z\mathbf e_Z),
    \qquad B\in\mathcal A_Z,
\end{align}
and extend linearly to the entire algebras. These define representations of the algebra $\mathcal{V}_X$ (respectively $\mathcal{V}_Z$) in the algebra of functions 
$\mathbb R^{\Omega_X}$ (respectively, 
$\mathbb R^{\Omega_Z}$). 
\end{definition}

Applying the representation maps to the decoded logical observables and
using Lemma~\ref{lem:algebraic-form} gives
\begin{align}
    [\Phi_X(T_X)](\mathbf e_X)
    &=
    x_{\mathrm{bare}}(\mathbf e_X)
    \varphi_X(H_X\mathbf e_X),
    \\
    [\Phi_Z(T_Z)](\mathbf e_Z)
    &=
    z_{\mathrm{bare}}(\mathbf e_Z)
    \varphi_Z(H_Z\mathbf e_Z).
\end{align}

\begin{definition}[Logical-error phases]\label{def:logical-error-phase}
Denote 
\begin{equation}
    \Lambda_X:=\Phi_X(T_X) \in \mathbb R^{\Omega_X},
    \qquad
    \Lambda_Z:=\Phi_Z(T_Z) \in \mathbb R^{\Omega_Z},
\end{equation}
as the \emph{logical-error phases}.
\end{definition}
Since
$x_{\mathrm{bare}}(\mathbf0)=z_{\mathrm{bare}}(\mathbf0)=1$
and
$\varphi_X(\mathbf0)=\varphi_Z(\mathbf0)=1$,
\begin{equation}
    \Lambda_X(\mathbf0)
    =
    \Lambda_Z(\mathbf0)
    =
    1.
\end{equation}

 The classical memory problem, however, was formulated
directly in terms of the binary decoded logical values $b_X,b_Z$ (Definition~\ref{def:decoded-logical-value}). The
following important lemma shows that these independently constructed descriptions of the quantum object $\Lambda$ and the classical object $b$ coincide. 

\begin{lemma}[Binary and phase descriptions of the decoded logical error]
\label{lem:logical-label-phase}
For every error vector,
\begin{equation}
    \Lambda_X(\mathbf e)
    =
    (-1)^{b_X(\mathbf e)},
    \qquad
    \Lambda_Z(\mathbf e)
    =
    (-1)^{b_Z(\mathbf e)}.
\end{equation}
\end{lemma}

\begin{proof}
We prove the $X$-sector identity. By
Definition~\ref{def:decoded-logical-value} and
Definition~\ref{def:logical-coordinate-identification}, using linearity of $\mu_X$
\begin{align}
    b_X(\mathbf e)
    &=
    \mu_X\!\left(
        \left[
            \mathbf e
            \oplus
            \mathsf{Dec}_X(H_X\mathbf e)
        \right]_X
    \right)
    \\
    &=
    (\ell_X\cdot\mathbf e)
    \oplus
    (\ell_X\cdot
    \mathsf{Dec}_X(H_X\mathbf e)).
\end{align}
Using Equation~\eqref{eq:decoder-phase} gives
\begin{align}
    (-1)^{b_X(\mathbf e)}
    &=
    (-1)^{\ell_X\cdot\mathbf e}
    (-1)^{
        \ell_X\cdot
        \mathsf{Dec}_X(H_X\mathbf e)
    }
    \\
    &=
    x_{\mathrm{bare}}(\mathbf e)\,
    \varphi_X(H_X\mathbf e)
    \\
    &=
    \Lambda_X(\mathbf e).
\end{align}
The $Z$-sector identity is identical.
\end{proof}

By Corollary~\ref{cor:sector-dynamics-general}, the sector-restricted
generator acts on an element of $\mathcal V_X$ through two elementary
operations: multiplication by flip-rate operator ($\Gamma_X$) in $\mathcal{A}_X$ and
conjugation by a physical Pauli $Z_i$. The next two lemmas show how each
operation is transported through $\Phi_X$; together they yield the
quantum--classical intertwining relation described in the diagram in Figure~\ref{fig:intertwine}.

The first action is multiplication by a member of the stabilizer algebra $\mathcal{A}$:

\begin{lemma}[Stabilizer algebra eigenvalue factorization]
\label{lem:stab-alg-factor}
If $A \in \mathcal{A}_X$ and $O = \bar{X}F \in \mathcal{V}_X$ where $F\in\mathcal{A}_X$, then $AO \in \mathcal{V}_X$ and its classical representative satisfies
\begin{equation}
    \Phi_X(AO)(\mathbf{e}_X) = A(H_X\mathbf{e}_X)\Phi_X(O)(\mathbf{e}_X),
\end{equation}
where $A(\cdot)$ is $A$'s syndrome eigenvalue function (Definition~\ref{def:eigenvalue-functions}).
\end{lemma}

\begin{proof}
Since $A, F \in \mathcal{A}_X$ and $\bar{X}$ commutes with all elements of $\mathcal{A}_X$, we reorder $AO = A\bar{X}F = \bar{X}(AF)$, and since $AF \in \mathcal{A}_X$, $AO \in \mathcal{V}_X$. Applying $\Phi_X$:
\begin{equation}
    \Phi_X(AO)(\mathbf{e}_X) = x_{\mathrm{bare}}(\mathbf{e}_X)(AF)(H_X\mathbf{e}_X) 
    = x_{\mathrm{bare}}(\mathbf{e}_X)A(H_X\mathbf{e}_X)F(H_X\mathbf{e}_X) = A(H_X\mathbf{e}_X)\Phi_X(O)(\mathbf{e}_X).
\end{equation}

\end{proof}

The second action, conjugation by a physical Pauli $Z_i$, is worth seeing directly too, as it is the step that turns a quantum error process into a classical random walk: a physical error becomes a single hop in the error-configuration space.
\begin{lemma}[Pauli conjugation as error vector shift]
\label{lem:pauli-conjugation-translation}
Let $O = \bar{X}A \in \mathcal{V}_X$. Conjugation by $Z_i$ yields 
$Z_iOZ_i \in \mathcal{V}_X$, and its classical representative satisfies:
\begin{equation}
    \Phi_X(Z_iOZ_i)(\mathbf{e}_X) = \Phi_X(O)(\mathbf{e}_X \oplus \hat{e}_i).
\end{equation}
Similarly, for $O = \bar{Z}B \in \mathcal{V}_Z$, conjugation by $X_i$ 
yields $X_iOX_i \in \mathcal{V}_Z$ with 
$\Phi_Z(X_iOX_i)(\mathbf{e}_Z) = \Phi_Z(O)(\mathbf{e}_Z \oplus \hat{e}_i)$.
\end{lemma}

\begin{proof}[Proof idea]
Conjugation by $Z_i$ acts separately on the two factors of $O=\bar XA$. On $\bar X$
it only produces the sign $(-1)^{(\ell_X)_i}$, since $Z_i\bar XZ_i=(-1)^{(\ell_X)_i}\bar X$.
On $A\in\mathcal A_X$ it produces $A^{(i)}:=Z_iAZ_i$, whose syndrome eigenvalue
function is $A$'s shifted by $H_X\hat{e}_i$, since $Z_i$ anticommutes with the
$X$-stabilizer generators supported on qubit $i$. Reading $Z_iOZ_i=\bar X\bigl((-1)^{(\ell_X)_i}A^{(i)}\bigr)$
back through $\Phi_X$ therefore produces the sign $(-1)^{(\ell_X)_i}$ times the
shifted syndrome function. This is $\Phi_X(O)$ evaluated at $\mathbf e_X\oplus \hat{e}_i$:
the syndrome shift accounts for $A(H_X(\mathbf e_X\oplus \hat{e}_i))$, and the sign
$(-1)^{(\ell_X)_i}$ is the factor by which $x_{\mathrm{bare}}$ itself changes
under $\mathbf e_X\mapsto\mathbf e_X\oplus \hat{e}_i$. The full proof is deferred to Appendix~\ref{app:reduction}.
\end{proof}

With these two facts in hand, we now establish the main result of this subsection: the representation maps $\Phi_X, \Phi_Z$ intertwine the quantum and classical evolutions.

\begin{proposition}[Intertwining of propagators]
\label{prop:classical-repr}
$\mathcal{V}_X$ is invariant under the instantaneous generator 
$\mathcal{L}^*(t)$ for every $t \ge 0$. Moreover, for any $O \in \mathcal{V}_X$,
\begin{equation}
    \Phi_X\!\left(\mathcal{L}^*(t)(O)\right) = G_X(t)\Phi_X(O), 
\end{equation}
where $G_X(t)$ is the Markov generator defined in Equation~\eqref{eq:ctmc-generator}. Consequently, the propagators $M$ and $\mathcal{W}^*$ intertwine:
\begin{equation}\label{eq:intertwine}
 \Phi_X\!\left(\mathcal{W}^*(t,0)(O)\right) = M^{(X)}(t,0)\Phi_X(O).   
\end{equation}

The analogous statements hold in the $Z$-sector.
\end{proposition}

\begin{figure}[h]
\centering
\begin{tikzcd}[row sep=5em, column sep=8em, ampersand replacement=\&]
O \in \mathcal{V}_X
  \arrow[r, "\Phi_X"]
  \arrow[d, "\mathcal{W}^*(t{,}0)"']
\& \Phi_X(O) \in \mathbb{R}^{\Omega_X}
  \arrow[d, "M^{(X)}(t{,}0)"] \\
\mathcal{W}^*(t{,}0)(O) \in \mathcal{V}_X
  \arrow[r, "\Phi_X"']
\& M^{(X)}(t{,}0)\,\Phi_X(O) \in \mathbb{R}^{\Omega_X}
\end{tikzcd}
\caption{The intertwining relation of Proposition~\ref{prop:classical-repr}. The left column is the Heisenberg evolution of an observable $O$ in the logical sector subspace $\mathcal V_X$ under a syndrome-based Lindbladian; the right column is the backward evolution of its classical representative under the induced $X$-sector Markov chain; the horizontal arrows are the classical representation map $\Phi_X$.}
\label{fig:intertwine}
\end{figure}

\begin{proof}[Proof idea]
The proof applies the two facts above termwise to $\Phi_X(O)$, then upgrades the generator identity to propagators by uniqueness of solutions to the associated linear ODE. Proof deferred to Appendix~\ref{app:reduction}.
\end{proof}

Proposition~\ref{prop:classical-repr} establishes that, under the representation maps, the Heisenberg evolution can be described by the corresponding classical Markov process.

\subsection{Evaluating expectation values}\label{subsec:expectation-values}

We now show that the quantum code-space inner product of the decoded logical observable and any element of $\mathcal{V}_X$ to evaluation of the classical representation at the vacuum error 
configuration $\mathbf{e}=\mathbf{0}$. This will be specifically useful for the time evolved decoded logical observable, which is also a member of $\mathcal{V}_X$.

\begin{lemma}[Classical representation of code space inner product] 
\label{lem:vacuum-evaluation}
Let \(O = \bar{X}F \in \mathcal{V}_X\). Then
\begin{equation}
    \left\langle T_X, O \right\rangle_{\mathrm{code}} = \Phi_X(O)(\mathbf{0}).
\end{equation}
Similarly, for \(O = \bar{Z}F \in \mathcal{V}_Z\),
\begin{equation}
    \left\langle T_Z, O \right\rangle_{\mathrm{code}} = \Phi_Z(O)(\mathbf{0}).
\end{equation}
\end{lemma}

\begin{proof}[Proof idea]
Since $\rho_{\mathrm{code}}$ is supported on the zero syndrome, any $\mathcal{A}_X$ element reduces to its value at $\mathbf{0}$, matching $\Phi_X$ at the vacuum configuration. Proof deferred to Appendix~\ref{app:reduction}.
\end{proof}

\subsection{Proof of the reduction identity}\label{subsec:main-reduction}
We now have all ingredients to prove the main result of this section, Proposition~\ref{prop:reduction-async-formal}.
\begin{proof}[Proof of  Proposition~\ref{prop:reduction-async-formal}]
We prove the \(X\)-sector statement; the \(Z\)-sector follows by exchanging
\(X\) and \(Z\).

Since \(T_X\in\mathcal V_X\) and \(\mathcal V_X\) is invariant under the
instantaneous generator $\mathcal{L}^*(t)$ by Proposition~\ref{prop:sector-closure},  by time integration this invariance is kept also by the 
propagator \(\mathcal W^*(t,0)\), thus
\begin{equation}
    \mathcal W^*(t,0)(T_X)\in\mathcal V_X.
\end{equation}
By the representation of the code space inner product Lemma~\ref{lem:vacuum-evaluation},
\begin{equation}\label{eq:inner-product-rep}
    \left\langle
        T_X,\mathcal W^*(t,0)(T_X)
    \right\rangle_{\mathrm{code}}
    =
    \Phi_X\!\left(
        \mathcal W^*(t,0)(T_X)
    \right)(\mathbf{0}).
\end{equation}
Using the propagator intertwining relation from
Proposition~\ref{prop:classical-repr} (Eq.~\eqref{eq:intertwine}), we obtain for the RHS of Equation~\eqref{eq:inner-product-rep}:
\begin{equation}
    \Phi_X\!\left(
        \mathcal W^*(t,0)(T_X)
    \right)(\mathbf{0})
    =
    \left(
        M^{(X)}(t,0)\Phi_X(T_X)
    \right)(\mathbf{0}).
\end{equation}
We plug this into Equation~\eqref{eq:inner-product-rep}, and use the definition $\Phi_X(T_X)=\Lambda_X$ (Def.~\ref{def:logical-error-phase}). We further use the definition of the classical propagator on the error sector error process, which is the same as
Eq.~\eqref{eq:classical-propagator} according to codeword-translation invariance (Claim~\ref{clm:translation-invariance}). This yields
\begin{equation}
    \left\langle
        T_X,\mathcal W^*(t,0)(T_X)
    \right\rangle_{\mathrm{code}}
    =
    \mathbb E\!\left[
        \Lambda_X(\xi_t^{(X)})
    \right],
\end{equation}
where we also used the fact that $\xi_0^{(X)}=\mathbf{0}$
By Lemma~\ref{lem:logical-label-phase},
\begin{equation}
    \Lambda_X(\mathbf e)
    =
    (-1)^{b_X(\mathbf e)},
\end{equation}
so
\begin{align}
    \left\langle
        T_X,\mathcal W^*(t,0)(T_X)
    \right\rangle_{\mathrm{code}}
    &=
    \mathbb E\!\left[
        (-1)^{b_X(\xi_t^{(X)})}
    \right]\\
    &=
    1
    -
    2\Pr\!\left[
        b_X(\xi_t^{(X)})=1
    \right],
\end{align}
The $Z$-sector identity follows identically.
\end{proof}

\subsection{From autocorrelation to quantum memory}
\label{subsec:memory}
The passage from decoded-logical autocorrelations to operational quantum
memory does not depend on whether the physical dynamics is continuous or
synchronous.  We therefore formulate it for an arbitrary physical CPTP map
\(\mathcal W_*\).

The only dynamical input needed is covariance under conjugation by the two
bare logical Paulis. Under this assumption, the induced effective logical
channel is a Pauli channel, and its \(X_Q\)- and \(Z_Q\)-eigenvalues are
exactly the corresponding decoded-logical autocorrelations. These
eigenvalues then give both lower and upper bounds on the worst-case encoded
fidelity.

For continuous-time syndrome-based Lindbladians, the required covariance
is verified below.  The synchronous analogue is established separately in
Section~\ref{sec:sync-reduction}; after that verification, the results of
this subsection apply verbatim.

First, we characterize the effective logical channel associated with the
physical map \(\mathcal W_*\) (Subsection~\ref{subsec:pauli-channel}).
Under covariance with respect to the two bare logical Paulis,
Lemma~\ref{lem:pauli-diagonality} shows that this effective channel is
diagonal in the carrier Pauli basis, hence is a Pauli channel. The same lemma
identifies its \(X_Q\)- and \(Z_Q\)-eigenvalues with the corresponding
decoded-logical autocorrelations.

Second, we relate these two eigenvalues to the worst-case encoded fidelity
(Subsection~\ref{subsec:fidelity-bound}). Proposition~\ref{prop:minimal-fidelity}
gives both lower and upper bounds on \(F_{\min}[\mathcal W_*]\) in terms of
the two decoded-logical autocorrelations. Combined with the corresponding
autocorrelation identities for the continuous or synchronous dynamics, these
bounds yield the classical--quantum memory implications in both directions.

\subsubsection{Pauli structure of the effective logical channel}
\label{subsec:pauli-channel}
The fidelity bound we are after depends on the effective logical channel
\(\mathcal{E}_L[\mathcal W_*]\) (Definition~\ref{def:effective-logical-channel}) --- the
composition of encoding, physical dynamics, and decoding that the logical qubit
experiences end to end. A priori this channel could mix \(X_Q,Y_Q,Z_Q\)
together in complicated ways; the next lemma shows that, under the
logical-Pauli covariance condition below, it does not: it is diagonal in the
Pauli basis. In the Heisenberg picture, this means that \(X_Q,Y_Q,Z_Q\) are
eigenvectors with eigenvalues \(\lambda_X,\lambda_Y,\lambda_Z\).

Importantly, \(\lambda_X\) and \(\lambda_Z\) are exactly the corresponding
decoded-logical autocorrelations. This is what allows us to convert the
classical memory bounds into bounds on the quantum worst-case fidelity.
Since our sector processes track only the \(X\)- and \(Z\)-sectors, they do
not determine \(\lambda_Y\), and therefore do not determine the exact
worst-case fidelity. Nevertheless, \(\lambda_X\) and \(\lambda_Z\) alone
provide the lower and upper bounds needed to establish the classical--quantum
memory equivalence below.

\begin{lemma}[Pauli covariance and the effective logical channel]
\label{lem:pauli-diagonality}
Let $\mathcal W_*:\mathcal T(\mathcal H)\to\mathcal T(\mathcal H)$
be a CPTP map whose Heisenberg adjoint satisfies
\begin{equation}\label{eq:bare-logical-covariance}
    \mathcal W^*(\bar P O\bar P)
    =
    \bar P\,\mathcal W^*(O)\,\bar P,
    \qquad
    \bar P\in\{\bar X,\bar Z\},
\end{equation}
for every \(O\in\mathcal B(\mathcal H)\).
Then the effective logical channel
\(\mathcal E_L[\mathcal W_*]\)
of Definition~\ref{def:effective-logical-channel} is a Pauli channel.
Equivalently,
\begin{equation}\label{eq:pauli-diagonal}
    \mathcal E_L^*[\mathcal W_*](X_Q)
    =
    \lambda_X X_Q,
    \qquad
    \mathcal E_L^*[\mathcal W_*](Y_Q)
    =
    \lambda_Y Y_Q,
    \qquad
    \mathcal E_L^*[\mathcal W_*](Z_Q)
    =
    \lambda_Z Z_Q.
\end{equation}
Moreover,
\begin{equation}\label{eq:generic-logical-eigenvalues}
    \lambda_X
    =
    \bigl\langle
        T_X,\mathcal W^*(T_X)
    \bigr\rangle_{\mathrm{code}},
    \qquad
    \lambda_Z
    =
    \bigl\langle
        T_Z,\mathcal W^*(T_Z)
    \bigr\rangle_{\mathrm{code}}.
\end{equation}
\end{lemma}

\begin{proof}[Proof idea]
Each of the three maps composing \(\mathcal E_L^*[\mathcal W_*]\) respects
conjugation by the logical Paulis: \(\mathcal D_{\mathrm{dec}}^*\) carries
conjugation by \(X_Q\) to conjugation by \(\bar X\),
\(\mathcal W^*\) preserves it by Eq.~\eqref{eq:bare-logical-covariance}, and
\(\mathcal E_{\mathrm{enc}}^*\) carries it back to conjugation by \(X_Q\);
similarly for \(Z_Q\). Hence \(\mathcal E_L^*[\mathcal W_*]\) commutes with
conjugation by \(X_Q\) and \(Z_Q\), so it maps each carrier Pauli to an
operator with the same commutation relations with \(X_Q\) and \(Z_Q\), that
is, to a multiple of itself. For the eigenvalues, pulling the trace against
\(X_Q\) back through the encoding gives
\(\lambda_X=\operatorname{Tr}\bigl(\rho_{\mathrm{code}}\bar X\,\mathcal W^*(T_X)\bigr)\),
and \(\bar X\) may be replaced by \(T_X\) since the two agree on the code
space; similarly for \(\lambda_Z\). The full proof is given in
Appendix~\ref{app:reduction}.
\end{proof}

\subsubsection{Fidelity bound from Pauli channel structure}
\label{subsec:fidelity-bound}
With the effective logical channel identified as a Pauli channel, 
we can now bound the worst-case fidelity in terms of the effective logical channel's eigenvalues.  
In the next subsection we combine this with the classical reduction (Proposition~\ref{prop:reduction-async-formal}) to complete the proof 
of Theorem~\ref{thm:main-async-formal}.

\begin{proposition}[Worst-case fidelity bounds from logical autocorrelations]
\label{prop:minimal-fidelity}
Under the hypotheses of Lemma~\ref{lem:pauli-diagonality},
\begin{equation}\label{eq:generic-fidelity-bounds}
    \frac12
    \left(
        \bigl\langle T_X,\mathcal W^*(T_X)\bigr\rangle_{\mathrm{code}}
        +
        \bigl\langle T_Z,\mathcal W^*(T_Z)\bigr\rangle_{\mathrm{code}}
    \right)
    \le
    F_{\min}[\mathcal W_*]
    \le
    \frac12
    \left(
        1+
        \min_{\sigma\in\{X,Z\}}
        \bigl\langle T_\sigma,\mathcal W^*(T_\sigma)\bigr\rangle_{\mathrm{code}}
    \right).
\end{equation}
\end{proposition}

\begin{proof}
By Lemma~\ref{lem:pauli-diagonality}, the effective logical channel is a single-qubit Pauli channel, 
\begin{equation}
    \mathcal E_L\mathcal W_* = \sum_{P\in\{I_Q,X_Q,Y_Q,Z_Q\}} p_P P\rho P,
\end{equation} with $p_I,p_X,p_Y,p_Z\ge0$ and $\sum_P p_P=1$. For any pure state $|\psi\rangle$, $F(\psi;\mathcal W_*) = \sum_P p_P|\langle\psi|P|\psi\rangle|^2 \ge p_I$, so $F_{\min}[\mathcal W_*]\ge p_I$. Since $\lambda_X=p_I+p_X-p_Y-p_Z$ and $\lambda_Z=p_I-p_X-p_Y+p_Z$, we have $(\lambda_X+\lambda_Z)/2 = p_I-p_Y \le p_I$, which proves $F_{\min}[\mathcal W_*] \ge (\lambda_X+\lambda_Z)/2$.

For the upper bounds, evaluating the fidelity on an $X_Q$-eigenstate $|x\rangle$ leaves only the $I_Q$ and $X_Q$ terms, giving $F(x;\mathcal W_*) = p_I+p_X = (1+\lambda_X)/2$, and hence $F_{\min}[\mathcal W_*] \le (1+\lambda_X)/2$. Likewise, evaluating on a $Z_Q$-eigenstate yields $F_{\min}[\mathcal W_*] \le (1+\lambda_Z)/2$. Finally, using Eq.~\eqref{eq:generic-logical-eigenvalues}, $\lambda_X = \bigl\langle T_X,\mathcal W_*(T_X)\bigr\rangle_{\mathrm{code}}$ and $\lambda_Z = \bigl\langle T_Z,\mathcal W_*(T_Z)\bigr\rangle_{\mathrm{code}}$, which yields the required result.
\end{proof}

\subsubsection{Specialization to Lindbladian evolution}
\label{subsec:memory-completion}

To apply the preceding results to the continuous-time dynamics, it remains
only to verify the logical-Pauli covariance
of Eq.~\eqref{eq:bare-logical-covariance} for a syndrome-based
Lindbladian.

\begin{lemma}[Logical-Pauli covariance of syndrome-based Lindbladian evolution]
\label{lem:async-pauli-covariance}
For every syndrome-based Lindbladian and every \(t\ge0\), its
Heisenberg-picture propagator satisfies
\begin{equation}\label{eq:async-pauli-covariance}
    \mathcal W^*(t,0)(\bar P O\bar P)
    =
    \bar P\,\mathcal W^*(t,0)(O)\,\bar P,
    \qquad
    \bar P\in\{\bar X,\bar Z\},
\end{equation}
for every \(O\in\mathcal B(\mathcal H)\).
\end{lemma}

\begin{proof}
Consider a jump operator
    $L=P_i f_{P,i,\alpha}(t)$
of the syndrome-based Lindbladian. Since
\(f_{P,i,\alpha}(t)\in\mathcal A\) and the bare logical operators lie in
the stabilizer centralizer,
    $[\bar P,f_{P,i,\alpha}(t)]=0.$
Moreover, \(P_i\) and \(\bar P\) are Pauli strings, so
$\bar P P_i\bar P=\pm P_i.$
Hence $\bar P L\bar P=\pm L.$
The adjoint dissipator is invariant under \(L\mapsto-L\), and therefore
\[
    \mathcal D[L]^*(\bar P O\bar P)
    =
    \bar P\,\mathcal D[L]^*(O)\,\bar P.
\]
The Hamiltonian belongs to \(\mathcal A\), so
\([\bar P,H(t)]=0\), and its Heisenberg contribution obeys the same
covariance. Thus
\[
    \mathcal L^*(t)(\bar P O\bar P)
    =
    \bar P\,\mathcal L^*(t)(O)\,\bar P.
\]
Uniqueness of the Heisenberg evolution then gives
Eq.~\eqref{eq:async-pauli-covariance}.
\end{proof}

We now combine this covariance with the continuous-time
autocorrelation identities of Proposition~\ref{prop:reduction-async-formal}
to prove the memory statement.

\begin{proof}[Proof of Theorem~\ref{thm:main-async-formal}]
Fix \(t'\in[0,t]\). By Lemma~\ref{lem:async-pauli-covariance},
the channel \(\mathcal W_*(t',0)\) satisfies the covariance hypothesis of
Lemma~\ref{lem:pauli-diagonality}, so
Proposition~\ref{prop:minimal-fidelity} applies.

Suppose first that the \(X\)- and \(Z\)-sector processes respectively
\(q_X\)- and \(q_Z\)-remember \(b_X\) and \(b_Z\) up to time \(t\).
Proposition~\ref{prop:reduction-async-formal} then gives
\begin{equation}
    \bigl\langle T_X,\mathcal W^*(t',0)(T_X)\bigr\rangle_{\mathrm{code}}
    \ge 1-2q_X,
    \qquad
    \bigl\langle T_Z,\mathcal W^*(t',0)(T_Z)\bigr\rangle_{\mathrm{code}}
    \ge 1-2q_Z.
\end{equation}
The lower bound of Proposition~\ref{prop:minimal-fidelity} therefore yields
\begin{equation}
    F_{\min}[\mathcal W_*(t',0)]
    \ge
    1-q_X-q_Z.
\end{equation}
Since \(t'\in[0,t]\) was arbitrary, the encoded system is a quantum
\((q_X+q_Z)\)-memory up to time \(t\).

Conversely, suppose that the encoded system is a quantum
\(\memerr\)-memory up to time \(t\). Then
\begin{equation}
    F_{\min}[\mathcal W_*(t',0)]
    \ge
    1-\memerr
\end{equation}
for every \(t'\in[0,t]\). Combining this with the upper bound of
Proposition~\ref{prop:minimal-fidelity} gives
\begin{equation}
    1-\memerr
    \le
    \frac12\left(
        1+
        \min_{\sigma\in\{X,Z\}}
        \bigl\langle
            T_\sigma,\mathcal W^*(t',0)(T_\sigma)
        \bigr\rangle_{\mathrm{code}}
    \right),
\end{equation}
and hence
\begin{equation}
    \min_{\sigma\in\{X,Z\}}
    \bigl\langle
        T_\sigma,\mathcal W^*(t',0)(T_\sigma)
    \bigr\rangle_{\mathrm{code}}
    \ge
    1-2\memerr.
\end{equation}
Thus both sector autocorrelations are at least \(1-2\memerr\).
Using Proposition~\ref{prop:reduction-async-formal},
\begin{equation}
    \Pr\!\left[
        b_\sigma(\xi_{t'}^{(\sigma)})
        \ne
        b_\sigma(\xi_0^{(\sigma)})
    \right]
    \le
    \memerr,
    \qquad
    \sigma\in\{X,Z\}.
\end{equation}
Since \(t'\in[0,t]\) was arbitrary, each induced sector process
\(\memerr\)-remembers its logical-error indicator up to time \(t\).
\end{proof}

\subsection{Application: thermal memory in the 4D toric code}
\label{subsec:alicki}

We close the asynchronous analysis with a simple application of the
reduction.  In the 4D toric code, the thermal Davies dynamics considered by
Alicki et al.~\cite{alicki_thermal_2010} is syndrome-based, so their thermal
memory result can be recovered by combining 
Theorem~\ref{thm:main-async-formal} with the corresponding classical Peierls
estimate.

\begin{remark}[The thermal 4D toric-code memory of Alicki et al.]
\label{rem:alicki}
The thermal-memory result of Alicki et
al.~\cite{alicki_thermal_2010} follows from 
Theorem~\ref{thm:main-async-formal} together with a purely classical
Peierls argument.  We use the immediate thermal-state variant of
Theorem~\ref{thm:main-async-formal}, in which the induced sector
processes are initialized in their Gibbs stationary measures; the
proof is unchanged, with the sector identity averaged over the
initial configuration.  By Remark~\ref{rem:thermal}, the Davies
generator of the 4D toric-code Hamiltonian is syndrome-based, and
its two induced sector processes are the two self-dual copies of the
classical 4D $\mathbb Z_2$ gauge model.  It therefore suffices to
prove classical memory for one sector.

Let $\pi_\beta$ denote the Gibbs measure of this sector, and let
$(\xi_s)_{s\ge0}$ be the corresponding stationary Markov process.
Define
\begin{equation}
    B_L
    :=
    \left\{
        \text{the syndrome contains a loop of length at least }L/8
    \right\}.
\end{equation}
The Peierls argument of Ref.~\cite{alicki_thermal_2010} shows that,
for sufficiently large $\beta$,
\begin{equation}
    \pi_\beta(B_L)
    \le
    \operatorname{poly}(L)e^{-c_\beta L}
\end{equation}
for some $c_\beta>0$.

Choose the classical decoder by filling each syndrome loop shorter
than $L/8$ by a shortest surface, and extend it arbitrarily on
$B_L$.  Let $b$ be the resulting decoded logical bit.  The
homological argument of Ref.~\cite{alicki_thermal_2010} shows that
if $\mathbf x'=\mathbf x\oplus e_i$ differs from $\mathbf x$ by a
single plaquette flip and both
$\mathbf x,\mathbf x'\notin B_L$, then
\begin{equation}
    b(\mathbf x)=b(\mathbf x').
\end{equation}
Thus a change of the decoded logical bit requires the process to
encounter a configuration in $B_L$.  Since the process is
stationary and reversible, and its local jump rates are uniformly
bounded in $L$, the Peierls estimate implies
\begin{equation}
    \Pr\!\left[
        b(\xi_s)\ne b(\xi_0)
    \right]
    \le
    \operatorname{poly}(L)\,s\,e^{-c_\beta L}.
\end{equation}
In particular, for every $s\le t$, this sector
$q_L(t)$-remembers its decoded logical bit with
\begin{equation}
    q_L(t)
    \le
    \operatorname{poly}(L)\,t\,e^{-c_\beta L}.
\end{equation}

By self-duality, the same bound holds for the other CSS sector.
Theorem~\ref{thm:main-async-formal} therefore yields
\begin{equation}
    F_{\min}(t)
    \ge
    1-q_X(t)-q_Z(t)
    \ge
    1-\operatorname{poly}(L)\,t\,e^{-c_\beta L}.
\end{equation}
Hence, below a constant critical temperature, the 4D toric code
stores an encoded qubit for time $e^{\Omega(L)}$, recovering the
thermal-memory result of Ref.~\cite{alicki_thermal_2010}.
\end{remark}

\section{Synchronous quantum-to-classical reduction}
\label{sec:sync-reduction}

We now prove the discrete-time counterparts, 
Proposition~\ref{prop:reduction-sync-formal} and 
Theorem~\ref{thm:main-sync}, of the continuous-time reduction of
Section~\ref{sec:async}.  All algebraic objects
used in the reduction are unchanged: the decoded logical observables
\(T_X,T_Z\), the logical sector subspaces
\(\mathcal V_X,\mathcal V_Z\), the representation maps
\(\Phi_X,\Phi_Z\), and the logical-error phases
\(\Lambda_X,\Lambda_Z\) are those defined and analyzed in
Section~\ref{sec:async}. In particular,
Lemma~\ref{lem:logical-label-phase} gives
\begin{equation}\label{eq:sync-logical-phase}
    \Phi_\sigma(T_\sigma)(\mathbf e)
    =
    (-1)^{b_\sigma(\mathbf e)}.
\end{equation}

The difference is entirely dynamical.  The continuous-time
intertwining
\(
\Phi_\sigma\mathcal L^*
=
G_\sigma\Phi_\sigma
\)
is replaced by an intertwining between the synchronous quantum
channel and its induced discrete-time sector Markov chain. Subsection~\ref{subsec:sync-spca-realization} then proves 
Corollary~\ref{cor:sync-spca-memory}.

\subsection{One-step sector dynamics}
\label{subsec:sync-sector-dynamics}

We first determine the action of one synchronous round on the logical
sector subspaces.  Repeated application of
Lemma~\ref{lem:pauli-conjugation-translation} gives, for
\(\mathbf u^X\in\Omega_X\),
\begin{equation}\label{eq:sync-multibit-translation-X}
    \Phi_X\!\left(
        Z^{\mathbf u^X}O_XZ^{\mathbf u^X}
    \right)(\mathbf e)
    =
    \Phi_X(O_X)(\mathbf e\oplus\mathbf u^X),
    \qquad
    O_X\in\mathcal V_X,
\end{equation}
and symmetrically
\begin{equation}\label{eq:sync-multibit-translation-Z}
    \Phi_Z\!\left(
        X^{\mathbf u^Z}O_ZX^{\mathbf u^Z}
    \right)(\mathbf e)
    =
    \Phi_Z(O_Z)(\mathbf e\oplus\mathbf u^Z).
\end{equation}

\begin{proposition}[Closed discrete sector dynamics]
\label{prop:sync-sector-dynamics}
Let \(\mathcal Q_*\) be a syndrome-based quantum channel.
For every \(O_X\in\mathcal V_X\),
\begin{equation}\label{eq:sync-sector-action-X}
    \mathcal Q^*(O_X)
    =
    \sum_{\mathbf u^X\in\Omega_X}
    \mathsf P_{\mathbf u^X}^{(X)}
    Z^{\mathbf u^X}O_XZ^{\mathbf u^X},
\end{equation}
and therefore
\(
\mathcal Q^*(\mathcal V_X)\subseteq\mathcal V_X
\).
Similarly,
\begin{equation}\label{eq:sync-sector-action-Z}
    \mathcal Q^*(O_Z)
    =
    \sum_{\mathbf u^Z\in\Omega_Z}
    \mathsf P_{\mathbf u^Z}^{(Z)}
    X^{\mathbf u^Z}O_ZX^{\mathbf u^Z},
\end{equation}
and
\(
\mathcal Q^*(\mathcal V_Z)\subseteq\mathcal V_Z
\).
\end{proposition}

\begin{proof}
We prove the \(X\)-sector statement.
Using the Kraus representation of
Definition~\ref{def:sync-syndrome-channel},
\begin{align}
    \mathcal Q^*(O_X)
    &=
    \sum_{\mathbf u^X,\mathbf u^Z,\alpha}
    f_{\mathbf u^X,\mathbf u^Z,\alpha}^\dagger
    Z^{\mathbf u^X}X^{\mathbf u^Z}
    O_X
    X^{\mathbf u^Z}Z^{\mathbf u^X}
    f_{\mathbf u^X,\mathbf u^Z,\alpha}.
\end{align}
Every \(O_X\in\mathcal V_X\) commutes with \(X^{\mathbf u^Z}\), since
\(\bar X\) and the elements of \(\mathcal A_X\) are linear combinations
of \(X\)-type Pauli strings. Hence
\begin{equation}
    \mathcal Q^*(O_X)
    =
    \sum_{\mathbf u^X,\mathbf u^Z,\alpha}
    f_{\mathbf u^X,\mathbf u^Z,\alpha}^\dagger
    Z^{\mathbf u^X}O_XZ^{\mathbf u^X}
    f_{\mathbf u^X,\mathbf u^Z,\alpha}.
\end{equation}
By Equation~\eqref{eq:sync-multibit-translation-X},
\(Z^{\mathbf u^X}O_XZ^{\mathbf u^X}\in\mathcal V_X\). Every element of
\(\mathcal V_X\) commutes with \(\mathcal A\), because \(\bar X\) commutes
with every stabilizer and \(\mathcal A\) is Abelian. Therefore
\begin{align}
    \mathcal Q^*(O_X)
    &=
    \sum_{\mathbf u^X,\mathbf u^Z,\alpha}
    f_{\mathbf u^X,\mathbf u^Z,\alpha}^\dagger
    f_{\mathbf u^X,\mathbf u^Z,\alpha}
    Z^{\mathbf u^X}O_XZ^{\mathbf u^X}
    \\
    &=
    \sum_{\mathbf u^X\in\Omega_X}
    \mathsf P_{\mathbf u^X}^{(X)}
    Z^{\mathbf u^X}O_XZ^{\mathbf u^X}.
\end{align}
Each term is the product of
\(\mathsf P_{\mathbf u^X}^{(X)}\in\mathcal A_X\) with an element of
\(\mathcal V_X\), and therefore lies in \(\mathcal V_X\).
The \(Z\)-sector proof is identical.
\end{proof}

\subsection{Classical representation and intertwining}
\label{subsec:sync-intertwining}

\begin{proposition}[Synchronous quantum--classical intertwining]
\label{prop:sync-intertwining}
For every \(O_X\in\mathcal V_X\),
\begin{equation}\label{eq:sync-intertwining-X}
    \Phi_X\!\left(
        \mathcal Q^*(O_X)
    \right)
    =
    \mathsf B_X\Phi_X(O_X),
\end{equation}
and similarly
\begin{equation}
    \Phi_Z\!\left(
        \mathcal Q^*(O_Z)
    \right)
    =
    \mathsf B_Z\Phi_Z(O_Z).
\end{equation}
Consequently, for every integer \(r\ge0\),
\begin{equation}\label{eq:sync-propagator-intertwining}
    \Phi_\sigma\!\left(
        (\mathcal Q^*)^r(O_\sigma)
    \right)
    =
    \mathsf B_\sigma^r
    \Phi_\sigma(O_\sigma).
\end{equation}
\end{proposition}

\begin{proof}
For the \(X\) sector, Proposition~\ref{prop:sync-sector-dynamics},
Lemma~\ref{lem:stab-alg-factor}, and
Equation~\eqref{eq:sync-multibit-translation-X} give
\begin{align}
    \Phi_X\!\left(
        \mathcal Q^*(O_X)
    \right)(\mathbf e)
    &=
    \sum_{\mathbf u^X}
    \mathsf P_{\mathbf u^X}^{(X)}(H_X\mathbf e)
    \,
    \Phi_X(O_X)(\mathbf e\oplus\mathbf u^X)
    \\
    &=
    \sum_{\mathbf u^X}
    p_X(
        \mathbf u^X
        \mid
        H_X\mathbf e
    )
    \Phi_X(O_X)(\mathbf e\oplus\mathbf u^X)
    \\
    &=
    \bigl(
        \mathsf B_X\Phi_X(O_X)
    \bigr)(\mathbf e).
\end{align}
This proves the one-step identity.  Iterating it yields
Equation~\eqref{eq:sync-propagator-intertwining}.
The \(Z\) sector is identical.
\end{proof}

Unlike the continuous-time case, no differential equation is involved:
the propagator identity follows directly by iteration of the one-step
intertwining.

\begin{proof}[Proof of Proposition~\ref{prop:reduction-sync-formal}]
We prove the \(X\)-sector statement.

By Proposition~\ref{prop:sync-sector-dynamics},
\(
(\mathcal Q^*)^r(T_X)\in\mathcal V_X
\).
Lemma~\ref{lem:vacuum-evaluation} therefore gives
\begin{equation}
    \left\langle
        T_X,
        (\mathcal Q^*)^r(T_X)
    \right\rangle_{\mathrm{code}}
    =
    \Phi_X\!\left(
        (\mathcal Q^*)^r(T_X)
    \right)(\mathbf0).
\end{equation}
Using Proposition~\ref{prop:sync-intertwining},
\begin{align}
    \Phi_X\!\left(
        (\mathcal Q^*)^r(T_X)
    \right)(\mathbf0)
    &=
    \bigl(
        \mathsf B_X^r\Phi_X(T_X)
    \bigr)(\mathbf0)
    \\
    &=
    \mathbb E_{\mathbf0}\!\left[
        \Phi_X(T_X)(\xi_r^{(X)})
    \right].
\end{align}
By Lemma~\ref{lem:logical-label-phase},
\(
\Phi_X(T_X)=(-1)^{b_X}
\),
and since
\(
\xi_0^{(X)}=\mathbf0
\)
and
\(b_X(\mathbf0)=0\),
\begin{align}
    \left\langle
        T_X,
        (\mathcal Q^*)^r(T_X)
    \right\rangle_{\mathrm{code}}
    &=
    \mathbb E\!\left[
        (-1)^{
            b_X(\xi_r^{(X)})
            \oplus
            b_X(\xi_0^{(X)})
        }
    \right]
    \\
    &=
    1
    -
    2\Pr\!\left[
        b_X(\xi_r^{(X)})
        \ne
        b_X(\xi_0^{(X)})
    \right].
    \label{eq:sync-autocorrelation}
\end{align}
The \(Z\)-sector identity follows identically.
\end{proof}

\subsection{From discrete classical memory to quantum memory}
\label{subsec:sync-memory}

The operational part of the argument is the same as in the
continuous-time setting.  It remains only to verify that synchronous
syndrome-based channels satisfy the logical-Pauli covariance required
by Lemma~\ref{lem:pauli-diagonality}.

\begin{lemma}[Logical-Pauli covariance of synchronous evolution]
\label{lem:sync-pauli-covariance}
Let \(\mathcal Q_*\) be a syndrome-based quantum channel.
Then its Heisenberg adjoint satisfies
\begin{equation}\label{eq:sync-pauli-covariance}
    \mathcal Q^*(\bar P O\bar P)
    =
    \bar P\,\mathcal Q^*(O)\,\bar P,
    \qquad
    \bar P\in\{\bar X,\bar Z\},
\end{equation}
for every \(O\in\mathcal B(\mathcal H)\).
Consequently,
\begin{equation}
    (\mathcal Q^*)^r(\bar P O\bar P)
    =
    \bar P\,(\mathcal Q^*)^r(O)\,\bar P
\end{equation}
for every \(r\ge0\).
\end{lemma}

\begin{proof}
Each Kraus operator has the form
\begin{equation}
    K_{\mathbf u^X,\mathbf u^Z,\alpha}
    =
    X^{\mathbf u^Z}
    Z^{\mathbf u^X}
    f_{\mathbf u^X,\mathbf u^Z,\alpha},
    \qquad
    f_{\mathbf u^X,\mathbf u^Z,\alpha}\in\mathcal A.
\end{equation}
The bare logical operators commute with every element of
\(\mathcal A\), while any two Pauli strings either commute or
anticommute. Hence
\begin{equation}
    \bar P
    K_{\mathbf u^X,\mathbf u^Z,\alpha}
    \bar P
    =
    \pm
    K_{\mathbf u^X,\mathbf u^Z,\alpha}.
\end{equation}
Therefore, using the adjoint Kraus representation,
\begin{align}
    \mathcal Q^*(\bar P O\bar P)
    &=
    \sum_{\mathbf u^X,\mathbf u^Z,\alpha}
    K_{\mathbf u^X,\mathbf u^Z,\alpha}^\dagger
    \bar P O\bar P
    K_{\mathbf u^X,\mathbf u^Z,\alpha}
    \\
    &=
    \bar P
    \left(
        \sum_{\mathbf u^X,\mathbf u^Z,\alpha}
        K_{\mathbf u^X,\mathbf u^Z,\alpha}^\dagger
        O
        K_{\mathbf u^X,\mathbf u^Z,\alpha}
    \right)
    \bar P
    \\
    &=
    \bar P\,\mathcal Q^*(O)\,\bar P.
\end{align}
Covariance of \((\mathcal Q^*)^r\) follows by composition.
\end{proof}

\begin{proof}[Proof of Theorem~\ref{thm:main-sync}]
Fix \(0\le r\le m\). By
Lemma~\ref{lem:sync-pauli-covariance}, the physical channel
    $\mathcal W_*=\mathcal Q_*^r$
satisfies the covariance hypothesis of
Lemma~\ref{lem:pauli-diagonality}. Hence
Proposition~\ref{prop:minimal-fidelity} applies.

Suppose first that the two sector processes respectively
\(q_X\)- and \(q_Z\)-remember \(b_X\) and \(b_Z\) through round \(m\).
By Proposition~\ref{prop:reduction-sync-formal},
\begin{align}
    \left\langle
        T_X,(\mathcal Q^*)^r(T_X)
    \right\rangle_{\mathrm{code}}
    &\ge 1-2q_X,
    \\
    \left\langle
        T_Z,(\mathcal Q^*)^r(T_Z)
    \right\rangle_{\mathrm{code}}
    &\ge 1-2q_Z.
\end{align}
The lower bound of Proposition~\ref{prop:minimal-fidelity} gives
\[
    F_{\min}[\mathcal Q_*^r]
    \ge
    1-q_X-q_Z.
\]
Since \(r\le m\) was arbitrary, the encoded system is a quantum
\((q_X+q_Z)\)-memory through round \(m\).

Conversely, suppose that the encoded system is a quantum
\(\memerr\)-memory through round \(m\). Then
\[
    F_{\min}[\mathcal Q_*^r]\ge1-\memerr
\]
for every \(0\le r\le m\). By the \(X\)-upper bound of
Proposition~\ref{prop:minimal-fidelity} and 
Proposition~\ref{prop:reduction-sync-formal},
\begin{align}
    1-\memerr
    &\le
    \frac12\left(
        1+
        \left\langle
            T_X,(\mathcal Q^*)^r(T_X)
        \right\rangle_{\mathrm{code}}
    \right)
    \\
    &=
    1-
    \Pr\!\left[
        b_X(\xi_r^{(X)})
        \ne
        b_X(\xi_0^{(X)})
    \right].
\end{align}
Thus
\[
    \Pr\!\left[
        b_X(\xi_r^{(X)})
        \ne
        b_X(\xi_0^{(X)})
    \right]
    \le\memerr.
\]
The \(Z\)-sector is identical. Hence each induced sector process
\(\memerr\)-remembers its logical-error indicator through round \(m\).
\end{proof}

\subsection{From classical SPCAs to local quantum channels}
\label{subsec:sync-spca-realization}

We now prove Corollary~\ref{cor:sync-spca-memory}. The channel is
built from the quantum lifts of the two SPCA rules, which implement the
corresponding noiseless sector updates.

\begin{definition}[Quantum lift of an SPCA rule]
\label{def:spca-lift}
For \(\mathbf a:\Sigma_X\to\Omega_X\), the \emph{quantum lift} of
\(\mathbf a\) is the CPTP map
\begin{equation}
    \widehat{\mathbf a}^X_*(\rho)
    :=
    \sum_{s\in\Sigma_X}
    Z^{\mathbf a(s)}\,\Pi^X_s\,\rho\,\Pi^X_s\,Z^{\mathbf a(s)} .
\end{equation}
For \(\mathbf a:\Sigma_Z\to\Omega_Z\), the lift
\(\widehat{\mathbf a}^Z_*\) is defined symmetrically, with
\(X^{\mathbf a(s)}\) and \(\Pi^Z_s\).
\end{definition}

\begin{lemma}[Properties of the lift]
\label{lem:spca-lift}
Let \(\mathbf a:\Sigma_X\to\Omega_X\).
\begin{enumerate}[label=(\roman*),leftmargin=*]
    \item \(\widehat{\mathbf a}^X_*\) is a syndrome-based
    quantum channel, with induced transition probabilities
    \(p_X(\mathbf u\mid s)=\delta_{\mathbf u,\mathbf a(s)}\) and
    \(p_Z(\mathbf u\mid s)=\delta_{\mathbf u,\mathbf 0}\).
    \item Suppose that each \(a_i\) depends only on the syndrome bits of
    the \(X\)-checks supported in \(N_R(i)\). Then
    \(\widehat{\mathbf a}^X_*\) is \(\max(R,1)\)-local.
\end{enumerate}
The same statements hold with \(X\) and \(Z\) exchanged.
\end{lemma}
Proof deferred to Appendix~\ref{app:sync}.

\begin{lemma}[Sector dynamics of the lifted channel]
\label{lem:sync-spca-run}
Let \(\mathbf a^{(X)}:\Sigma_X\to\Omega_X\) and
\(\mathbf a^{(Z)}:\Sigma_Z\to\Omega_Z\), let
\(\widehat{\mathbf a}^{(X)}_*\) be the \(X\)-sector lift of
\(\mathbf a^{(X)}\) and \(\widehat{\mathbf a}^{(Z)}_*\) the \(Z\)-sector
lift of \(\mathbf a^{(Z)}\) (Definition~\ref{def:spca-lift}), let
\(\mathcal N_*\) be a synchronous Pauli noise layer, and let
\(
    \mathcal Q_*
    :=
    \mathcal N_*\circ\widehat{\mathbf a}^{(Z)}_*\circ\widehat{\mathbf a}^{(X)}_*
\).
Then \(\mathcal Q_*\) is a syndrome-based quantum channel,
with
\begin{equation}\label{eq:sync-spca-kernel}
    p_X(\mathbf u\mid s)
    =
    \pi_X\bigl(\mathbf u\oplus\mathbf a^{(X)}(s)\bigr),
    \qquad
    \pi_X(\mathbf v)
    :=
    \prod_{i=1}^n
    \bigl(\theta_i^{(X)}\bigr)^{v_i}
    \bigl(1-\theta_i^{(X)}\bigr)^{1-v_i},
\end{equation}
and symmetrically for the \(Z\) sector. Consequently the induced
\(X\)-sector error process has the update rule
\begin{equation}\label{eq:sync-spca-recursion}
    \xi_{r+1}^{(X)}
    =
    \xi_r^{(X)}
    \oplus
    \mathbf a^{(X)}\bigl(H_X\xi_r^{(X)}\bigr)
    \oplus
    \mathbf n_r,
    \qquad
    \xi_0^{(X)}=\mathbf 0,
\end{equation}
where the coordinates \(n_{r,i}\) are independent Bernoulli variables
with means \(\theta_i^{(X)}\), independent of
\(\xi_0^{(X)},\ldots,\xi_r^{(X)}\) and of
\(\mathbf n_0,\ldots,\mathbf n_{r-1}\).
\end{lemma}
Proof deferred to Appendix~\ref{app:sync}.

\begin{proof}[Proof of Corollary~\ref{cor:sync-spca-memory}]
Let \(\mathcal Q_*:=\mathcal N_*\circ\widehat{\mathbf a}^{(Z)}_*\circ\widehat{\mathbf a}^{(X)}_*\);
it is syndrome-based by Lemma~\ref{lem:sync-spca-run}. Distances in the
interaction graph of the CSS code are at most those in the interaction graph
of \(\mathcal C_\sigma\), so Lemma~\ref{lem:spca-lift}(ii) applies with
\(R=R_\sigma\), and \(\widehat{\mathbf a}^{(\sigma)}_*\) is
\(\max(R_\sigma,1)\)-local. Since \(\mathcal N_*\) is \(0\)-local,
\(\mathcal Q_*\) is \(R\)-local with \(R:=\max(R_X,1)+\max(R_Z,1)\),
independent of \(L\).

Set \(p:=2\eta<\pcth\). Since \(\theta_{P,i}\le\eta\), we have
\(\theta_i^{(X)}=\theta_{Z,i}+\theta_{Y,i}\le p\). By
Lemma~\ref{lem:sync-spca-run}, the induced \(X\)-sector error process
follows Equation~\eqref{eq:sync-spca-recursion}. Realize the noise vectors
as follows: draw \(F_r\) by including each site independently with
probability \(p\), independently of the past, and given \(F_r\) set
\(n_{r,i}=1\) with probability \(\theta_i^{(X)}/p\), independently for
\(i\in F_r\), and \(n_{r,i}=0\) for \(i\notin F_r\) (for \(p=0\), take
\(F_r=\emptyset\) and \(\mathbf n_r=\mathbf 0\)). The coordinates of
\(\mathbf n_r\) are then independent with means \(\theta_i^{(X)}\) and
independent of the past, as Lemma~\ref{lem:sync-spca-run} requires, and
\(\operatorname{supp}\mathbf n_r\subseteq F_r\). Hence
\(\zeta_r:=\mathbf n_r\) is an admissible fault vector, and the induced
\(X\)-sector process is a run of \(\Psi_X\) at fault rate \(p\), started at
\(\mathbf 0\). By the memory hypothesis
(Definition~\ref{def:syndrome-pca-classical}), it \(q_X(m,L)\)-remembers
\(b_X\) through round \(m\) (Definition~\ref{def:classical-memory}). The
\(Z\) sector is identical, with
\(\theta_i^{(Z)}=\theta_{X,i}+\theta_{Y,i}\le p\). 
Theorem~\ref{thm:main-sync} then shows that the encoded system is a
quantum \(\bigl(q_X(m,L)+q_Z(m,L)\bigr)\)-memory through round \(m\).

Finally, if the code pair, both SPCAs and \(\mathcal N_*\) are translationally
invariant, translations permute the syndrome projectors and the Pauli strings
\(Z^{\mathbf a^{(X)}(s)}\), \(X^{\mathbf a^{(Z)}(s)}\) compatibly, so each
factor of \(\mathcal Q_*\), and hence \(\mathcal Q_*\) itself, commutes with
translations.
\end{proof}

\begin{remark}[Relation to the continuous-time construction]
The synchronous model treats one application of the local channel as an
ideal clocked update round.  It specifies the causal locality and sector
dynamics of the round without resolving that round into microscopic
continuous-time operations.
Section~\ref{sec:periodic} addresses this additional implementation problem.
There we construct a geometrically local periodically driven
Lindbladian whose local operations occur stochastically and whose code
qubits and classical registers are noisy, while retaining a global periodic
Reset--Read--Execute schedule.
\end{remark}

\section{Application: the 4D Toom rule}
\label{sec:toom}

We apply Corollary~\ref{cor:sync-spca-memory} to the
four-dimensional toric code equipped with the synchronous Toom rule.
Since that corollary takes only classical input, the whole section is
classical: we verify its hypotheses for the two component codes of the
4D toric code, and the quantum statement follows by invoking it.

The argument proceeds in three steps. First, we verify that the 4D Toom
rules define valid SPCAs for the two component codes
(Subsection~\ref{subsec:toom-geometry}). Second, we invoke
Ahn's~\cite[\S 7]{ahn_extending_2004} threshold analysis, which shows that
below a finite fault threshold the probability of a homologically
nontrivial decoded error at each round is exponentially small in the
linear system size, and deduce from it the retention hypothesis of 
Corollary~\ref{cor:sync-spca-memory}
(Subsection~\ref{subsec:toom-threshold}). Third, we invoke 
Corollary~\ref{cor:sync-spca-memory} to prove 
Theorem~\ref{thm:toom-sync-formal}. By self-duality, it suffices to describe
the $X$-sector explicitly; the $Z$-sector follows by the same
construction on the dual lattice.

The Toom rule is a deterministic local correction rule originally
defined on a $2$-dimensional lattice of classical spins, where each
site checks the values of its positively-oriented neighbors and flips
if all of them differ\cite{toom_stable_1980}. The asymmetry in the choice of positive
directions breaks detailed balance and gives the rule its
noise-suppressing properties. The 4D version we use here differs from
the original in a crucial way: rather than conditioning on neighboring
spin values directly, it conditions only on the syndrome --- the parity
of errors on boundary edges of each plaquette. Paired with a
stochastic-adversarial fault model, this makes it an SPCA in the sense
of Definition~\ref{def:syndrome-pca-classical}
(Lemma~\ref{lem:toom-pca-validation}), which is what allows Corollary~\ref{cor:sync-spca-memory} to be applied. In four dimensions,
plaquettes play the role of sites, edges carry the syndrome, and the
rule checks the two positively-oriented boundary edges of each
plaquette.

\subsection{Geometry and the Toom predicate}\label{subsec:toom-geometry}

We first fix the lattice geometry and the classical component codes,
then define the Toom correction rule on top of them.

\begin{definition}[4D toric code]
\label{def:4d-toric-geometry}
Let $\Lambda=(\mathbb Z/L\mathbb Z)^4$ be the four-dimensional
hypercubic lattice with periodic boundary conditions, and let
$\Delta_m(\Lambda)$ denote the set of $m$-cells. Bits reside on
plaquettes, which are also the sites of the SPCAs below; we index them
by $i\in\Delta_2(\Lambda)$, so that $n=|\Delta_2(\Lambda)|$ and
$\Omega_X=\Omega_Z=\mathbb F_2^{\Delta_2(\Lambda)}$. The $X$-checks are
associated with edges $e\in\Delta_1(\Lambda)$ and the $Z$-checks with
three-cubes $c\in\Delta_3(\Lambda)$: for
$\mathbf x\in\mathbb F_2^{\Delta_2(\Lambda)}$,
\begin{equation}
    [H_X\mathbf x]_e=\sum_{i\supset e}\mathbf x_i \pmod 2,
    \qquad
    [H_Z\mathbf x]_c=\sum_{i\subset c}\mathbf x_i \pmod 2 .
\end{equation}
The 4D toric code is the CSS code of Definition~\ref{def:css-code} with
these parity-check matrices; its stabilizers are
$S_e^X=\prod_{i\supset e}X_i$ and $S_c^Z=\prod_{i\subset c}Z_i$.
\end{definition}

The map $H_X$ is the mod-$2$ boundary map $\partial_2$ from $2$-chains
to $1$-chains, and $H_Z^T$ is the boundary map $\partial_3$ from
$3$-chains to $2$-chains, so the CSS condition $H_XH_Z^T=0$ is
$\partial_2\partial_3=0$. Hence $\mathcal C_X=\ker H_X$ is the space of
$2$-cycles, $\operatorname{row}(H_Z)$ is the space of $2$-boundaries,
and the $X$-logical classes of Definition~\ref{def:logical-values} are
the homology classes,
$\mathcal C_X/\operatorname{row}(H_Z)=H_2(\Lambda;\mathbb F_2)\cong\mathbb F_2^{6}$.
Dually, $\mathcal C_Z/\operatorname{row}(H_X)=H^2(\Lambda;\mathbb F_2)\cong\mathbb F_2^{6}$.
In both sectors we call $\mathbf c\in\mathcal C_\sigma$ homologically
nontrivial if $[\mathbf c]_\sigma\ne[\mathbf 0]_\sigma$.

With the geometry and component codes fixed, we now specify the
correction rule itself: the Toom predicate that each site evaluates
from the syndrome.

\begin{definition}[4D Toom predicate]
\label{def:toom-predicate}
Fix the standard coordinate orientation of $\Lambda$. If a plaquette
$i$ lies in the coordinate plane spanned by $(\hat u_\mu,\hat u_\nu)$,
let $e^+_{i,\mu}$ and $e^+_{i,\nu}$ be its two positively oriented
boundary edges. For $s\in\Sigma_X$, define the $X$-sector Toom
predicate by
\begin{equation}
    a_i^{(X)}(s)=s_{e^+_{i,\mu}}\,s_{e^+_{i,\nu}}.
\end{equation}
Equivalently, $a_i^{(X)}(s)=1$ if and only if both positively oriented
boundary edges of $i$ carry syndrome. The $Z$-sector predicate
$a_i^{(Z)}$ is defined by the same construction on the dual lattice.
\end{definition}

\begin{lemma}[4D Toom rules are valid SPCAs]
\label{lem:toom-pca-validation}
Equip $\Delta_2(\Lambda)$ with the graph metric in which two plaquettes
are adjacent when they share an edge. For every fault rate $p\in[0,1]$,
the predicates $\mathbf a^{(X)}=(a_i^{(X)})_{i\in\Delta_2(\Lambda)}$ of
Definition~\ref{def:toom-predicate} define a translationally invariant
$1$-local SPCA $\Psi_X$ for $\mathcal C_X$ at stochastic-adversarial
fault rate $p$ (Definition~\ref{def:syndrome-pca-classical}); dually,
$\mathbf a^{(Z)}$ defines such an SPCA $\Psi_Z$ for $\mathcal C_Z$.
\end{lemma}

\begin{proof}
We verify the $X$-sector; the $Z$-sector is dual. The update on
plaquette $i$ is $a_i^{(X)}(H_X\mathbf x)$, which depends only on the
syndrome bits of the checks $e^+_{i,\mu}$ and $e^+_{i,\nu}$. The check of an edge $e$ of $i$ is supported on the plaquettes
containing $e$, each of which is acted on together with $i$ by $S_e^X$,
and hence lies in $N_1(i)$. Dually, $a_i^{(Z)}$ reads two three-cubes
containing $i$, whose supports lie in $N_1(i)$ by the same argument
with $S_c^Z$. This is the syndrome-based update condition
with $R=1$, and the fault model is the stochastic-adversarial model of
Definition~\ref{def:syndrome-pca-classical} by construction. The
predicate refers only to the fixed coordinate orientation, so it
commutes with lattice translations.
\end{proof}

Lemma~\ref{lem:toom-pca-validation} gives the SPCA structure required
by Corollary~\ref{cor:sync-spca-memory}; the retention hypothesis
is supplied by Ahn's estimate in Subsection~\ref{subsec:toom-threshold}.

\subsection{Exponential memory threshold}\label{subsec:toom-threshold}

We now use Ahn's classical threshold estimate for the noisy synchronous
$4$D Toom PCA~\cite[§7.2.3]{ahn_extending_2004}. Ahn's argument is
formulated in terms of spacetime locations at which the prescribed
local update rule is violated. The same explanation-tree estimate
applies to the stochastic-adversarial model of
Definition~\ref{def:syndrome-pca-classical}: the set of actual
violations is contained in an i.i.d.\ Bernoulli fault set, so any
prescribed set of $k$ violating spacetime locations occurs with
probability at most $p^k$. We take the readout decoders
$\mathsf{Dec}_X,\mathsf{Dec}_Z$ of Definition~\ref{def:readout-decoder}
to be the decoders of the recovery step in Ahn's analysis.

\begin{proposition}[Ahn's encoded-step failure bound]
\label{prop:ahn-suppression}
Fix $\sigma\in\{X,Z\}$, let $(\mathbf e_r^{(\sigma)})_{r\ge0}$
be any run of the 4D Toom SPCA $\Psi_\sigma$ of
Lemma~\ref{lem:toom-pca-validation} at fault rate $p$, started at
$\mathbf e_0^{(\sigma)}=\mathbf 0$, and let
\begin{equation}
    \tau_\sigma
    :=
    \min\Bigl\{r\ge1:
    \bigl[\mathbf e_r^{(\sigma)}\oplus
    \mathsf{Dec}_\sigma\bigl(H_\sigma\mathbf e_r^{(\sigma)}\bigr)\bigr]_\sigma
    \ne[\mathbf 0]_\sigma\Bigr\}
\end{equation}
be the first round at which the decoded string is homologically
nontrivial. There exist constants $\pcth,c,C>0$, independent of $L$,
such that for every $p<\pcth$ and every $r\ge1$,
\begin{equation}
    \Pr\bigl[\tau_\sigma=r\mid \tau_\sigma\ge r\bigr]\le Ce^{-cL}.
\end{equation}
\end{proposition}

The important point is that Proposition~\ref{prop:ahn-suppression} is a
one-step estimate, rather than a bound on the logical error probability
that is uniform for arbitrarily long times. On a finite system, encoded
failures eventually accumulate. What Ahn's estimate provides is an
exponentially small hazard for such a failure at each round, as long as
the encoded information has survived the preceding rounds.
Consequently, over $m$ rounds the probability that an encoded failure
has occurred at least once grows at most linearly in $m$, giving a
bound of order $me^{-cL}$.

To connect this encoded-failure event to the retention hypothesis of
Corollary~\ref{cor:sync-spca-memory}, recall that readout is
determined by the logical-error indicator $b_\sigma$, which reads only
the component of the decoded class selected by $\mu_\sigma$. A label
error therefore requires a homologically nontrivial decoded string. The
converse fails, since $\mu_\sigma$ vanishes on a five-dimensional
subspace of the logical classes.

\begin{lemma}[A label error implies a homologically nontrivial decoded string]
\label{lem:toom-homological-label}
For $\sigma\in\{X,Z\}$ and every $\mathbf x\in\Omega_\sigma$,
\begin{equation}
    b_\sigma(\mathbf x)=1
    \quad\Longrightarrow\quad
    \bigl[\mathbf x\oplus\mathsf{Dec}_\sigma(H_\sigma\mathbf x)\bigr]_\sigma
    \ne[\mathbf 0]_\sigma .
\end{equation}
\end{lemma}

\begin{proof}
By Definition~\ref{def:decoded-logical-value},
$b_\sigma(\mathbf x)=\mu_\sigma\bigl([\mathbf x\oplus\mathsf{Dec}_\sigma(H_\sigma\mathbf x)]_\sigma\bigr)$,
and $\mu_\sigma([\mathbf 0]_\sigma)=0$ since $\mu_\sigma$ is linear
(Definition~\ref{def:logical-values}).
\end{proof}

Combining this observation with Ahn's exponentially small per-round
failure probability gives the retention hypothesis of 
Corollary~\ref{cor:sync-spca-memory}.

\begin{corollary}[Classical memory of the 4D Toom SPCAs]
\label{cor:toom-classical-memory}
For every $p<\pcth$ and every $m\ge0$, the 4D Toom SPCAs $\Psi_X$ and
$\Psi_Z$ at fault rate $p$, started at $\mathbf 0$, $Cme^{-cL}$-remember
$b_X$ and $b_Z$, respectively, through round $m$
(Definition~\ref{def:classical-memory}).
\end{corollary}

\begin{proof}
Fix $\sigma\in\{X,Z\}$ and $1\le r\le m$. By
Proposition~\ref{prop:ahn-suppression},
\begin{equation}
    \Pr[\tau_\sigma\le r]
    =
    \sum_{j=1}^{r}
    \Pr[\tau_\sigma=j\mid \tau_\sigma\ge j]\,\Pr[\tau_\sigma\ge j]
    \le
    rCe^{-cL}
    \le
    mCe^{-cL}.
\end{equation}
Lemma~\ref{lem:toom-homological-label}, together with
$b_\sigma(\mathbf e_0^{(\sigma)})=b_\sigma(\mathbf 0)=0$ and the
definition of $\tau_\sigma$, gives
\begin{equation}
    \bigl\{b_\sigma(\mathbf e_r^{(\sigma)})\ne b_\sigma(\mathbf e_0^{(\sigma)})\bigr\}
    \subseteq
    \Bigl\{\bigl[\mathbf e_r^{(\sigma)}\oplus
    \mathsf{Dec}_\sigma\bigl(H_\sigma\mathbf e_r^{(\sigma)}\bigr)\bigr]_\sigma
    \ne[\mathbf 0]_\sigma\Bigr\}
    \subseteq
    \{\tau_\sigma\le r\},
\end{equation}
so
$\Pr\bigl[b_\sigma(\mathbf e_r^{(\sigma)})\ne b_\sigma(\mathbf e_0^{(\sigma)})\bigr]\le mCe^{-cL}$.
For $r=0$ this probability vanishes.
\end{proof}

Lemma~\ref{lem:toom-pca-validation} establishes the SPCA structure,
Proposition~\ref{prop:ahn-suppression} supplies the quantitative
threshold, and Lemma~\ref{lem:toom-homological-label} translates that
threshold into the retention hypothesis of 
Corollary~\ref{cor:sync-spca-memory}. Invoking that corollary proves
 Theorem~\ref{thm:toom-sync-formal}.

\begin{proof}[Proof of Theorem~\ref{thm:toom-sync-formal}]
We verify the hypotheses of Corollary~\ref{cor:sync-spca-memory}.
By Lemma~\ref{lem:toom-pca-validation}, \(\Psi_X\) and \(\Psi_Z\) are
SPCAs for the component codes of
Definition~\ref{def:4d-toric-geometry}. Their locality in the
interaction graph of the 4D toric code is as follows. The predicate
\(a_i^{(X)}\) reads the checks of the two edges
\(e^+_{i,\mu},e^+_{i,\nu}\) of \(i\); the check of an edge \(e\) of \(i\)
is supported on the plaquettes containing \(e\), each of which is acted
on together with \(i\) by \(S_e^X\), and hence lies in \(N_1(i)\).
Dually, \(a_i^{(Z)}\) reads the checks of two three-cubes containing
\(i\), whose supports lie in \(N_1(i)\) by the same argument with
\(S_c^Z\). Thus \(R_X=R_Z=1\), and the channel is
\(R'=2\)-local. The Toom predicates commute with lattice translations,
and hence so do their lifts.

Fix a designated logical qubit, with logical values \(\mu_X,\mu_Z\) as
in Definition~\ref{def:logical-coordinate-identification} and the
readout decoders chosen above. Since
Lemma~\ref{lem:toom-homological-label} uses only the linearity of
\(\mu_\sigma\), Corollary~\ref{cor:toom-classical-memory} applies to
this choice: for every \(p<\pcth\), both SPCAs \(Cme^{-cL}\)-remember
their logical-error indicators through round \(m\).

Set \(\eta_c:=\pcth/2\). If \(\theta_{P,i}\le\physrate<\eta_c\) for all
\(P\) and \(i\), then the sector flip probabilities satisfy
\(\theta_i^{(\sigma)}\le2\physrate=:p<\pcth\). 
Corollary~\ref{cor:sync-spca-memory} then gives
\begin{equation}
    \min_{0\le r\le m}
    F_{\min}\!\left[\mathcal Q_{\mathrm{Toom}*}^r\right]
    \ge 1-2Cme^{-cL}.
\end{equation}

\end{proof}

\begin{remark}[Autonomous version]
\label{rem:toom-lindbladian}
The same classical input also satisfies the hypotheses of 
Theorem~\ref{thm:periodic-formal}: the plaquette adjacency graph has
degree bounded independently of \(L\), and every check has weight
\(6\). That theorem therefore gives, in addition, a geometrically
local, translationally invariant periodic Lindbladian, with period
independent of \(L\), such that below a nonzero Pauli-noise rate the
encoded system is a quantum \(2Cme^{-cL}\)-memory at cycle boundaries
through cycle \(m\).
\end{remark}

\section{Periodic Lindbladians from synchronous SPCAs}
\label{sec:periodic}

Corollary~\ref{cor:sync-spca-memory} shows how a pair of local SPCAs with classical memory can be lifted directly to a synchronous quantum memory.  This is already useful, since most classical cellular-automaton memory constructions and threshold results are formulated in the synchronous setting.  From the quantum perspective, however, continuous-time Lindbladian dynamics is the more natural model for local many-body evolution.  The next question is therefore whether the same synchronous classical memories can also be used to construct quantum memories in continuous time.

In the synchronous construction, one application of the quantum channel implements one complete SPCA round: all local predicates are evaluated on the same input configuration, and the corresponding corrections are applied within a single synchronous update.  Passing from this dynamics to continuous time is not automatic.  Under a local Lindbladian, the elementary operations occur at random times rather than as one simultaneous update, and replacing the synchronous SPCA rule by local syndrome-conditioned jumps generally changes the resulting classical dynamics.  In particular, an operation occurring at one site may modify the syndrome before another site acts.  Thus a classical memory theorem for the synchronous SPCA cannot simply be applied to the continuous-time sector process.

Theorem~\ref{thm:periodic-formal} gives a partial resolution of this problem.  For any pair of local SPCAs satisfying the memory assumptions of Corollary~\ref{cor:sync-spca-memory}, we construct a geometrically local periodic Lindbladian whose quantum memory inherits the same classical retention bounds at the boundaries of its driving cycles.  The construction is still periodically driven, and hence retains a global temporal structure; it does not solve the stronger problem of realizing the SPCA memory with a time-independent Lindbladian.  It does, however, replace the ideal synchronous local update by continuous-time stochastic local dynamics, including noise on both the data and the auxiliary degrees of freedom.

The central statement needed for this bridge is classical.  Each sector of the constructed Lindbladian induces a continuous-time Markov process.  Denoting its data component by \(\xi_t^{(\sigma)}\), we sample it at the cycle boundaries,
\[
    \hat\xi_r^{(\sigma)}
    :=
    \xi_{r\Delta}^{(\sigma)} .
\]
We prove that the sampled process \((\hat\xi_r^{(\sigma)})_{r\ge0}\) is itself a run of the original SPCA \(\Psi_\sigma\), at an effective fault rate \(p_{\mathrm{eff}}\) that can be made arbitrarily small by choosing the implementation parameters appropriately and taking the physical noise sufficiently weak.  Once this is established, the SPCA memory threshold applies directly to the sampled sector process, and the quantum--classical reduction converts the resulting classical retention bound into the fidelity bound of Theorem~\ref{thm:periodic-formal}.

To formulate this construction, we first extend the syndrome-based reduction to Lindbladians acting on the data together with local classical registers.  We then construct the periodic dynamics and prove the sampled-process statement above.

The argument proceeds in four steps.

First, we extend the continuous-time quantum--classical reduction to
syndrome-based Lindbladians acting on the data together with local classical
registers (Subsection~\ref{subsec:register-extension}).  Each sector of the
augmented dynamics is then represented by a classical continuous-time Markov
process whose state contains both the data error configuration and the register
configuration.

Second, we construct the periodic Lindbladian from the two SPCA update rules
(Subsection~\ref{subsec:quantum-lift}).  The construction uses local auxiliary
registers to implement the synchronous correction rule through continuous-time
local dynamics, and we verify that the resulting Lindbladian belongs to the
register-extended syndrome-based class introduced in the first step.

Third, we prove the central classical statement of the construction
(Subsection~\ref{subsec:periodic-equivalence}).  For each sector, we sample the data
component of the induced continuous-time process at the correction-cycle
boundaries and show that the resulting discrete-time process is a valid run of
the original SPCA at an effective fault rate.  The possible deviations from the
ideal SPCA update are confined to rare, locally correlated failure sets, which
can be dominated cycle by cycle by the fresh independent fault sets allowed by
the SPCA noise model.

Finally, we combine this identification with the assumed memory property of the
input SPCAs (Subsection~\ref{subsubsec:memory-bound}).  Since the sampled sector
processes are admissible SPCA runs below the fault threshold, they inherit the
corresponding classical logical-retention bounds.  The classical-register
extension of the quantum--classical reduction then converts these bounds into
the quantum fidelity estimate of 
Theorem~\ref{thm:periodic-formal}.

\subsection{Syndrome-based Lindbladians with classical registers}
\label{subsec:register-extension}

The construction of this section acts on the data qubits together with
auxiliary qubits holding classical bits, and some of its jumps flip a
data qubit and a register at once. Such generators lie outside
Definition~\ref{def:syndrome-based-lindbladian}, Proposition~\ref{prop:reduction-async-formal}
extends to them without a new idea. We state the extension here; the
proofs are given in Appendix~\ref{app:registers}.

For each sector $\sigma\in\{X,Z\}$, let $\mathsf{Reg}_\sigma$ be a finite set of
register qubits, with $\mathsf{Reg}_X\cap\mathsf{Reg}_Z=\varnothing$, and let
$\mathcal H_{\mathrm{ext}}:=\mathcal H\otimes
(\mathbb C^2)^{\otimes(\mathsf{Reg}_X\sqcup\mathsf{Reg}_Z)}$. We call the $n$ qubits of $\mathcal{H}$ the data qubits to separate from the new register qubits. 
The registers
are read in the computational basis, so we adjoin their $Z$ operators to
the stabilizer algebras,
\begin{equation}
    \mathcal A^{\mathrm{ext}}
    :=
    \mathrm{Alg}\bigl(
        \mathcal A\cup\{Z_j:j\in\mathsf{Reg}_X\sqcup\mathsf{Reg}_Z\}
    \bigr),
    \qquad
    \mathcal A_\sigma^{\mathrm{ext}}
    :=
    \mathrm{Alg}\bigl(
        \mathcal A_\sigma\cup\{Z_j:j\in\mathsf{Reg}_\sigma\}
    \bigr),
\end{equation}
and enlarge the sector configuration spaces to
\begin{equation}
    \Omega_\sigma^{\mathrm{ext}}
    :=
    \Omega_\sigma\times\Omega_{\mathsf{Reg}_\sigma},
    \qquad
    \Omega_{\mathsf{Reg}_\sigma}
    :=
    \{0,1\}^{\mathsf{Reg}_\sigma}.
\end{equation}

Since $\mathcal A_\sigma$ is Abelian and the $Z_j$ act on distinct register
qubits, $\mathcal A_\sigma^{\mathrm{ext}}$ is Abelian. Its joint eigenspaces are
labeled by a data syndrome $s\in\Sigma_\sigma$ and a register configuration
$\mathbf c\in\Omega_{\mathsf{Reg}_\sigma}$: the eigenspace $(s,\mathbf c)$ consists of
the states with data syndrome $s$ and the registers of $\mathsf{Reg}_\sigma$ in
configuration $\mathbf c$, with the registers of the other sector arbitrary. Its
projector is $\Pi_s^\sigma\otimes|\mathbf c\rangle\langle\mathbf c|$, where
$|\mathbf c\rangle\langle\mathbf c|$ acts on $\mathsf{Reg}_\sigma$ and the identity
on the other registers. Every $F\in\mathcal A_\sigma^{\mathrm{ext}}$ acts on each
of them as a scalar, which we denote $F(s,\mathbf c)$, extending
Equation~\eqref{def:eigenvalue-functions}.

For a single-qubit data Pauli $P_i$, $P\in\{I,X,Y,Z\}$, and a Pauli string
$\Theta$ on the registers, write $\Theta=i^{k}X^{\mathbf v}Z^{\mathbf w}$ with
$k\in\{0,1,2,3\}$ and
$\mathbf v=(\mathbf v_X,\mathbf v_Z),\ \mathbf w=(\mathbf w_X,\mathbf w_Z)\in
\Omega_{\mathsf{Reg}_X}\times\Omega_{\mathsf{Reg}_Z}$, in the notation of
Definition~\ref{def:pauli-readout-corrections}. Only $\mathbf v$ changes the
register configuration: $X^{\mathbf v}$ flips the registers $j$ with $v_j=1$,
while $Z^{\mathbf w}$ only multiplies computational basis states by signs. The
\emph{$\sigma$-flip} of $P_i\Theta$ is the pair
$(\mathbf u,\mathbf v_\sigma)\in\Omega_\sigma\times\Omega_{\mathsf{Reg}_\sigma}$, where
$\mathbf u=\hat e_i$ if $P\in\{Z,Y\}$ for $\sigma=X$, or if $P\in\{X,Y\}$ for
$\sigma=Z$, and $\mathbf u=\mathbf 0$ otherwise.

\begin{definition}[Register-extended syndrome-based Lindbladian]
\label{def:classical-register-extension}
A Lindbladian on $\mathcal H_{\mathrm{ext}}$ is a
\emph{register-extended syndrome-based Lindbladian} if it has the form of
Definition~\ref{def:syndrome-based-lindbladian} with the following changes.
\begin{enumerate}[label=(\roman*)]
    \item The Hamiltonian $H(t)$ is a Hermitian element of
    $\mathcal A^{\mathrm{ext}}$.
        \item Each jump $P_if_{P,i,\alpha}(t)$ becomes
    $P_i\Theta_{P,i,\alpha}f_{P,i,\alpha}(t)$, where $\Theta_{P,i,\alpha}$ is a
    Pauli string on the registers and the syndrome function
    $f_{P,i,\alpha}(t)$ lies in $\mathcal A^{\mathrm{ext}}$. For $P=I$ we
    have $P_i=\mathbb I$, so the jump acts on no data qubit and the label
    $i$ serves only to index it; for instance, each register operation of
    Subsection~\ref{subsec:quantum-lift} is labeled by the site it
    serves.
    \item The flip-rate operators are indexed by the sector flip,
    \begin{equation}\label{eq:extended-flip-rates}
        \Gamma^{(\sigma)}_{\mathbf u,\mathbf v}(t)
        :=
        \sum_{(P,i,\alpha)\in J^{(\sigma)}_{\mathbf u,\mathbf v}}
        f_{P,i,\alpha}(t)^\dagger f_{P,i,\alpha}(t),
        \qquad
        (\mathbf u,\mathbf v)\ne(\mathbf 0,\mathbf 0),
    \end{equation}
        where $J^{(\sigma)}_{\mathbf u,\mathbf v}$ is the set of indices
    $(P,i,\alpha)$ of jumps whose $\sigma$-flip is $(\mathbf u,\mathbf v)$.
    A jump with index $(P,i,\alpha)$ has $\mathbf u=\hat e_i$ if $P$ flips
    qubit $i$ in sector $\sigma$, and $\mathbf u=\mathbf 0$ otherwise; in
    particular, a jump with $P=I$ has $\mathbf u=\mathbf 0$ whatever its
    label $i$.
    \item The flip-rate operators satisfy the \emph{extended rate condition}
    \begin{equation}\label{eq:extended-rate-condition}
        \Gamma^{(\sigma)}_{\mathbf u,\mathbf v}(t)\in\mathcal A_\sigma^{\mathrm{ext}}
        \qquad\text{for all }(\mathbf u,\mathbf v)\ne(\mathbf 0,\mathbf 0)
        \text{ and }t\ge0.
    \end{equation}
\end{enumerate}
\end{definition}

Its Heisenberg propagator $\mathcal W_{\mathrm{ext}}^*(t,0)$ is defined as in
Equation~\eqref{eq:q-propgator}.  Definition~\ref{def:syndrome-based-lindbladian} is the case
$\mathsf{Reg}_X=\mathsf{Reg}_Z=\varnothing$: every $\Theta$ is trivial, the only
flips are $(\hat{e}_i,\mathbf 0)$, and $\Gamma^{(\sigma)}_{\hat{e}_i,\mathbf 0}=\Gamma^{(\sigma)}_i$.

Pauli noise on the registers, extending Definition~\ref{def:pauli-noise}, is an
instance with scalar syndrome functions (since it is oblivious to the syndrome):
\begin{equation}\label{eq:extended-noise}
    \mathcal L_{\mathrm{noise}}^{\mathrm{ext}}
    :=
    \mathcal L_{\mathrm{noise}}
    +
    \sum_{j\in\mathsf{Reg}_X\sqcup\mathsf{Reg}_Z}
    \sum_{P\in\{X,Y,Z\}}
    \gamma_{P,j}^{\mathrm{reg}}\,
    \mathcal D[P_j],
    \qquad
    \gamma_{P,j}^{\mathrm{reg}}\ge0.
\end{equation}
The right term is the register noise. A register qubit $j\in\mathsf{Reg}_\sigma$ is flipped by $X_j$ and
$Y_j$ (regardless of $\sigma$ because the register qubits are classical), so the register noise term contributes $\nu_j\mathbb I$ to
$\Gamma^{(\sigma)}_{\mathbf 0,e_j}$, with
\begin{equation}\label{eq:register-noise-rates}
    \nu_j
    :=
    \gamma_{X,j}^{\mathrm{reg}}
    +
    \gamma_{Y,j}^{\mathrm{reg}}.
\end{equation}

The induced \emph{sector processes} are the continuous-time Markov chains on
$\Omega_\sigma^{\mathrm{ext}}$ , initialized at
$(\mathbf 0,\mathbf 0)$ with transitions
$(\mathbf e,\mathbf c)\mapsto(\mathbf e\oplus\mathbf u,\mathbf c\oplus\mathbf v)$
at rates
\begin{equation}\label{eq:classical-register-rate}
    r^{(\sigma)}_{\mathbf u,\mathbf v}(\mathbf e,\mathbf c,t)
    :=
    \Gamma^{(\sigma)}_{\mathbf u,\mathbf v}(t)(H_\sigma\mathbf e,\mathbf c)
    \ge0,
\end{equation}
which are well defined by the extended rate condition~\eqref{eq:extended-rate-condition}. We denote them by $(\xi_t^{(\sigma)},\mathbf c_t^{(\sigma)})$.
Their backward generators $G_\sigma^{\mathrm{ext}}(t)$ and propagators
$M_\sigma^{\mathrm{ext}}(t,t')$ are defined as in
Definition~\ref{def:ctmc-generators}, with the sum over single-bit flips
replaced by the sum over sector flips $(\mathbf u,\mathbf v)$. The registers are
initialized in the state $|\mathbf 0\rangle$, independently of the encoded data,
and we write
\begin{equation}
    \rho_{\mathrm{code}}^{\mathrm{ext}}
    :=
    \rho_{\mathrm{code}}\otimes|\mathbf 0\rangle\langle\mathbf 0|_{\mathrm{reg}},
    \qquad
    \langle O_1,O_2\rangle_{\mathrm{code}}^{\mathrm{ext}}
    :=
    \operatorname{Tr}\!\left(
        \rho_{\mathrm{code}}^{\mathrm{ext}}O_1^\dagger O_2
    \right),
\end{equation}
where operators on $\mathcal H$ are extended by the identity on the registers.
We fix $|\mathbf 0\rangle$ for definiteness; for any computational-basis state
$|\mathbf c\rangle$, $\mathbf c=(\mathbf c_X,\mathbf c_Z)$, the same statements
hold with the sector-$\sigma$ process started at $(\mathbf 0,\mathbf c_\sigma)$.

\begin{proposition}[Classical-register extension of Proposition~\ref{prop:reduction-async-formal}]
\label{prop:classical-register-reduction}
Let
\(\mathcal L_{\mathrm{ext}}(t)\)
be a register-extended syndrome-based Lindbladian with
Heisenberg propagator
\(\mathcal W_{\mathrm{ext}}^*(t,0)\).
For each sector \(\sigma\in\{X,Z\}\), let
\[
    \bigl(\xi_t^{(\sigma)},\mathbf c_t^{(\sigma)}\bigr)_{t\ge0}
\]
be the induced sector process. Then, for every \(\sigma\in\{X,Z\}\) and every \(t\ge0\),
\begin{align}
    \left\langle
        T_\sigma,
        \mathcal W_{\mathrm{ext}}^*(t,0)(T_\sigma)
    \right\rangle_{\mathrm{code}}^{\mathrm{ext}}
    &=
    \mathbb E\!\left[
        (-1)^{
            b_\sigma(\xi_t^{(\sigma)})
            \oplus
            b_\sigma(\xi_0^{(\sigma)})
        }
    \right]
    \nonumber\\
    &=
    1
    -
    2\Pr\!\left[
        b_\sigma(\xi_t^{(\sigma)})
        \ne
        b_\sigma(\xi_0^{(\sigma)})
    \right].
\label{eq:classical-register-autocorrelation}
\end{align}
\end{proposition}

\begin{corollary}[Classical-register extension of Theorem~\ref{thm:main-async-formal}]
\label{cor:classical-register-memory}
Let
\(\mathcal L_{\mathrm{ext}}(t)\)
be a register-extended syndrome-based Lindbladian, and let
\[
    \bigl(\xi_t^{(X)},\mathbf c_t^{(X)}\bigr)_{t\ge0},
    \qquad
    \bigl(\xi_t^{(Z)},\mathbf c_t^{(Z)}\bigr)_{t\ge0}
\]
be its induced sector processes.
Let
\(\mathcal E_L^{\mathrm{ext}}(t)\)
be the effective logical channel obtained by preparing the registers
in \(|\mathbf0\rangle\) at encoding and tracing them out at readout.

If the data components of the \(X\)- and \(Z\)-sector processes
respectively \(q_X\)- and \(q_Z\)-remember their logical-error
indicators \(b_X\) and \(b_Z\) up to time \(t\), then the encoded
system is a quantum \((q_X+q_Z)\)-memory up to time \(t\):
\begin{equation}\label{eq:extended-fidelity-bound}
    \inf_{0\le t'\le t}
    F_{\min}^{\mathrm{ext}}(t')
    \ge
    1-q_X-q_Z .
\end{equation}

Conversely, if the encoded system is a quantum
\(\memerr\)-memory up to time \(t\), then each data component
\((\xi_{t'}^{(\sigma)})_{0\le t'\le t}\)
\(\memerr\)-remembers its corresponding logical-error indicator:
\begin{equation}
    \sup_{0\le t'\le t}
    \Pr\!\left[
        b_\sigma(\xi_{t'}^{(\sigma)})
        \ne
        b_\sigma(\xi_0^{(\sigma)})
    \right]
    \le
    \memerr,
    \qquad
    \sigma\in\{X,Z\}.
\end{equation}
\end{corollary}

Both proofs are straightforward extensions of the non-extended version, and are deferred to Appendix~\ref{app:registers}

\subsection{Quantum algebraic lift: constructing the Lindbladian from the PCA}\label{subsec:quantum-lift}

The construction proceeds in two parts. We first augment the Hilbert
space with local ancilla registers that store the syndrome decision and
track whether the correction has been applied. We then define a periodic
Lindbladian cycling through reset, read, and execute stages at rate
$\kappa$.

\begin{definition}[SPCA registers and predicate projectors]
\label{def:predicate-augmented-space}
For each sector $\sigma\in\{X,Z\}$ and each data qubit
$i\in\{1,\ldots,n\}$, introduce a decision register $D_i^\sigma$ and an
execution register $E_i^\sigma$, and set
\begin{equation}
    \mathsf{Reg}_\sigma
    :=
    \{D_i^\sigma , E_i^\sigma :i=1 ,\ldots,n\}.
\end{equation}
We use the classical-register extension of Subsection~\ref{subsec:register-extension} with these
register sets.  All registers are initialized at time $t=0$ in the
all-zero computational-basis state, independently of the encoded
data, and are traced out at readout.  No reinitialization is assumed
at intermediate cycle boundaries.
For each site $i$, let $\Pi_i^X\in\mathcal A_X$ be the syndrome
projector whose eigenvalue function is the $X$-sector PCA predicate,
\begin{equation}
    \Pi_i^X(s)=a_i^{(X)}(s),
    \qquad s\in\Sigma_X,
\end{equation}
and define $\Pi_i^Z\in\mathcal A_Z$ analogously by
\begin{equation}
    \Pi_i^Z(s)=a_i^{(Z)}(s),
    \qquad s\in\Sigma_Z.
\end{equation}
Since every check read by $a_i^{(\sigma)}$ is supported in $N_R(i)$, $\Pi_i^\sigma$ is supported in $N_R(i)$.
\end{definition}

The execution flag $E_i^\sigma$ is essential: without it, a
site where the predicate fires would receive repeated corrections within
the execute stage; the flag blocks the execute jump once $E_i^\sigma=1$,
ensuring exactly one correction per cycle.

The two sectors are given disjoint ancilla registers so that the $X$- and
$Z$-corrections act on separate flags. This is the structural separation required by the extension
(Definition~\ref{def:classical-register-extension}), and it is what lets the 
reduction be applied to each sector independently; the explicit check of the
remaining hypotheses is given after Definition~\ref{def:periodic-schedule}.

With the augmented space in place, we define the generators that drive
the correction cycle.

\subsection{One correction cycle}
\label{subsec:cycle-mechanism}

A synchronous update rule cannot be implemented directly by local jumps.
Jumps act one at a time, at random times, so each jump is triggered by a
syndrome that the preceding jumps may already have changed. A site would
then respond to the partially updated configuration rather than to the
syndrome at the start of the round. The periodic construction avoids this
by separating reading the syndrome from acting on the data, with three
stages per cycle. In the Reset stage, the registers of every site are
cleared. In the Read stage, every site only reads: it evaluates its local
predicate \(a_i^{(\sigma)}\) on the syndrome and, if the predicate is
satisfied, records this by setting its decision register
\(D_i^\sigma=1\). No jump of this stage acts on the data, so in the
absence of noise neither the qubits nor the syndrome change during Read,
and every site reads the same syndrome. In the Execute stage, every site
with \(D_i^\sigma=1\) gets the chance to apply its correction, and its
execution flag \(E_i^\sigma\) ensures that the correction is applied only
once. Since Execute is conditioned on the registers written during Read,
and not on the current syndrome, corrections applied early in the stage
do not affect the decisions of other sites. In the ideal case, made
precise in Lemma~\ref{lem:deterministic-correctness} below, one cycle
therefore implements exactly one step of the synchronous update rule on
the data. See Figure~\ref{fig:stage-timeline} for illustration.

To control these events it is convenient to separate the randomness of
the dynamics from its state. The rate of every correction jump is either
\(\kappa\) or \(0\), depending only on the stage and on the current
configuration, and the rate of every noise jump is a constant.
Instead of thinking of each jump as occurring at a state-dependent rate,
we may therefore attach to each operation an independent Poisson clock
that ticks at a fixed rate, regardless of the state of the system, and
let a tick act only if the operation is allowed at that moment. A Read
clock of site \(i\), for example, keeps ticking at rate \(\kappa\)
throughout the cycle; a tick during the Read stage sets
\(D_i^\sigma=1\) if the predicate is satisfied and \(D_i^\sigma=0\), and
any other tick does nothing. The resulting process is the same induced
sector process (Claim~\ref{clm:clock-representation}), but the events
that matter now become events about individual clocks. Whether an
operation of site \(i\) gets a chance to act during its stage, or
whether noise strikes a given coordinate during the cycle, is decided by
the clocks alone, independently of what the process is doing.
Probabilities of such events are elementary: a clock of rate \(\kappa\)
has no tick during a stage of length \(\tau\) with probability
\(e^{-\kappa\tau}\). Clocks of different sites, and the same clock on
disjoint time intervals, are independent.

\begin{claim}[Clock representation]
\label{clm:clock-representation}
Fix a sector \(\sigma\). Give each site \(i\) four Poisson clocks of rate
\(\kappa\), its \emph{correction clocks}, one for each of its Reset-\(D\),
Reset-\(E\), Read and Execute operations, and give each data coordinate
\(q\) and each register \(q\in\mathsf{Reg}_\sigma\) a \emph{noise clock}
of rate \(\eta_q^{(\sigma)}\) or \(\nu_q\), respectively
(Equations~\eqref{eq:noise-flip-rates} and~\eqref{eq:register-noise-rates}),
all clocks being independent. Let a tick of a correction clock perform
its operation if the time lies in the operation's stage and the
operation's rate in Equation~\eqref{eq:classical-register-rate} is
nonzero at the current configuration, and have no effect otherwise; let
a tick of a noise clock flip its coordinate. Started at
\((\mathbf 0,\mathbf 0)\), the resulting process has the law of the
induced sector process \((\xi_t^{(\sigma)},\mathbf c_t^{(\sigma)})\).
\end{claim}

\begin{proof}
Each correction jump of Definition~\ref{def:stage-generators} has rate
\(\kappa\) times the indicator that its stage is active and that the
current configuration allows it, and each noise jump has a constant
rate. By superposition and thinning of independent Poisson
processes~\cite[\S2.6]{norris_markov_1997}, the clock process is a
Markov process whose rate for each transition is the sum of the rates
of the clocks that can perform it, which is
Equation~\eqref{eq:classical-register-rate}. A Markov process is
determined by its transition rates and its initial state.
\end{proof}

We now explicitly define the periodic Lindbladian.

\begin{definition}[Three generators]\label{def:stage-generators}

Let $\kappa > 0$ be the coupling rate.
The correction cycle is implemented by three time-independent
generators, applied according to the schedule below. In the $X$-sector:
\begin{enumerate}[label=(\alph*), leftmargin=*]
    \item \textbf{Reset:} Clears all auxiliary registers.
    \begin{equation} \mathcal{L}_{\mathrm{reset}}^{(X)} := \kappa \sum_{i=1}^n \Big( \mathcal{D}\Big[ X_{D_i^X} |1\rangle\langle 1|_{D_i^X} \Big] + \mathcal{D}\Big[ X_{E_i^X} |1\rangle\langle 1|_{E_i^X} \Big] \Big) \end{equation}
    \item \textbf{Read:} Flips the decision bit $D_i^X$ from $0$ to $1$
when the predicate $\Pi_i^X$ is satisfied and the decision register
is unset.
    \begin{equation} \mathcal{L}_{\mathrm{read}}^{(X)} := \kappa \sum_{i=1}^n \mathcal{D}\Big[ X_{D_i^X} |0\rangle\langle 0|_{D_i^X} \Pi_i^X \Big] \end{equation}
    \item \textbf{Execute:} Applies the $Z$-correction and flips the execution flag, ensuring a single application.
    \begin{equation} \mathcal{L}_{\mathrm{exec}}^{(X)} := \kappa \sum_{i=1}^n \mathcal{D}\Big[ Z_i X_{E_i^X} |1\rangle\langle 1|_{D_i^X} |0\rangle\langle 0|_{E_i^X} \Big] \end{equation}

\end{enumerate}
The $Z$-sector generators $\mathcal{L}_{\mathrm{read}}^{(Z)}$, $\mathcal{L}_{\mathrm{exec}}^{(Z)}$, and $\mathcal{L}_{\mathrm{reset}}^{(Z)}$ are defined symmetrically with \(\Pi_i^Z\in\mathcal A_Z\) denoting the quantum projector associated
with the \(Z\)-sector predicate \(a_i^{(Z)}(s)=1\). The combined generators acting on both sectors simultaneously are 
$\mathcal{L}_{\mathrm{read}} := \mathcal{L}_{\mathrm{read}}^{(X)} + 
\mathcal{L}_{\mathrm{read}}^{(Z)}$, and analogously for 
$\mathcal{L}_{\mathrm{exec}}$ and $\mathcal{L}_{\mathrm{reset}}$.

\end{definition}

Every jump of $\mathcal L_{\mathrm{PCA}}(t)$ acts on at most one data qubit
and on the registers of one site $i$, and its syndrome function depends only on syndrome
bits of checks within distance $R$ of $i$. Every jump is therefore supported in \(N_R(i)\) together with the
registers of \(i\), so \(\mathcal L_{\mathrm{PCA}}(t)\) is \(R\)-local in
the sense of Definition~\ref{def:local-lindbladian}, with \(R\)
independent of \(L\).

\begin{definition}[Periodic execution schedule]
\label{def:periodic-schedule}
Let $\tau>0$ be the duration of each of the three stages and let
$\Delta=3\tau$ be one complete correction cycle.  The correction
generator is
\begin{equation}
\mathcal L_{\mathrm c}(t)
=
\begin{cases}
\mathcal L_{\mathrm{reset}},
&
t\in[m\Delta,m\Delta+\tau),
\\[1mm]
\mathcal L_{\mathrm{read}},
&
t\in[m\Delta+\tau,m\Delta+2\tau),
\\[1mm]
\mathcal L_{\mathrm{exec}},
&
t\in[m\Delta+2\tau,(m+1)\Delta),
\end{cases}
\end{equation}
for $m\in\mathbb N_0$.
We use the local Pauli-noise model of
Equation~\eqref{eq:extended-noise}, with
\begin{equation}
    0\le\gamma_{P,i}\le\physrate
    \qquad\text{and}\qquad
    0\le\gamma_{P,j}^{\mathrm{reg}}\le\physrate
\end{equation}
for every $P\in\{X,Y,Z\}$, every data qubit $i$, and every register
qubit $j\in\mathsf{Reg}_X\sqcup\mathsf{Reg}_Z$.
We denote this instance of the extended noise generator by
$\mathcal L_{\mathrm{noise}}^{\mathrm{aug}}$.
Together,
\begin{equation}
    \mathcal L_{\mathrm{PCA}}(t)
    :=
    \mathcal L_{\mathrm c}(t)
    +
    \mathcal L_{\mathrm{noise}}^{\mathrm{aug}}.
\end{equation}
\end{definition}

We write $\mathcal W_{\mathrm{PCA}}^*(t,t')$ for the Heisenberg
propagator generated by $\mathcal L_{\mathrm{PCA}}(t)$, as in
Equation~\eqref{eq:q-propgator}.  Over one complete cycle it
factors into the three stage propagators,
\begin{equation}
    \mathcal W_{\mathrm{PCA}}^*\bigl((m+1)\Delta,m\Delta\bigr)
    =
    \mathcal W_{\mathrm{reset}}^*(\tau)
    \circ
    \mathcal W_{\mathrm{read}}^*(\tau)
    \circ
    \mathcal W_{\mathrm{exec}}^*(\tau),
\end{equation}
independently of $m$, where the physical Pauli-noise generator acts
throughout all three stages.
The timing is illustrated in Figure~\ref{fig:stage-timeline}.

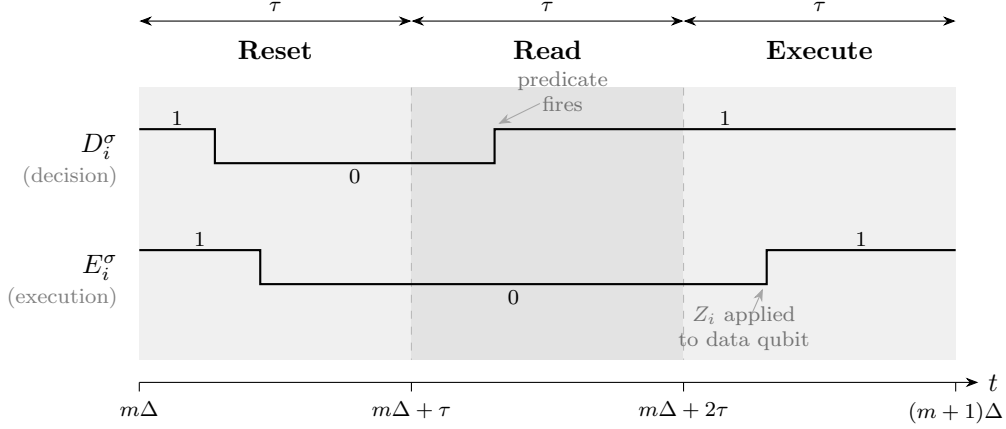
\begin{figure}[h]
\centering
\begin{tikzpicture}[>=Stealth, font=\small]

\pgfmathsetmacro{\pB}{3.6}
\pgfmathsetmacro{\pC}{7.2}
\pgfmathsetmacro{\pD}{10.8}
\pgfmathsetmacro{\hi}{0.45}
\pgfmathsetmacro{\yD}{-1.5}
\pgfmathsetmacro{\yE}{-3.1}
\pgfmathsetmacro{\tax}{-4.4}
\pgfmathsetmacro{\fD}{1.0}
\pgfmathsetmacro{\rD}{4.7}
\pgfmathsetmacro{\fE}{1.6}
\pgfmathsetmacro{\rE}{8.3}

\fill[gray!12] (0,-0.5)   rectangle (\pB,-4.1);
\fill[gray!22] (\pB,-0.5) rectangle (\pC,-4.1);
\fill[gray!12] (\pC,-0.5) rectangle (\pD,-4.1);
\draw[gray!55,thin,dashed] (\pB,-0.5)--(\pB,-4.1);
\draw[gray!55,thin,dashed] (\pC,-0.5)--(\pC,-4.1);

\node[font=\small\bfseries] at ({0.5*\pB},       0){Reset};
\node[font=\small\bfseries] at ({0.5*(\pB+\pC)}, 0){Read};
\node[font=\small\bfseries] at ({0.5*(\pC+\pD)}, 0){Execute};

\draw[<->,thin] (0,   0.38)--(\pB, 0.38) node[midway,above,font=\scriptsize]{$\tau$};
\draw[<->,thin] (\pB, 0.38)--(\pC, 0.38) node[midway,above,font=\scriptsize]{$\tau$};
\draw[<->,thin] (\pC, 0.38)--(\pD, 0.38) node[midway,above,font=\scriptsize]{$\tau$};

\node[anchor=east,font=\small]           at (-0.15,{\yD+0.22}){$D_i^\sigma$};
\node[anchor=east,font=\scriptsize,gray] at (-0.15,{\yD-0.18}){(decision)};
\draw[thick]
  (0,{\yD+\hi})--(\fD,{\yD+\hi})--(\fD,\yD)--(\rD,\yD)--(\rD,{\yD+\hi})--(\pD,{\yD+\hi});
\node[font=\scriptsize] at ({0.5*\fD},       {\yD+\hi+0.14}){$1$};
\node[font=\scriptsize] at ({0.5*(\fD+\rD)}, {\yD-0.17}){$0$};
\node[font=\scriptsize] at ({0.5*(\rD+\pD)}, {\yD+\hi+0.14}){$1$};
\node[font=\scriptsize,text=black!55,align=center]
  at ({\rD+0.9},{\yD+\hi+0.50}){predicate\\fires};
\draw[->,thin,gray!70]
  ({\rD+0.6},{\yD+\hi+0.34})--({\rD+0.05},{\yD+\hi+0.07});

\node[anchor=east,font=\small]           at (-0.15,{\yE+0.22}){$E_i^\sigma$};
\node[anchor=east,font=\scriptsize,gray] at (-0.15,{\yE-0.18}){(execution)};
\draw[thick]
  (0,{\yE+\hi})--(\fE,{\yE+\hi})--(\fE,\yE)--(\rE,\yE)--(\rE,{\yE+\hi})--(\pD,{\yE+\hi});
\node[font=\scriptsize] at ({0.5*\fE},       {\yE+\hi+0.14}){$1$};
\node[font=\scriptsize] at ({0.5*(\fE+\rE)}, {\yE-0.17}){$0$};
\node[font=\scriptsize] at ({0.5*(\rE+\pD)}, {\yE+\hi+0.14}){$1$};
\node[font=\scriptsize,text=black!55,align=center]
  at ({\rE-0.3},{\yE-0.58}){$Z_i$ applied\\to data qubit};
\draw[->,thin,gray!70]
  ({\rE-0.2},{\yE-0.40})--({\rE-0.05},{\yE-0.06});

\draw[->] (0,\tax)--({\pD+0.3},\tax) node[right]{$t$};
\draw (0,\tax)--(0,{\tax-0.12});
\draw (\pB,\tax)--(\pB,{\tax-0.12});
\draw (\pC,\tax)--(\pC,{\tax-0.12});
\draw (\pD,\tax)--(\pD,{\tax-0.12});
\node[below,font=\scriptsize] at (0,   {\tax-0.12}){$m\Delta$};
\node[below,font=\scriptsize] at (\pB, {\tax-0.12}){$m\Delta+\tau$};
\node[below,font=\scriptsize] at (\pC, {\tax-0.12}){$m\Delta+2\tau$};
\node[below,font=\scriptsize] at (\pD, {\tax-0.12}){$(m+1)\Delta$};

\end{tikzpicture}

\caption{One correction cycle of the periodic Lindbladian at a site
where $a_i^{(\sigma)}=1$ (correction required), shown for a site that
was also corrected in the preceding cycle, so that both registers enter
the cycle set.
\textbf{Reset:} both registers return to $0$ when their reset clocks
fire, whatever state the preceding cycle left them in.
\textbf{Read:} $D_i^\sigma$ latches the local predicate ($0\!\to\!1$)
when its Poisson clock fires within the Read stage.
\textbf{Execute:} once $D_i^\sigma=1$, the Pauli correction ($Z_i$ in
the $X$-sector) is applied to the data qubit and $E_i^\sigma$ flips
$0\!\to\!1$, ensuring a single application.  The registers remain set
at the cycle boundary and are cleared by the next Reset stage.  If an
   operation does not occur within its stage, the update at that site may
   deviate from the SPCA rule (Lemma~\ref{lem:deterministic-correctness}).}
\label{fig:stage-timeline}
\end{figure}

\begin{claim}\label{clm:pca-register-extended}
The Lindbladian $\mathcal L_{\mathrm{PCA}}(t)$ of
Definition~\ref{def:periodic-schedule} is a register-extended syndrome-based
Lindbladian (Definition~\ref{def:classical-register-extension}) with the
register sets $\mathsf{Reg}_\sigma$ of
Definition~\ref{def:predicate-augmented-space}.
\end{claim}

\begin{proof}
 Absorbing $\kappa$ via
$\kappa\mathcal D[L]=\mathcal D[\sqrt\kappa L]$, each correction jump of
Definition~\ref{def:stage-generators} has the form $P_i\Theta f$ required by
Definition~\ref{def:classical-register-extension}: a single-qubit data Pauli
$P_i$, possibly $\mathbb I$, a Pauli string $\Theta$ on the registers, and a
syndrome function $f\in\mathcal A^{\mathrm{ext}}$. In the $X$-sector, Reset-$D$
has $P_i=\mathbb I$, $\Theta=X_{D_i^X}$ and $f=\sqrt\kappa\,|1\rangle\langle1|_{D_i^X}$;
Reset-$E$ has $P_i=\mathbb I$, $\Theta=X_{E_i^X}$ and
$f=\sqrt\kappa\,|1\rangle\langle1|_{E_i^X}$; Read has $P_i=\mathbb I$,
$\Theta=X_{D_i^X}$ and $f=\sqrt\kappa\,|0\rangle\langle0|_{D_i^X}\Pi_i^X$; and
Execute has $P_i=Z_i$, $\Theta=X_{E_i^X}$ and
$f=\sqrt\kappa\,|1\rangle\langle1|_{D_i^X}|0\rangle\langle0|_{E_i^X}$. Since
$|0\rangle\langle0|=(\mathbb I+Z)/2$, $|1\rangle\langle1|=(\mathbb I-Z)/2$ and
$\Pi_i^X\in\mathcal A_X$, these syndrome functions lie in
$\mathcal A_X^{\mathrm{ext}}$. The $Z$-sector correction jumps have
$P_i\in\{\mathbb I,X_i\}$ and $\Theta$ supported on $\mathsf{Reg}_Z$, so their
$X$-sector flip is $(\mathbf 0,\mathbf 0)$ and they enter no
$\Gamma^{(X)}_{\mathbf u,\mathbf v}(t)$. The noise jumps, of the form in
Equation~\eqref{eq:extended-noise}, have scalar syndrome functions. Hence every
$\Gamma^{(X)}_{\mathbf u,\mathbf v}(t)$ is a sum of products $f^\dagger f$ with
$f\in\mathcal A_X^{\mathrm{ext}}$ and of scalars, and the extended rate
condition holds. The $Z$-sector is identical with $X$ and $Z$ exchanged.
\end{proof}

\subsection{Sampled sector processes as SPCA runs}\label{subsec:periodic-equivalence}

By Claim~\ref{clm:pca-register-extended}, the periodic Lindbladian
$\mathcal L_{\mathrm{PCA}}(t)$ is a register-extended syndrome-based
Lindbladian. For each sector $\sigma\in\{X,Z\}$ it therefore induces a
classical Markov process $(\xi_t^{(\sigma)},\mathbf c_t^{(\sigma)})$ on
$\Omega_\sigma^{\mathrm{ext}}$, with the transition rates of
Equation~\eqref{eq:classical-register-rate}, started at $(\mathbf 0,\mathbf 0)$;
here $\mathbf c_t^{(\sigma)}$ collects the registers $D_i^\sigma,E_i^\sigma$, and
we write $D_i^\sigma,E_i^\sigma$ also for their bits. Throughout this
subsection we fix a sector $\sigma$.

Our goal is to show that the data component of this process, sampled at cycle
boundaries, is a trajectory of the SPCA $\Psi_\sigma$ of
Definition~\ref{def:syndrome-pca-classical}: in each cycle it performs the
deterministic update $\mathbf a^{(\sigma)}$ except on a random set of sites,
and this set contains each site independently with a small probability and is
drawn afresh in every cycle. The argument has three steps. First, a site at
which all register operations of a cycle take place and no noise occurs nearby
performs exactly its SPCA update in that cycle
(Lemma~\ref{lem:deterministic-correctness}). Second, the sites where this may
fail form a \emph{failure set}; it is rare, independent of the past, and
correlated only between nearby sites
(Lemma~\ref{lem:local-failure-set}). Third, since the SPCA allows arbitrary
deviations on its fault sites, it suffices to enlarge the failure set to
a set with independent sites; a standard domination theorem does this at a
slightly larger rate (Proposition~\ref{prop:sampled-sector-spca}).

\begin{definition}[Sampled process and deviation]
\label{def:sampled-deviation}
The \emph{sampled data process} at cycle boundaries is
\begin{equation}
    \hat\xi_r^{(\sigma)}:=\xi_{r\Delta}^{(\sigma)},
    \qquad r\in\mathbb N_0,
\end{equation}
and the \emph{deviation} of cycle $r$ from one deterministic SPCA step is
\begin{equation}
    \zeta_r^{(\sigma)}
    :=
    \hat\xi_{r+1}^{(\sigma)}
    \oplus
    \hat\xi_r^{(\sigma)}
    \oplus
    \mathbf a^{(\sigma)}\!\left(H_\sigma\hat\xi_r^{(\sigma)}\right).
    \label{eq:implementation-deviation}
\end{equation}
\end{definition}

Here $\mathbf a^{(\sigma)}$ is the update rule of $\Psi_\sigma$
(Definition~\ref{def:syndrome-pca-classical}), so
$\hat\xi_r^{(\sigma)}\oplus\mathbf a^{(\sigma)}(H_\sigma\hat\xi_r^{(\sigma)})$ is
the configuration that one noiseless SPCA step would produce from
$\hat\xi_r^{(\sigma)}$. The deviation $\zeta_r^{(\sigma)}$ is the difference
between this configuration and the one actually reached at the end of
the cycle; it vanishes at exactly those sites where the cycle performs
the SPCA update. Comparing with
Definition~\ref{def:syndrome-pca-classical}, the sampled process is a
run of $\Psi_\sigma$ at fault rate $p$ precisely when the deviations are
supported in fault sets with the properties required there.

\begin{definition}[Local neighborhoods]
\label{def:local-neighbourhoods}
For a site $i$, let $\mathcal J_\sigma(i)$ be the set of checks $j$ such that
$a_i^{(\sigma)}$ depends on the syndrome bit $s_j$, and let $H_{\sigma,j}$ be
row $j$ of $H_\sigma$. The \emph{data neighborhood} and the
\emph{neighborhood} of $i$ are
\begin{equation}
    \mathcal N_\sigma(i)
    :=
    \{i\}
    \cup
    \bigcup_{j\in\mathcal J_\sigma(i)}\operatorname{supp}(H_{\sigma,j}),
    \qquad
    \widetilde{\mathcal N}_\sigma(i)
    :=
    \mathcal N_\sigma(i)\sqcup\{D_i^\sigma,E_i^\sigma\}.
    \label{eq:implementation-neighborhood}
\end{equation}
The \emph{influence set} of a data coordinate or register $q$ is
$\mathcal I_\sigma(q):=\{i:q\in\widetilde{\mathcal N}_\sigma(i)\}$, and the
\emph{overlap graph} $\mathcal G_\sigma$ is the graph on sites in which $i\ne j$
are adjacent when
$\widetilde{\mathcal N}_\sigma(i)\cap\widetilde{\mathcal N}_\sigma(j)\ne\varnothing$.
\end{definition}

The data neighborhood consists of site $i$ and every data coordinate read by
the predicate $a_i^{(\sigma)}$; the influence set of $q$ consists of the sites
whose update can be affected by a change of $q$. We show that their sizes are bounded.

\begin{lemma}[Locality bound]
\label{lem:locality-bound}
There is a constant $K<\infty$, independent of $L$, such that for both sectors
\begin{equation}
    \max_i\bigl|\widetilde{\mathcal N}_\sigma(i)\bigr|\le K,
    \qquad
    \max_q\bigl|\mathcal I_\sigma(q)\bigr|\le K,
    \label{eq:implementation-locality-bound}
\end{equation}
where $q$ ranges over data coordinates and the registers in
$\mathsf{Reg}_\sigma$. In particular, $\mathcal G_\sigma$ has maximal degree at
most $K^2$.
\end{lemma}
\begin{proof}
Let $d_{\max}$ be the maximal degree of the interaction graph, which is
independent of $L$, and let $K_0:=\sum_{k=0}^{R}d_{\max}^k$, so that
$|N_R(q)|\le K_0$ for every data coordinate $q$. Every check read by
$a_i^{(\sigma)}$ is supported in $N_R(i)$, hence
$\mathcal N_\sigma(i)\subseteq N_R(i)$ and
$|\widetilde{\mathcal N}_\sigma(i)|\le K_0+2$. If a data coordinate $q$
lies in $\mathcal N_\sigma(i)$, then $d(i,q)\le R$, so
$\mathcal I_\sigma(q)\subseteq N_R(q)$ and $|\mathcal I_\sigma(q)|\le K_0$;
for a register $q\in\{D_i^\sigma,E_i^\sigma\}$,
$\mathcal I_\sigma(q)=\{i\}$. Set $K:=K_0+2$. Finally, if $j\ne i$ is
adjacent to $i$ in $\mathcal G_\sigma$, then
$\widetilde{\mathcal N}_\sigma(i)$ and $\widetilde{\mathcal N}_\sigma(j)$
share a coordinate $q$, which is a data coordinate since each register
belongs to a single site, and $j\in\mathcal I_\sigma(q)$. There are at
most $K$ choices of $q$ and at most $K$ choices of $j$ for each, so the
degree of $\mathcal G_\sigma$ is at most $K^2$.
\end{proof}

We first show that a site at which the cycle is ideal, in the sense of
Subsection~\ref{subsec:cycle-mechanism}, performs exactly its SPCA
update. In terms of the clocks of Claim~\ref{clm:clock-representation},
the ideal cycle is the following.

\begin{lemma}[Deterministic correctness of one cycle]
\label{lem:deterministic-correctness}
Fix a cycle $r$ and a site $i$. Suppose that during the cycle
$[r\Delta,(r+1)\Delta)$ each correction clock of $i$ ticks at least once in its
stage, and no noise clock of a coordinate in $\widetilde{\mathcal N}_\sigma(i)$
ticks. Then $\zeta_{r,i}^{(\sigma)}=0$, that is,
\begin{equation}
    \hat\xi_{r+1,i}^{(\sigma)}
    =
    \hat\xi_{r,i}^{(\sigma)}
    \oplus
    a_i^{(\sigma)}\!\left(H_\sigma\hat\xi_r^{(\sigma)}\right).
    \label{eq:local-cycle-correctness}
\end{equation}
No assumption is made on the values of $D_i^\sigma$ and $E_i^\sigma$ at time
$r\Delta$.
\end{lemma}

\begin{proof}
During the Reset stage the first tick of each reset clock of $i$ sets its
register to $0$, whatever its value at $r\Delta$, and without noise ticks on
the registers both stay $0$ until the Read stage. Before the Execute stage the
data change only through noise ticks, and there are none in
$\mathcal N_\sigma(i)$, which contains the supports of the checks read by
$a_i^{(\sigma)}$; hence throughout the Read stage the predicate at $i$ equals
$a_i^{(\sigma)}(H_\sigma\hat\xi_r^{(\sigma)})$. If this bit is $1$, the first Read
tick sets $D_i^\sigma=1$; otherwise $D_i^\sigma$ stays $0$. During the Execute
stage $D_i^\sigma$ is frozen and $E_i^\sigma=0$ at its start. If $D_i^\sigma=1$,
the first Execute tick flips data coordinate $i$ and sets $E_i^\sigma=1$, after
which further ticks have no effect; if $D_i^\sigma=0$, Execute ticks have no
effect. Execute ticks of other sites flip other coordinates, and no noise clock
of coordinate $i$ ticks. The net change of coordinate $i$ over the cycle is
therefore $a_i^{(\sigma)}(H_\sigma\hat\xi_r^{(\sigma)})$.
\end{proof}

In particular, a cycle in which the hypotheses of
Lemma~\ref{lem:deterministic-correctness} hold at every site has
$\zeta_r^{(\sigma)}=\mathbf 0$ and performs exactly one noiseless SPCA step.

\begin{definition}[Failure set]
\label{def:local-failure-set}
The \emph{failure set} $S_r^{(\sigma)}$ of cycle $r$ is the set of sites
at which the hypothesis of Lemma~\ref{lem:deterministic-correctness} fails: some
correction clock of $i$ does not tick in its stage of cycle $r$, or some noise
clock of a coordinate in $\widetilde{\mathcal N}_\sigma(i)$ ticks during
$[r\Delta,(r+1)\Delta)$.
\end{definition}

A correction clock misses its stage with probability $e^{-\kappa\tau}$, and
each noise clock has rate at most $2\physrate$, since its rate is a sum of two
Pauli rates. With the locality range $K$ as in Lemma~\ref{lem:locality-bound}, set
\begin{equation}
    \delta
    :=
    \min\!\left\{
        1,\,
        4e^{-\kappa\tau}
        +K\bigl(1-e^{-2\physrate\Delta}\bigr)
    \right\}.
    \label{eq:seed-probability}
\end{equation}

Throughout, the letter $p$ is reserved for fault rates of an SPCA, and $\delta$
for this uniform per-site, per-cycle failure-probability bound.

\begin{lemma}[Properties of the failure set]
\label{lem:local-failure-set}
For every cycle $r$:
\begin{enumerate}[label=(\roman*)]
    \item $\operatorname{supp}(\zeta_r^{(\sigma)})\subseteq S_r^{(\sigma)}$;
    \item $\Pr[i\in S_r^{(\sigma)}]\le \delta$ for every site $i$;
    \item $S_r^{(\sigma)}$ is independent of the restrictions of all the Poisson
    processes to $[0,r\Delta)$, and hence of the process up to time
    $r\Delta$;
    \item the indicators of $S_r^{(\sigma)}$ on two disjoint sets of sites with
    no edge of $\mathcal G_\sigma$ between them are independent.
\end{enumerate}
\end{lemma}

\begin{proof}
(i) This is Lemma~\ref{lem:deterministic-correctness}: at a site outside
$S_r^{(\sigma)}$ its hypothesis holds, so $\zeta_{r,i}^{(\sigma)}=0$.

(ii) The four correction clocks of $i$ each miss their stage with probability
$e^{-\kappa\tau}$. The noise clock of a coordinate has rate at most
$2\physrate$, so it ticks during the cycle with probability at most
$1-e^{-2\physrate\Delta}$, and $\widetilde{\mathcal N}_\sigma(i)$ has at most $K$
coordinates. A union bound gives $\delta$.

(iii) $S_r^{(\sigma)}$ is determined by the clocks restricted to
$[r\Delta,(r+1)\Delta)$, and restrictions of Poisson processes to disjoint time
intervals are independent. The process up to time $r\Delta$ is determined by
the restrictions to $[0,r\Delta)$.

(iv) The event $\{i\in S_r^{(\sigma)}\}$ is determined by the correction clocks
of $i$ and the noise clocks of the coordinates in
$\widetilde{\mathcal N}_\sigma(i)$. For two disjoint sets of sites with no edge
between them, these families of clocks are disjoint, and the indicators are
functions of disjoint families of independent clocks.
\end{proof}

If the sites of $S_r^{(\sigma)}$ were independent, then, after enlarging
it independently at each site to density exactly $\delta$, items
(i)--(iii) would already make it an SPCA fault set at fault rate
$\delta$, and the sampled process a run of $\Psi_\sigma$. They are not
independent: two sites whose neighborhoods share a coordinate fail together
when that coordinate is flipped by noise. The SPCA, however, allows arbitrary
deviations on its fault sites, so it suffices to find a larger set, with
independent sites, that contains $S_r^{(\sigma)}$. Because the correlations of
$S_r^{(\sigma)}$ are local and its sites are rarely in it, such a set exists
with only a slightly larger density.

\begin{proposition}[The sampled sector process is an SPCA trajectory]
\label{prop:sampled-sector-spca}
There exists a function
\begin{equation}
    \effrate:[0,1]\to[0,1],
    \qquad
    \effrate(\delta)\xrightarrow[\delta\to0]{}0,
\end{equation}
depending only on $K$, such that the sampled data process
$(\hat\xi_r^{(\sigma)})_{r\ge0}$ is a trajectory of the SPCA
$\Psi_\sigma$ at fault rate $\effrate(\delta)$. That is, there are random sets
$F_r^{(\sigma)}\subseteq\{1,\dots,n\}$, $r\ge0$, such that
\begin{enumerate}[label=(\roman*)]
    \item each $F_r^{(\sigma)}$ contains each site independently with probability
    $\effrate(\delta)$;
    \item each $F_r^{(\sigma)}$ is independent of the process up to time $r\Delta$
    together with $F_0^{(\sigma)},\dots,F_{r-1}^{(\sigma)}$;
    \item $\operatorname{supp}(\zeta_r^{(\sigma)})\subseteq F_r^{(\sigma)}$ for every
    $r$.
\end{enumerate}
\end{proposition}

\begin{proof}[Proof idea]
The Liggett--Schonmann--Stacey theorem~\cite{liggett_domination_1997} provides,
for each cycle, a set of independent sites containing $S_r^{(\sigma)}$, built
from $S_r^{(\sigma)}$ and fresh auxiliary randomness. Item~(iii) then follows
from Lemma~\ref{lem:local-failure-set}(i), and item~(ii) from
Lemma~\ref{lem:local-failure-set}(iii). Proof deferred to
Appendix~\ref{app:pca}.
\end{proof}

Items~(i)--(iii) are exactly the SPCA fault model: the sites where a cycle may
deviate from the deterministic rule form a fresh set with independent sites
of density $\effrate(\delta)$, and the deviation on that set is
unrestricted. It may depend on the history, for instance through register
values left by an earlier cycle when a reset is missed.

\subsection{Classical memory transfer and memory bound}\label{subsubsec:memory-bound}

We now transfer the assumed SPCA threshold to the sampled sector processes.
We keep the notation of Subsection~\ref{subsec:periodic-equivalence}: for each
sector $\sigma$, $\hat\xi_r^{(\sigma)}$ is the sampled data process of the sector
process induced by $\mathcal L_{\mathrm{PCA}}(t)$, and $\effrate(\delta)$ is the
fault rate of the SPCA it realizes (Proposition~\ref{prop:sampled-sector-spca}),
a function of the failure probability $\delta$ of
Equation~\eqref{eq:seed-probability}. 

\begin{claim}[Classical memory transfer]
\label{clm:memory-transfer}
Fix a sector $\sigma$, and suppose that for every $p<\pcth$ the SPCA
$\Psi_\sigma$ at fault rate $p$ $q_\sigma(m,L)$-remembers $b_\sigma$
through round $m$ (Definition~\ref{def:syndrome-pca-classical}). If
$\effrate(\delta)<\pcth$, then the sampled data process
$q_\sigma(m,L)$-remembers $b_\sigma$ through round $m$
(Definition~\ref{def:classical-memory}).
\end{claim}

\begin{proof}
By Proposition~\ref{prop:sampled-sector-spca}, $(\hat\xi_r^{(\sigma)})_{r\ge0}$
is a run of $\Psi_\sigma$ at fault rate $\effrate(\delta)<\pcth$, started at
$\hat\xi_0^{(\sigma)}=\mathbf 0$, and the hypothesis applies to every such
run.
\end{proof}

We finish by choosing the implementation parameters so that
$\effrate(\delta)<\pcth$, and then applying the register-extended fidelity
bound.

\begin{proof}[Proof of Theorem~\ref{thm:periodic-formal}]
Let $K$ be as in Lemma~\ref{lem:locality-bound} and $\effrate$ as in
Proposition~\ref{prop:sampled-sector-spca}. Since $\effrate(\delta)\to0$ as
$\delta\to0$, there is a $\delta_c>0$, depending only on $K$ and $\pcth$, such
that $\effrate(\delta)<\pcth$ for every $\delta\le\delta_c$. By
Equation~\eqref{eq:seed-probability}, it then suffices to make
$\delta\le\delta_c$. First choose $\kappa\tau$ large enough that
$4e^{-\kappa\tau}\le\delta_c/2$. With $\kappa$ and $\tau$, hence $\Delta=3\tau$,
fixed, choose $\eta_c>0$ small enough that
\begin{equation}
    K\bigl(1-e^{-2\eta_c\Delta}\bigr)\le\frac{\delta_c}{2}.
\end{equation}
These choices are independent of $L$. Every noise rate $\physrate<\eta_c$ then
gives $\delta\le\delta_c$, hence $\effrate(\delta)<\pcth$. Applying
Claim~\ref{clm:memory-transfer} to each sector $\sigma$ gives, for every
$0\le r\le m$,
\begin{equation}
    \Pr\!\left[
    b_\sigma\!\left(\hat\xi_r^{(\sigma)}\right)
    \ne
    b_\sigma\!\left(\hat\xi_0^{(\sigma)}\right)
\right]
\le
q_\sigma(m,L).
\end{equation}
Corollary~\ref{cor:classical-register-memory} at the cycle boundary
$t_r=r\Delta$ yields
\begin{align}
    F_{\min}^{\mathrm{ext}}(t_r)
    &\ge
    1
    -\Pr\!\left[
        b_X\!\left(\hat\xi_r^{(X)}\right)
        \ne
        b_X\!\left(\hat\xi_0^{(X)}\right)
    \right]
    -\Pr\!\left[
        b_Z\!\left(\hat\xi_r^{(Z)}\right)
        \ne
        b_Z\!\left(\hat\xi_0^{(Z)}\right)
    \right]
    \nonumber\\
    &\ge 1-q_X(m,L)-q_Z(m,L).
\end{align}
for every $0\le r\le m$. By Corollary~\ref{cor:classical-register-memory},
\(F_{\min}^{\mathrm{ext}}\) is the worst-case fidelity of the effective
logical channel in which the registers are prepared in
\(|\mathbf 0\rangle\) at encoding and traced out at readout, which is the
effective logical channel of this dynamics. Hence the augmented encoded system is a quantum
\(\bigl(q_X(m,L)+q_Z(m,L)\bigr)\)-memory at the cycle boundaries in the sense of
Definition~\ref{def:quantum-memory-cycles}.

The bound does not depend on the initial register values.
Lemma~\ref{lem:deterministic-correctness} makes no assumption on the
register values at the start of a cycle, and the failure set of
Definition~\ref{def:local-failure-set}, together with the properties of
Lemma~\ref{lem:local-failure-set}, is defined by the clocks alone. A
missed Reset, which would leave a register at whatever value it held, is
itself one of the failure events. Hence the same bound holds for
registers starting in any state.

$\mathcal L_{\mathrm{PCA}}(t)$ is $R$-local by the paragraph following Definition~\ref{def:periodic-schedule}. If the CSS code family and the two
SPCA families are translationally invariant, so is the constructed
Lindbladian family.
\end{proof}

\section{The classical problem}
\label{sec:classical-problem}

We isolate here the purely classical problem that is sufficient for the
quantum-memory constructions of Corollary~\ref{cor:sync-spca-memory}
and Theorem~\ref{thm:periodic-formal}.
The problem involves only binary linear codes, local syndrome-based
dynamics, and a classical logical readout. No continuous-time dynamics
is needed.

For each system size $L$, let
\begin{equation}
    H_X\in\mathbb F_2^{m_X\times n},
    \qquad
    H_Z\in\mathbb F_2^{m_Z\times n},
\end{equation}
with $n=n(L)$, satisfy
\begin{equation}
    H_XH_Z^T=0.
\end{equation}
Write
\begin{equation}
    \mathcal C_X:=\ker H_X,
    \qquad
    \mathcal C_Z:=\ker H_Z,
    \qquad
    \Omega_X=\Omega_Z:=\mathbb F_2^n.
\end{equation}
The bit and parity-check locations are equipped with a bounded-degree
interaction geometry, uniform in $L$, and the parity checks have uniformly
bounded spatial support.

We first specify the classical logical degree of freedom to be protected.
Choose
\begin{equation}
    \ell_X\in\mathcal C_Z,
    \qquad
    \ell_Z\in\mathcal C_X,
    \qquad
    \ell_X\cdot\ell_Z=1.
\end{equation}
These vectors define conjugate binary logical coordinates
\begin{equation}
    \mu_X:
    \mathcal C_X/\operatorname{row}(H_Z)\longrightarrow\mathbb F_2,
    \qquad
    \mu_X([\mathbf c]_X):=\ell_X\cdot\mathbf c,
\end{equation}
and
\begin{equation}
    \mu_Z:
    \mathcal C_Z/\operatorname{row}(H_X)\longrightarrow\mathbb F_2,
    \qquad
    \mu_Z([\mathbf c]_Z):=\ell_Z\cdot\mathbf c.
\end{equation}
The condition $\ell_X\cdot\ell_Z=1$ ensures that the two coordinates
correspond to the two conjugate logical directions of the same logical
degree of freedom.

For each sector $\sigma\in\{X,Z\}$, fix a readout decoder
\begin{equation}
    \mathsf{Dec}_\sigma:
    \Sigma_\sigma\longrightarrow\Omega_\sigma,
    \qquad
    H_\sigma\mathsf{Dec}_\sigma(s)=s,
    \qquad
    \mathsf{Dec}_\sigma(\mathbf0)=\mathbf0,
\end{equation}
where $\Sigma_\sigma=\operatorname{im}H_\sigma$. The corresponding
decoded logical value is
\begin{equation}
    b_\sigma(\mathbf e)
    :=
    \mu_\sigma\!\left(
        \left[
            \mathbf e
            \oplus
            \mathsf{Dec}_\sigma(H_\sigma\mathbf e)
        \right]_\sigma
    \right),
    \qquad
    \mathbf e\in\Omega_\sigma.
    \label{eq:classical-problem-b}
\end{equation}
The decoder is used only to define the readout criterion and need not
be local.

The dynamical object to be constructed is instead local. For each sector
$\sigma$, choose an $R$-local syndrome-based update
\begin{equation}
    \mathbf a^{(\sigma)}
    =
    \bigl(
        a_1^{(\sigma)},\ldots,a_n^{(\sigma)}
    \bigr),
\end{equation}
where each $a_i^{(\sigma)}$ depends only on syndrome bits whose parity
checks lie within a fixed distance $R$ of site $i$, with $R$ independent
of $L$. At stochastic-adversarial fault rate $p$, the resulting SPCA
trajectory is
\begin{equation}
    \mathbf e_{r+1}^{(\sigma)}
    =
    \mathbf e_r^{(\sigma)}
    \oplus
    \mathbf a^{(\sigma)}
        \!\left(H_\sigma\mathbf e_r^{(\sigma)}\right)
    \oplus
    \zeta_r^{(\sigma)}.
    \label{eq:classical-problem-spca}
\end{equation}
At each round $r$, a fault set
$F_r^{(\sigma)}\subseteq\{1,\ldots,n\}$ is sampled by including every
site independently with probability $p$, independently of the process
through round $r$ and of all previous fault sets. The deviation
$\zeta_r^{(\sigma)}$ may otherwise be arbitrary, including
history-dependent, subject only to
\begin{equation}
    \operatorname{supp}\zeta_r^{(\sigma)}
    \subseteq
    F_r^{(\sigma)}.
\end{equation}
This is precisely the stochastic-adversarial SPCA fault model of
Definition~\ref{def:syndrome-pca-classical}.

The classical problem is to construct such local update rules in both
sectors and prove a nonzero fault threshold. Concretely, one seeks a
constant
\begin{equation}
    \pcth>0,
\end{equation}
independent of $L$, and retention bounds
$q_X(m,L),q_Z(m,L)$ such that, for every $p<\pcth$, every sector
$\sigma\in\{X,Z\}$, and every run of
Eq.~\eqref{eq:classical-problem-spca} initialized at
\begin{equation}
    \mathbf e_0^{(\sigma)}=\mathbf0,
\end{equation}
one has
\begin{equation}
    \max_{0\le r\le m}
    \Pr\!\left[
        b_\sigma\!\left(\mathbf e_r^{(\sigma)}\right)
        \ne
        b_\sigma\!\left(\mathbf e_0^{(\sigma)}\right)
    \right]
    \le
    q_\sigma(m,L).
    \label{eq:classical-problem-retention}
\end{equation}
Since $b_\sigma(\mathbf0)=0$, this is equivalently
\begin{equation}
    \max_{0\le r\le m}
    \Pr\!\left[
        b_\sigma\!\left(\mathbf e_r^{(\sigma)}\right)=1
    \right]
    \le
    q_\sigma(m,L).
\end{equation}
The desired memory-time regime is one in which
$q_\sigma(m,L)\to0$ as $L\to\infty$ while the admissible number of
rounds $m=m(L)$ diverges.

This formulation deliberately starts the SPCA from the vacuum
$\mathbf0$. The variable $\mathbf e^{(\sigma)}$ is the error
configuration relative to the initially stored codeword, rather than the
stored codeword itself. By the codeword-translation invariance established
in Claim~\ref{clm:translation-invariance}, this normalization loses no
information about retention of an arbitrary initial logical value. Thus
the classical input required by the quantum construction is exactly the
vacuum stability statement~\eqref{eq:classical-problem-retention}; one
does not separately need to prove stability for every initial codeword.

Solving this classical problem for both sectors is the entire classical
input to both quantum constructions. In the synchronous setting, 
Corollary~\ref{cor:sync-spca-memory} lifts the two update rules to a local
syndrome-based quantum channel whose induced sector processes, under Pauli
noise of rate $\eta$, are runs of the same SPCAs at fault rate $2\eta$; a
nonzero $\pcth$ satisfying Eq.~\eqref{eq:classical-problem-retention}
therefore gives the physical threshold $\eta<\pcth/2$ directly. In the
asynchronous setting, Section~\ref{sec:periodic} implements the same rules
by a periodic syndrome-based Lindbladian, and
Proposition~\ref{prop:sampled-sector-spca} shows that its sector processes,
sampled at cycle boundaries, are runs of the same SPCAs at an effective
fault rate that can be made arbitrarily small by weakening the physical
noise and choosing the implementation parameters appropriately; 
Theorem~\ref{thm:periodic-formal} then again yields a nonzero physical
noise threshold. In both cases the quantum memory error is
$q_X(m,L)+q_Z(m,L)$, at cycle boundaries in the periodic case.

\section*{Acknowledgments}
We thank Shankar Balasubramanian, Tal Barak,  
Margarita Davydova and Ethan Lake for useful discussions.    
D.A. acknowledges the generous support of the advanced ERC grant number 101201649. 
The research was supported by the ISF Quantum Science and Technology (2074/19) and by the VATAT Research Focus Centers National Quantum Science and Technology Program. Grant no. 134814

\newpage
\appendix

\section{Supporting Proofs for Section~\ref{sec:setup}}
\label{app:setup}
\begin{proof}[Proof of Claim~\ref{clm:translation-invariance}]
By definition of the decoder (Definition~\ref{def:readout-decoder}) the strings $\mathbf x$ and
$T_{\mathbf w}(\mathbf x)$ carry the same syndrome. The transition rates of
Definition~\ref{def:ctmc} depend on the string only through its syndrome, and
the noise rates $\eta_i^{(X)}$ not at all, so the rate of a flip at site $i$
agrees at the two strings. Using this in Eq.~\eqref{eq:ctmc-generator},
together with
$T_{\mathbf w}(\mathbf x)\oplus \hat{e}_i=T_{\mathbf w}(\mathbf x\oplus \hat{e}_i)$,
\begin{align}
    \bigl(G_X(t)h_X\bigr)\bigl(T_{\mathbf w}(\mathbf x)\bigr)
    &=
    \sum_i
    \left[
        r_i^{(X)}\bigl(H_XT_{\mathbf w}(\mathbf x),t\bigr)+\eta_i^{(X)}
    \right]
    \bigl(
        h_X(T_{\mathbf w}(\mathbf x)\oplus \hat{e}_i)
        -h_X(T_{\mathbf w}(\mathbf x))
    \bigr)
    \nonumber\\
    &=
    \sum_i
    \left[
        r_i^{(X)}(H_X\mathbf x,t)+\eta_i^{(X)}
    \right]
    \bigl(
        (h_X\circ T_{\mathbf w})(\mathbf x\oplus \hat{e}_i)
        -(h_X\circ T_{\mathbf w})(\mathbf x)
    \bigr)
    \nonumber\\
    &=
    \Bigl(G_X(t)\bigl(h_X\circ T_{\mathbf w}\bigr)\Bigr)(\mathbf x),
\end{align}
which is Eq.~\eqref{eq:generator-intertwining}. Taking $\mathbf w=\mathbf w^X_0$, let $\widetilde G_X(t)$ denote the generator
of the transformed process
$\xi_t^{(X)}=T_{\mathbf w^X_0}(\chi_t^{(X)})$. Since
$T_{\mathbf w^X_0}^{-1}=T_{\mathbf w^X_0}$, for every
$h_X:\Omega_X\to\mathbb R$,
\begin{align}
\bigl(\widetilde G_X(t)h_X\bigr)(\mathbf x)
&=
\Bigl(
G_X(t)(h_X\circ T_{\mathbf w^X_0})
\Bigr)\bigl(T_{\mathbf w^X_0}(\mathbf x)\bigr)
\\
&=
\Bigl(
(G_X(t)h_X)\circ T_{\mathbf w^X_0}
\Bigr)\bigl(T_{\mathbf w^X_0}(\mathbf x)\bigr)
\\
&=
\bigl(G_X(t)h_X\bigr)(\mathbf x),
\end{align}
where the second equality is
Eq.~\eqref{eq:generator-intertwining}. Hence
$\widetilde G_X(t)=G_X(t)$.

Moreover,
$
    \xi_0^{(X)}
    =
    T_{\mathbf w^X_0}(\chi_0^{(X)})
    =
    T_{\mathbf w^X_0}(\mathbf w^X_0)
    =
    \mathbf 0.
$
Thus the relative process has generator $G_X(t)$ and therefore the same
backward transition operators $M^{(X)}(t,t')$ as the original sector
process. Its deterministic initial state is $\mathbf 0$, independently of
$\mathbf w^X_0$. Hence its law is independent of the initially stored codeword.
\end{proof}

\begin{proof}[Proof of Claim~\ref{clm:sync-locality-H}]
Fix \(A\) and write \(N:=N_R(A)\). Suppose
\(\mathcal Q^*(O_A\otimes\mathbb I_{\bar A})=\tilde O\otimes\mathbb I_{\bar N}\)
for every \(O_A\). Fix a state \(\omega\) on \(\mathcal H_{\bar N}\) and
set
\(\mathcal Q_{A*}(\tau):=\operatorname{Tr}_{\bar A}\mathcal Q_*(\tau\otimes\omega)\).
For every state \(\rho\),
\(
    \operatorname{Tr}\bigl[O_A\operatorname{Tr}_{\bar A}\mathcal Q_*(\rho)\bigr]
    =
    \operatorname{Tr}\bigl[\tilde O\operatorname{Tr}_{\bar N}\rho\bigr]
\),
which depends on \(\rho\) only through \(\operatorname{Tr}_{\bar N}\rho\).
Applying this to \(\rho\) and to
\(\operatorname{Tr}_{\bar N}\rho\otimes\omega\) gives
Equation~\eqref{eq:sync-locality-S}. Conversely, if
Equation~\eqref{eq:sync-locality-S} holds, then
\(X\mapsto\operatorname{Tr}[\mathcal Q^*(O_A\otimes\mathbb I_{\bar A})X]\)
coincides on states, and hence by linearity on all operators, with
\(X\mapsto\operatorname{Tr}[O_A\,\mathcal Q_{A*}(\operatorname{Tr}_{\bar N}X)]\).
It therefore vanishes on \(\ker\operatorname{Tr}_{\bar N}\), whose
Hilbert--Schmidt orthogonal complement consists of the operators
supported on \(N\).
\end{proof}

\section{Supporting Proofs for Section~\ref{sec:async}}
\label{app:reduction}

\begin{proof}[Proof of Lemma~\ref{lem:algebraic-form}]
We evaluate the $X$-sector observable. Substituting the decoder error scalars from Definition~\ref{def:decoder-signs} into Definition~\ref{def:decoded-observables} yields:
\begin{equation}
    T_X = \sum_{s \in \Sigma_X} \Pi_s^X \varphi_X(s) \bar{X} \Pi_s^X.
\end{equation}
Because $\bar{X}$ is an $X$-type Pauli logical, it commutes with every 
$X$-stabilizer generator, and hence with every element of $\mathcal{A}_X$, 
including the syndrome projectors $\Pi_s^X$. The operator $\bar{X}$ factors out of the summation:
\begin{equation}
    T_X = \bar{X} \sum_{s \in \Sigma_X} \varphi_X(s) \Pi_s^X.
\end{equation}
Recognizing the sum $\sum_{s \in \Sigma_X} \varphi_X(s)\,\Pi_s^X$ as the 
operator $\varphi_X$ from Definition~\ref{def:decoder-signs} establishes 
$T_X = \bar{X}\varphi_X$. The $Z$-sector proof follows by symmetry.
\end{proof}

\begin{lemma}[The stabilizer algebra commutes with the sector subspaces]
\label{lem:hamiltonian-invariance}
Every $F\in\mathcal A$ satisfies $[F,T]=0$ for all $T\in\mathcal V_X\cup\mathcal V_Z$.
In particular, the Hamiltonian term of a syndrome-based Lindbladian vanishes on
$\mathcal V_X$ and $\mathcal V_Z$.
\end{lemma}

\begin{proof}
Take $T=\bar XA\in\mathcal V_X$ with $A\in\mathcal A_X$. The algebra $\mathcal A$ is
generated by the stabilizers in $\mathcal S_X\cup\mathcal S_Z$. These commute with one
another, hence with $A$, and they commute with $\bar X$, since the bare logical lies
in the centralizer of the stabilizer group. Hence every $F\in\mathcal A$ commutes
with $T$. The $\mathcal V_Z$ case follows by symmetry.
\end{proof}

\begin{proof}[Proof of Lemma~\ref{lem:jump-action}]
We prove the $X$-sector statement. Write $O_X=\bar XA$ with $A\in\mathcal A_X$. We
first show that $P_iO_XP_i\in\mathcal V_X$. For $P\in\{I,X\}$, $P_i$ commutes with
$\bar X$ and with $A$, so $P_iO_XP_i=O_X$. For $P=Z$, $Z_i\bar XZ_i=\pm\bar X$ and
$Z_iAZ_i\in\mathcal A_X$, since conjugation by $Z_i$ only changes the signs of the
$X$-stabilizer generators supported on qubit $i$. For $P=Y$, writing $Y_i=iX_iZ_i$ and
using that $X_i$ commutes with $O_X$ gives $Y_iO_XY_i=Z_iO_XZ_i$. By
Lemma~\ref{lem:hamiltonian-invariance}, $f$ and $f^\dagger$ commute with both $O_X$ and
$P_iO_XP_i$. Since $P_i$ is a Hermitian unitary, $(P_if)^\dagger(P_if)=f^\dagger f$,
and Equation~\eqref{eq:dissipator-adjoint} gives
\begin{equation}
    \mathcal D[P_if]^*(O_X)
    =
    f^\dagger P_iO_XP_i f-\tfrac12\bigl\{f^\dagger f,O_X\bigr\}
    =
    f^\dagger f\,\bigl(P_iO_XP_i-O_X\bigr).
\end{equation}
Substituting the three cases for $P_iO_XP_i$ yields the claim. The $Z$-sector
statement follows by exchanging $X$ and $Z$.
\end{proof}

\begin{proof}[Proof of Lemma~\ref{lem:pauli-conjugation-translation}]
We prove the $X$-sector statement. Conjugation by $Z_i$ gives 
$Z_iOZ_i = Z_i\bar{X}AZ_i$. Using the bare logical commutation relation 
$Z_i\bar{X}Z_i = (-1)^{(\ell_X)_i}\bar{X}$ and defining 
$A^{(i)} := Z_iAZ_i \in \mathcal{A}_X$:
\begin{equation}
    Z_iOZ_i = \bar{X}\bigl((-1)^{(\ell_X)_i}A^{(i)}\bigr) \in \mathcal{V}_X.
\end{equation}
Applying $\Phi_X$:
\begin{equation}
    [\Phi_X(Z_iOZ_i)](\mathbf{e}_X) 
    = x_{\mathrm{bare}}(\mathbf{e}_X)(-1)^{(\ell_X)_i}A^{(i)}(H_X\mathbf{e}_X).
\end{equation}
Since $Z_i$ anticommutes with each $X$-stabilizer generator supported on 
qubit $i$, conjugation by $Z_i$ shifts the syndrome argument by $H_X\hat{e}_i$, where $\hat{e}_i \in \mathbb{F}_2^n$ denotes the $i$-th standard basis vector:
\begin{equation}
    A^{(i)}(H_X\mathbf{e}_X) = A(H_X\mathbf{e}_X \oplus H_X\hat{e}_i) 
    = A(H_X(\mathbf{e}_X \oplus \hat{e}_i)).
\end{equation}
Concurrently, $x_{\mathrm{bare}}(\mathbf{e}_X \oplus \hat{e}_i) 
= (-1)^{(\ell_X)_i}x_{\mathrm{bare}}(\mathbf{e}_X)$. Substituting both 
identities reconstructs the shifted functional:
\begin{equation}
    [\Phi_X(Z_iOZ_i)](\mathbf{e}_X) 
    = x_{\mathrm{bare}}(\mathbf{e}_X \oplus \hat{e}_i)\,A(H_X(\mathbf{e}_X \oplus \hat{e}_i)) 
    = [\Phi_X(O)](\mathbf{e}_X \oplus \hat{e}_i).
\end{equation}
The $Z$-sector proof follows by symmetry.
\end{proof}

\begin{lemma}[Classical representation of normalizing-Pauli jumps]
\label{lem:normalizing-jump-repr}
Let $\sigma$ be a Hermitian Pauli that normalizes the stabilizer algebra,
$\sigma\mathcal{A}_X\sigma=\mathcal{A}_X$, and satisfies $\sigma\bar X\sigma=\pm\bar X$.
Then conjugation by $\sigma$ induces a translation $\pi_\sigma:\Omega_X\to\Omega_X$
with $\Phi_X(\sigma O\sigma)=\Phi_X(O)\circ\pi_\sigma$ for every $O\in\mathcal{V}_X$.
Consequently, for any syndrome function $f\in\mathcal{A}_X$ the adjoint dissipator acts on
$\mathcal{V}_X$ as
\begin{equation}
    \mathcal{D}[\sigma f]^*(O)=\bigl(f^\dagger f\bigr)\bigl(\sigma O\sigma-O\bigr)\in\mathcal{V}_X,
\end{equation}
and its image under $\Phi_X$ is the single Markov transition
\begin{equation}
    \Phi_X\!\bigl(\mathcal{D}[\sigma f]^*(O)\bigr)(\mathbf{e})
    =\bigl(f^\dagger f\bigr)(H_X\mathbf{e})\,
    \bigl[\Phi_X(O)(\pi_\sigma\mathbf{e})-\Phi_X(O)(\mathbf{e})\bigr].
\end{equation}
\end{lemma}
Lemma~\ref{lem:pauli-conjugation-translation} is the special case $\sigma=Z_i$,
$\pi_\sigma(\mathbf{e})=\mathbf{e}\oplus \hat{e}_i$. The proof uses only that $\mathcal{A}_X$
is an abelian stabilizer algebra normalized by $\sigma$ and that $\Phi_X$ sends its
elements to their eigenvalue functions; the statement therefore holds with
$\mathcal{A}_X$ and $\Phi_X$ replaced by any abelian stabilizer algebra and its
associated representation, as we use in Section~\ref{sec:periodic}.

\begin{proof}
Since $\sigma$ is Hermitian with $\sigma^2=\mathbb{I}$ and every $f\in\mathcal{A}_X$
commutes with $\mathcal{V}_X$, the adjoint dissipator collapses to
$\mathcal{D}[\sigma f]^*(O)=f^\dagger\sigma O\sigma f-\tfrac12\{f^\dagger f,O\}
=(f^\dagger f)(\sigma O\sigma-O)$,
where $\sigma O\sigma\in\mathcal{V}_X$: writing $O=\bar X A$, normalization gives
$\sigma O\sigma=(\sigma\bar X\sigma)(\sigma A\sigma)=(\pm\bar X)(\sigma A\sigma)$ with
$\sigma A\sigma\in\mathcal{A}_X$. Conjugation by $\sigma$ flips the sign of exactly those
generators of $\mathcal{A}_X$ that it anticommutes with, and so permutes the joint
eigenspaces of $\mathcal{A}_X$ by a fixed shift of their eigenvalue labels. Under the
parametrization by which $\Phi_X$ reads off these labels, this permutation is a translation
$\pi_\sigma:\mathbf{e}\mapsto\mathbf{e}\oplus b_\sigma$ of the configuration space, so that
$(\sigma A\sigma)(H_X\mathbf{e})=A(H_X\pi_\sigma\mathbf{e})$. For a data Pauli $\sigma=Z_i$
the anticommuting generators are the $X$-stabilizers supported on qubit $i$, giving
$b_\sigma=\hat{e}_i$ and the syndrome shift $H_X\hat{e}_i$. Matching $\sigma\bar X\sigma=\pm\bar X$
against $x_{\mathrm{bare}}(\pi_\sigma\mathbf{e})=\pm x_{\mathrm{bare}}(\mathbf{e})$ then gives
$\Phi_X(\sigma O\sigma)=\Phi_X(O)\circ\pi_\sigma$, and the transition formula follows by
applying Lemma~\ref{lem:stab-alg-factor} to $f^\dagger f\in\mathcal{A}_X$.
\end{proof}

\begin{proof}[Proof of Proposition~\ref{prop:classical-repr}]
We prove the $X$-sector statement. Fix $O = \bar{X}F \in \mathcal{V}_X$. 
By Corollary~\ref{cor:sector-dynamics-general}, at any time $t$,
\begin{equation}
    \mathcal{L}^*(t)(O) = \sum_{i=1}^n
    \Gamma_i^{(X)}(t)
    \bigl(Z_iOZ_i - O\bigr),
\end{equation}
where $\Gamma_i^{(X)}(t)$ is the flip-rate operator of Equation~\eqref{eq:flip-rates},
which lies in $\mathcal{A}_X$ by the rate condition~\eqref{eq:rate-condition}.
By Lemma~\ref{lem:pauli-conjugation-translation}, $Z_iOZ_i \in \mathcal{V}_X$.
Because both $O$ and $Z_iOZ_i$ belong to $\mathcal{V}_X$, we apply 
Lemma~\ref{lem:stab-alg-factor} to each term individually:
\begin{align}
    \Phi_X\!\left(\Gamma_i^{(X)}(t)\, Z_iOZ_i\right)(\mathbf{e}_X) 
    &= \Gamma_i^{(X)}(t)(H_X\mathbf{e}_X)\, \Phi_X(Z_iOZ_i)(\mathbf{e}_X), \\
    \Phi_X\!\left(\Gamma_i^{(X)}(t)\, O\right)(\mathbf{e}_X) 
    &= \Gamma_i^{(X)}(t)(H_X\mathbf{e}_X)\, \Phi_X(O)(\mathbf{e}_X).
\end{align}
By linearity of $\Phi_X$:
\begin{equation}
    \Phi_X\!\left(\Gamma_i^{(X)}(t)(Z_iOZ_i - O)\right)(\mathbf{e}_X) 
    = \Gamma_i^{(X)}(t)(H_X\mathbf{e}_X) 
    \left[\Phi_X(Z_iOZ_i)(\mathbf{e}_X) - \Phi_X(O)(\mathbf{e}_X)\right].
\end{equation}
By Definition~\ref{def:ctmc}, $\Gamma_i^{(X)}(t)(H_X\mathbf{e}_X) = 
r_i^{(X)}(H_X\mathbf{e}_X, t)$. By Lemma~\ref{lem:pauli-conjugation-translation}, 
$\Phi_X(Z_iOZ_i)(\mathbf{e}_X) = \Phi_X(O)(\mathbf{e}_X \oplus \hat{e}_i)$. 
Substituting and summing over $i$ reproduces the classical generator:
\begin{equation}
    \Phi_X\!\left(\mathcal{L}^*(t)(O)\right)(\mathbf{e}_X) 
    = \sum_{i=1}^n r_i^{(X)}(H_X\mathbf{e}_X, t)\,
    \bigl[\Phi_X(O)(\mathbf{e}_X \oplus \hat{e}_i) - \Phi_X(O)(\mathbf{e}_X)\bigr]
    = G_X(t)\Phi_X(O).
\end{equation}

It remains to pass from generators to propagators. Differentiating
$\Phi_X\mathcal{W}^*(t,0)$ and using the linearity of $\Phi_X$
together with Eq.~\eqref{eq:q-propgator} shows it satisfies the linear
initial value problem $\frac{d}{dt}Y(t) = Y(t)\mathcal{L}^*(t)$,
$Y(0)=\Phi_X$. Differentiating $M^{(X)}(t,0)\Phi_X$ and using its own
defining equation gives
$\frac{d}{dt}\bigl(M^{(X)}(t,0)\Phi_X\bigr) = M^{(X)}(t,0)\bigl(G_X(t)\Phi_X\bigr)$,
which by the generator identity $G_X(t)\Phi_X = \Phi_X\mathcal{L}^*(t)$
just established equals $M^{(X)}(t,0)\Phi_X\,\mathcal{L}^*(t)$; so
$M^{(X)}(t,0)\Phi_X$ solves the same initial value problem. Uniqueness
of solutions implies equality for all $t \ge 0$, giving
$\Phi_X(\mathcal{W}^*(t,0)(O)) = M^{(X)}(t,0)\Phi_X(O)$ for any
$O \in \mathcal{V}_X$. The $Z$-sector statement follows by exchanging
$X$ and $Z$.
\end{proof}

\begin{proof}[Proof of Lemma~\ref{lem:vacuum-evaluation}]
We prove the \(X\)-sector statement. Write \(T_X = \bar{X}\varphi_X\) and 
\(O = \bar{X}F\). Since $\varphi_X$ and $\bar{X}$ are both Hermitian and 
commute, $T_X$ is self-adjoint. Since \(\varphi_X, F \in \mathcal{A}_X\) 
both commute with \(\bar{X}\), and \(\bar{X}^2 = \mathbbm{1}\):
\begin{equation}
    T_X^\dagger O = T_X O = \bar{X}\varphi_X \bar{X} F = \varphi_X F.
\end{equation}
Since $\rho_{\mathrm{code}}$ is supported on the code space, 
$\Pi_{\mathbf{0}}^X \rho_{\mathrm{code}} = \rho_{\mathrm{code}}$. 
For any $A \in \mathcal{A}_X$, $A\Pi_{\mathbf{0}}^X = A(\mathbf{0})\Pi_{\mathbf{0}}^X$, 
so $\mathrm{Tr}(\rho_{\mathrm{code}} A) = A(\mathbf{0})$. Applying this to 
$\varphi_X F \in \mathcal{A}_X$:
\begin{equation}
    \left\langle T_X, O \right\rangle_{\mathrm{code}} 
    = \mathrm{Tr}(\rho_{\mathrm{code}}\,\varphi_X F) 
    = \varphi_X(\mathbf{0}) F(\mathbf{0}) = F(\mathbf{0}),
\end{equation}
where we used $\varphi_X(\mathbf{0}) = 1$ 
(Definition~\ref{def:decoder-signs}). On the other hand, by 
definition of $\Phi_X$ and since $x_{\mathrm{bare}}(\mathbf{0}) = 1$ 
and $H_X\mathbf{0} = \mathbf{0}$:
\begin{equation}
    \Phi_X(O)(\mathbf{0}) = x_{\mathrm{bare}}(\mathbf{0})\,F(H_X\mathbf{0}) 
    = F(\mathbf{0}).
\end{equation}
The \(Z\)-sector proof follows by symmetry.
\end{proof}

\section{Supporting Proofs for Section~\ref{sec:sync-reduction}}
\label{app:sync}

\begin{proof}[Proof of Lemma~\ref{lem:spca-lift}]
(i) The Kraus operators \(Z^{\mathbf a(s)}\Pi^X_s\), \(s\in\Sigma_X\),
have the form of Equation~\eqref{eq:sync-kraus} with
\(\mathbf u^X=\mathbf a(s)\), \(\mathbf u^Z=\mathbf 0\), auxiliary index
\(s\), and syndrome function \(\Pi^X_s\in\mathcal A_X\subseteq\mathcal A\).
Completeness is \(\sum_{s\in\Sigma_X}\Pi^X_s=\mathbb I\). The sector
probability operators are
\begin{equation}
    \mathsf P^{(X)}_{\mathbf u}
    =
    \sum_{s\in\Sigma_X:\,\mathbf a(s)=\mathbf u}\Pi^X_s
    \in\mathcal A_X,
    \qquad
    \mathsf P^{(Z)}_{\mathbf 0}=\mathbb I,
    \qquad
    \mathsf P^{(Z)}_{\mathbf u}=0\ \ (\mathbf u\neq\mathbf 0),
\end{equation}
and \(\Pi^X_s(s')=\delta_{s,s'}\) gives the stated transition
probabilities.

(ii) We verify the condition of Claim~\ref{clm:sync-locality-H}. Fix
a set of qubits $A$ and an operator \(O\) supported on \(A\). Let \(M_1\)
be the set of \(X\)-checks supported in \(N_R(i)\) for some \(i\in A\),
let \(M_2\) be the set of \(X\)-checks whose support meets \(A\), and
let \(M:=M_1\cup M_2\). Every check in \(M_1\) is supported in
\(N_R(A)\). Every check in \(M_2\) is supported in \(N_1(A)\), since any
two qubits acted on by the same generator are adjacent. For
\(t\in\mathbb F_2^{M}\) and \(t'\in\mathbb F_2^{M^c}\), where \(M^c\) is
the set of remaining \(X\)-checks, let
\begin{equation}
    \Pi_{M,t}
    :=
    \prod_{c\in M}\frac{\mathbb I+(-1)^{t_c}S_c^X}{2},
    \qquad
    \Pi_{M^c,t'}
    :=
    \prod_{c\in M^c}\frac{\mathbb I+(-1)^{t'_c}S_c^X}{2},
\end{equation}
so that \(\Pi^X_s=\Pi_{M,s|_M}\,\Pi_{M^c,s|_{M^c}}\) for every
\(s\in\mathbb F_2^{m_X}\).

For \(i\notin A\), \(Z_i\) commutes with \(O\), so
\(Z^{\mathbf a(s)}OZ^{\mathbf a(s)}=Z^{\mathbf a_A(s)}OZ^{\mathbf a_A(s)}\),
where \(\mathbf a_A(s)\) agrees with \(\mathbf a(s)\) on \(A\) and
vanishes outside \(A\). By hypothesis there is a function
\(\mathbf g:\mathbb F_2^{M}\to\Omega_X\) with
\(\mathbf a_A(s)=\mathbf g(s|_M)\) for every \(s\in\Sigma_X\). Set
\(Y_t:=Z^{\mathbf g(t)}OZ^{\mathbf g(t)}\), which is supported on \(A\).
Since \(\Pi^X_s=0\) for \(s\notin\Sigma_X\),
\begin{equation}
    \widehat{\mathbf a}^{X*}(O)
    =
    \sum_{s\in\Sigma_X}
    \Pi^X_s\,Z^{\mathbf a(s)}OZ^{\mathbf a(s)}\,\Pi^X_s
    =
    \sum_{t\in\mathbb F_2^{M}}
    \sum_{t'\in\mathbb F_2^{M^c}}
    \Pi_{M,t}\,\Pi_{M^c,t'}\,Y_t\,\Pi_{M^c,t'}\,\Pi_{M,t}.
\end{equation}
The checks in \(M^c\) do not act on \(A\), so \(\Pi_{M^c,t'}\) commutes
with \(Y_t\) and with \(\Pi_{M,t}\). Using
\(\Pi_{M^c,t'}^2=\Pi_{M^c,t'}\) and
\(\sum_{t'}\Pi_{M^c,t'}=\mathbb I\),
\begin{equation}
    \widehat{\mathbf a}^{X*}(O)
    =
    \sum_{t\in\mathbb F_2^{M}}
    \Pi_{M,t}\,Y_t\,\Pi_{M,t},
\end{equation}
which is supported on \(A\) together with the supports of the checks in
\(M\), hence on \(N_{\max(R,1)}(A)\).
\end{proof}

\begin{proof}[Proof of Lemma~\ref{lem:sync-spca-run}]
For a Pauli string \(P=P_1\otimes\cdots\otimes P_n\), let
\(\theta(P):=\prod_{i=1}^n\theta_{P_i,i}\), with
\(\theta_{I,i}:=1-\theta_{X,i}-\theta_{Y,i}-\theta_{Z,i}\), and write
\(P=c_PX^{\mathbf v^Z(P)}Z^{\mathbf v^X(P)}\) with \(|c_P|=1\). The Kraus
operators of \(\mathcal Q_*\) are
\(
    \sqrt{\theta(P)}\,P\,
    X^{\mathbf a^Z(s')}\Pi^Z_{s'}\,
    Z^{\mathbf a^X(s)}\Pi^X_s
\),
for Pauli strings \(P\), \(s\in\Sigma_X\) and \(s'\in\Sigma_Z\). Since
\(\Pi^Z_{s'}\) commutes with \(Z\)-type Pauli strings, each equals
\begin{equation}
    X^{\mathbf v^Z(P)\oplus\mathbf a^Z(s')}\,
    Z^{\mathbf v^X(P)\oplus\mathbf a^X(s)}\,
    f,
    \qquad
    f:=c\sqrt{\theta(P)}\,\Pi^Z_{s'}\Pi^X_s\in\mathcal A,
\end{equation}
for a phase \(c\), which is the form of
Equation~\eqref{eq:sync-kraus}. Since
\(f^\dagger f=\theta(P)\,\Pi^X_s\Pi^Z_{s'}\), summing over all Kraus
operators with \(\mathbf v^X(P)\oplus\mathbf a^X(s)=\mathbf u\) and
using \(\sum_{s'\in\Sigma_Z}\Pi^Z_{s'}=\mathbb I\) gives
\begin{equation}
    \mathsf P^{(X)}_{\mathbf u}
    =
    \sum_{s\in\Sigma_X}
    \Biggl(
        \sum_{P:\ \mathbf v^X(P)=\mathbf u\oplus\mathbf a^X(s)}
        \theta(P)
    \Biggr)
    \Pi^X_s
    =
    \sum_{s\in\Sigma_X}
    \pi_X\bigl(\mathbf u\oplus\mathbf a^X(s)\bigr)\,
    \Pi^X_s
    \in\mathcal A_X,
\end{equation}
because \(v^X_i(P)=1\) exactly when \(P_i\in\{Z,Y\}\), which has
probability \(\theta_i^{(X)}\), independently over \(i\). Evaluating at
\(s\) gives Equation~\eqref{eq:sync-spca-kernel}, and the \(Z\) sector is
symmetric. By Definition~\ref{def:sync-sector-process}, the induced
\(X\)-sector process moves from \(\mathbf e\) to
\(\mathbf e\oplus\mathbf a^X(H_X\mathbf e)\oplus\mathbf v\) with
probability \(\pi_X(\mathbf v)\), independently of the past, which is
Equation~\eqref{eq:sync-spca-recursion}.
\end{proof}

\begin{proof}[Proof of Lemma~\ref{lem:pauli-diagonality}]
We work throughout in the Heisenberg picture. We prove the statements
involving \(X_Q\); the \(Z_Q\) case is identical.

We first record the covariance of the physical decoding channel. Since
\(\bar X\) commutes with every syndrome projector, and every readout
correction \(R_\sigma(s)\) is a Pauli string,
\[
    R_\sigma(s)\bar X
    =
    \pm\bar X R_\sigma(s).
\]
The two signs cancel in each term of the adjoint decoding channel, so
for every physical observable \(O\),
\begin{equation}\label{eq:physical-decoder-X-covariance}
    \mathcal E_{\mathrm{dec}}^*(\bar X O\bar X)
    =
    \bar X\,\mathcal E_{\mathrm{dec}}^*(O)\,\bar X.
\end{equation}

The physical decoder maps every syndrome sector into the code space.
Hence its adjoint depends only on the restriction of an observable to
the code space:
\begin{equation}\label{eq:decoding-code-restriction}
    \mathcal E_{\mathrm{dec}}^*(O)
    =
    \mathcal E_{\mathrm{dec}}^*
    \left(
        \Pi_{\mathrm{code}}O\Pi_{\mathrm{code}}
    \right).
\end{equation}
Using the encoding condition~\eqref{eq:encoding-condition}, for every
logical observable \(O_Q\),
\[
\begin{aligned}
&\Pi_{\mathrm{code}}
U(X_QO_QX_Q\otimes I_{\mathrm{rest}})
U^\dagger
\Pi_{\mathrm{code}}
\\
&\qquad=
\bar X\,
\Pi_{\mathrm{code}}
U(O_Q\otimes I_{\mathrm{rest}})
U^\dagger
\Pi_{\mathrm{code}}
\bar X .
\end{aligned}
\]
Equations~\eqref{eq:physical-decoder-X-covariance} and
\eqref{eq:decoding-code-restriction} therefore give
\begin{equation}\label{eq:logical-decoder-X-covariance}
    \mathcal D_{\mathrm{dec}}^*(X_QO_QX_Q)
    =
    \bar X\,
    \mathcal D_{\mathrm{dec}}^*(O_Q)\,
    \bar X.
\end{equation}

The encoding adjoint satisfies the converse relation,
\begin{equation}\label{eq:encoding-X-covariance}
    \mathcal E_{\mathrm{enc}}^*(\bar X O\bar X)
    =
    X_Q\,\mathcal E_{\mathrm{enc}}^*(O)\,X_Q.
\end{equation}
Indeed, on the initialized encoded subspace the bare logical \(\bar X\)
acts as \(X_Q\) on the designated logical qubit, by
Eq.~\eqref{eq:encoding-condition}.

Combining Eq.~\eqref{eq:logical-decoder-X-covariance}, the covariance
hypothesis~\eqref{eq:bare-logical-covariance}, and
Eq.~\eqref{eq:encoding-X-covariance}, we obtain
\begin{align}
    \mathcal E_L^*[\mathcal W_*](X_QO_QX_Q)
    &=
    \mathcal E_{\mathrm{enc}}^*
    \left(
        \mathcal W^*
        \left(
            \mathcal D_{\mathrm{dec}}^*(X_QO_QX_Q)
        \right)
    \right)
    \\
    &=
    \mathcal E_{\mathrm{enc}}^*
    \left(
        \mathcal W^*
        \left(
            \bar X\mathcal D_{\mathrm{dec}}^*(O_Q)\bar X
        \right)
    \right)
    \\
    &=
    \mathcal E_{\mathrm{enc}}^*
    \left(
        \bar X
        \mathcal W^*
        \left(
            \mathcal D_{\mathrm{dec}}^*(O_Q)
        \right)
        \bar X
    \right)
    \\
    &=
    X_Q\,
    \mathcal E_L^*[\mathcal W_*](O_Q)\,
    X_Q.
\end{align}
Similarly,
\begin{equation}
    \mathcal E_L^*[\mathcal W_*](Z_QO_QZ_Q)
    =
    Z_Q\,
    \mathcal E_L^*[\mathcal W_*](O_Q)\,
    Z_Q.
\end{equation}

Applying these identities to \(O_Q=X_Q\) shows that
\(\mathcal E_L^*[\mathcal W_*](X_Q)\) commutes with \(X_Q\) and
anticommutes with \(Z_Q\). Hence it is proportional to \(X_Q\):
\[
    \mathcal E_L^*[\mathcal W_*](X_Q)=\lambda_X X_Q.
\]
The same argument gives
\[
    \mathcal E_L^*[\mathcal W_*](Y_Q)=\lambda_Y Y_Q,
    \qquad
    \mathcal E_L^*[\mathcal W_*](Z_Q)=\lambda_Z Z_Q.
\]
Since \(\mathcal E_L[\mathcal W_*]\) is CPTP, this is a single-qubit
Pauli channel.

It remains to identify \(\lambda_X\) and \(\lambda_Z\). By Pauli
orthogonality,
\begin{align}
    \lambda_X
    &=
    \frac12
    \operatorname{Tr}_Q
    \left[
        X_Q\,
        \mathcal E_L^*[\mathcal W_*](X_Q)
    \right]
    \\
    &=
    \frac12
    \operatorname{Tr}_Q
    \left[
        X_Q\,
        \mathcal E_{\mathrm{enc}}^*
        \left(
            \mathcal W^*
            \left(
                \mathcal D_{\mathrm{dec}}^*(X_Q)
            \right)
        \right)
    \right].
\end{align}

The physical decoder maps into the code space, where
\(U(X_Q\otimes I_{\mathrm{rest}})U^\dagger\) agrees with \(\bar X\).
Hence
\[
    \mathcal D_{\mathrm{dec}}^*(X_Q)
    =
    \mathcal E_{\mathrm{dec}}^*(\bar X).
\]
Using
\[
    \mathcal E_{\mathrm{dec}}^*
    =
    \mathcal E_{\mathrm{dec}}^{(Z)*}
    \circ
    \mathcal E_{\mathrm{dec}}^{(X)*},
\]
Definition~\ref{def:decoded-observables}, and the fact that \(T_X\) commutes
with every \(\Pi^Z_{s'}\) and is invariant under conjugation by every
\(R_Z(s')\) (Remark~\ref{rem:decoder-commute}),
\[
    \mathcal D_{\mathrm{dec}}^*(X_Q)
    =
    \mathcal E_{\mathrm{dec}}^{(Z)*}(T_X)
    =
    T_X.
\]
Therefore, by adjointness of the encoding map,
\begin{equation}
    \lambda_X
    =
    \frac12
    \operatorname{Tr}
    \left[
        \mathcal E_{\mathrm{enc}}(X_Q)\,
        \mathcal W^*(T_X)
    \right].
\end{equation}
The encoding condition and Definition~\ref{def:reference-state} give
\[
    \mathcal E_{\mathrm{enc}}(X_Q)
    =
    2\rho_{\mathrm{code}}\bar X.
\]
Thus
\[
    \lambda_X
    =
    \operatorname{Tr}
    \left[
        \rho_{\mathrm{code}}\bar X\,
        \mathcal W^*(T_X)
    \right].
\]
Finally, \(\rho_{\mathrm{code}}\) is supported on the trivial
\(X\)-syndrome sector and \(R_X(\mathbf0)=\mathbb I\), so
\[
    \rho_{\mathrm{code}}T_X
    =
    \rho_{\mathrm{code}}\bar X.
\]
Hence
\[
    \lambda_X
    =
    \operatorname{Tr}
    \left[
        \rho_{\mathrm{code}}T_X\,
        \mathcal W^*(T_X)
    \right]
    =
    \bigl\langle
        T_X,\mathcal W^*(T_X)
    \bigr\rangle_{\mathrm{code}}.
\]
The \(Z_Q\) identity follows identically.
\end{proof}

\section{Supporting Proofs for the classical-register extension}
\label{app:registers}

We now prove Proposition~\ref{prop:reduction-async-formal}. We start with proving Lemma~\ref{lem:extended-jump-action} and then use it for the proof; we will make use here of the register-extensions of sector subspaces and representation maps, defined as follows:   
\begin{equation}
    \mathcal V_X^{\mathrm{ext}}
    :=
    \operatorname{span}\{\bar XF:F\in\mathcal A_X^{\mathrm{ext}}\},
    \qquad
    [\Phi_X^{\mathrm{ext}}(\bar XF)](\mathbf e,\mathbf c)
    :=
    x_{\mathrm{bare}}(\mathbf e)\,F(H_X\mathbf e,\mathbf c),
\end{equation}
extended linearly, and their $Z$-sector analogues. These are the
objects of Definitions~\ref{def:sector-subspaces}
and~\ref{def:identification-map} with the register configuration
$\mathbf c$ carried along as an additional argument.

\begin{lemma}[Action of a single register-extended jump]
\label{lem:extended-jump-action}
Let $P\in\{I,X,Y,Z\}$, let $\Theta$ be a Hermitian Pauli string on the
registers, let $f\in\mathcal A^{\mathrm{ext}}$ be a syndrome function, and let
$(\mathbf u,\mathbf v)$ be the $X$-sector flip of $P_i\Theta$. For every
$O\in\mathcal V_X^{\mathrm{ext}}$,
\begin{equation}
    \mathcal D[P_i\Theta f]^*(O)
    =
    f^\dagger f\,
    \bigl(
        Z^{\mathbf u}X^{\mathbf v}\,O\,X^{\mathbf v}Z^{\mathbf u}
        -
        O
    \bigr),
\end{equation}
and $Z^{\mathbf u}X^{\mathbf v}\,O\,X^{\mathbf v}Z^{\mathbf u}\in
\mathcal V_X^{\mathrm{ext}}$. The $Z$-sector statement is identical
with $X$ and $Z$ exchanged on the data.
\end{lemma}

\begin{proof}
Write $O=\bar XF$ with $F\in\mathcal A_X^{\mathrm{ext}}$. First,
$\mathcal A^{\mathrm{ext}}$ commutes with $O$. The generators of
$\mathcal A$ commute with $\bar X$ and with the data part of $F$, as in
Lemma~\ref{lem:hamiltonian-invariance}. Each register operator $Z_j$
commutes with $\bar X$ and with the data part of $F$, since their
supports are disjoint, and with the register generators of
$\mathcal A_X^{\mathrm{ext}}$.

Next, conjugation by $P_i\Theta$ acts on $O$ as conjugation by
$Z^{\mathbf u}X^{\mathbf v}$. The data factor $P_i$ acts as in the proof
of Lemma~\ref{lem:jump-action}. Up to a sign, which cancels in the
conjugation, $\Theta$ is a product of single-register Paulis. Those on
registers in $\mathsf{Reg}_Z$ act on tensor factors on which $O$ is the
identity, and a factor $Z_j$ with $j\in\mathsf{Reg}_X$ commutes with $O$.
For $j\in\mathsf{Reg}_X$, $X_j$ anticommutes with $Z_j$ and commutes with
every other factor of $O$, and $Y_j=iX_jZ_j$ gives
$Y_jOY_j=X_jOX_j$ because $Z_j$ commutes with $O$. Conjugation by
$Z_i$ maps $\bar X$ to $\pm\bar X$ and $\mathcal A_X^{\mathrm{ext}}$ into
itself, and conjugation by $X_j$ maps $Z_j$ to $-Z_j$, so
$Z^{\mathbf u}X^{\mathbf v}OX^{\mathbf v}Z^{\mathbf u}\in
\mathcal V_X^{\mathrm{ext}}$.

Finally, $L:=P_i\Theta$ is a Hermitian unitary, so $(Lf)^\dagger(Lf)=f^\dagger f$,
and $f,f^\dagger$ commute with both $O$ and $LOL$ by the first step.
Equation~\eqref{eq:dissipator-adjoint} gives
$\mathcal D[Lf]^*(O)=f^\dagger LOLf-\tfrac12\{f^\dagger f,O\}
=f^\dagger f\,(LOL-O)$.
\end{proof}

\begin{proof}[Proof of Proposition~\ref{prop:classical-register-reduction}]
We prove the $X$-sector statement. Let $O\in\mathcal V_X^{\mathrm{ext}}$.
The Hamiltonian term vanishes on $O$, since
$H(t)\in\mathcal A^{\mathrm{ext}}$ commutes with $O$ by the first step of
the proof of Lemma~\ref{lem:extended-jump-action}. Applying that lemma to
every jump, the jumps with trivial $X$-sector flip vanish, and the
remaining ones group by their flip into
\begin{equation}
    \mathcal L_{\mathrm{ext}}^*(t)(O)
    =
    \sum_{(\mathbf u,\mathbf v)\ne(\mathbf 0,\mathbf 0)}
    \Gamma^{(X)}_{\mathbf u,\mathbf v}(t)
    \bigl(
        Z^{\mathbf u}X^{\mathbf v}\,O\,X^{\mathbf v}Z^{\mathbf u}
        -
        O
    \bigr).
\end{equation}
By the extended rate condition, each
$\Gamma^{(X)}_{\mathbf u,\mathbf v}(t)$ lies in the abelian algebra
$\mathcal A_X^{\mathrm{ext}}$, which commutes with $\bar X$. The proof of
Lemma~\ref{lem:stab-alg-factor} therefore applies verbatim with
$\mathcal A_X^{\mathrm{ext}}$ and $\Phi_X^{\mathrm{ext}}$: each term lies
in $\mathcal V_X^{\mathrm{ext}}$, and multiplication by
$\Gamma^{(X)}_{\mathbf u,\mathbf v}(t)$ becomes multiplication by the rate
$r^{(X)}_{\mathbf u,\mathbf v}(\mathbf e,\mathbf c,t)$ of
Equation~\eqref{eq:classical-register-rate}. Lemma~\ref{lem:normalizing-jump-repr},
applied to the Pauli $Z^{\mathbf u}X^{\mathbf v}$ with
$\mathcal A_X^{\mathrm{ext}}$ and $\Phi_X^{\mathrm{ext}}$ in place of
$\mathcal A_X$ and $\Phi_X$, gives
$\Phi_X^{\mathrm{ext}}(Z^{\mathbf u}X^{\mathbf v}OX^{\mathbf v}Z^{\mathbf u})
=\Phi_X^{\mathrm{ext}}(O)\circ\pi$ with
$\pi(\mathbf e,\mathbf c)=(\mathbf e\oplus\mathbf u,\mathbf c\oplus\mathbf v)$;
the register component of the shift holds because conjugation by $X_j$
changes the eigenvalue of $Z_j$ from $(-1)^{c_j}$ to $(-1)^{c_j\oplus1}$.
Hence $\mathcal V_X^{\mathrm{ext}}$ is invariant and
$\Phi_X^{\mathrm{ext}}\circ\mathcal L_{\mathrm{ext}}^*(t)
=G_X^{\mathrm{ext}}(t)\circ\Phi_X^{\mathrm{ext}}$ on it, and the
ODE-uniqueness argument of Proposition~\ref{prop:classical-repr} gives
\begin{equation}\label{eq:extended-intertwining}
    \Phi_X^{\mathrm{ext}}\!\left(
        \mathcal W_{\mathrm{ext}}^*(t,0)(O)
    \right)
    =
    M_X^{\mathrm{ext}}(t,0)\,\Phi_X^{\mathrm{ext}}(O).
\end{equation}

For $O=\bar XF\in\mathcal V_X^{\mathrm{ext}}$, the proof of
Lemma~\ref{lem:vacuum-evaluation} extends directly: the code state
selects the trivial data syndrome and the register state selects
$\mathbf c=\mathbf 0$, so
$\langle T_X,O\rangle_{\mathrm{code}}^{\mathrm{ext}}
=[\Phi_X^{\mathrm{ext}}(O)](\mathbf 0,\mathbf 0)$.
Moreover $[\Phi_X^{\mathrm{ext}}(T_X)](\mathbf e,\mathbf c)=\Lambda_X(\mathbf e)$
does not depend on $\mathbf c$. Applying the evaluation identity to
$O=\mathcal W_{\mathrm{ext}}^*(t,0)(T_X)$ and then
Equation~\eqref{eq:extended-intertwining} gives
$[M_X^{\mathrm{ext}}(t,0)\Lambda_X](\mathbf 0,\mathbf 0)$, which by the definition
of $M_X^{\mathrm{ext}}$ is $\mathbb E[\Lambda_X(\xi_t^{(X)})]$ for the process
started at $(\mathbf 0,\mathbf 0)$, and by Lemma~\ref{lem:logical-label-phase}
equals $\mathbb E[(-1)^{b_X(\xi_t^{(X)})}]$.
\end{proof}

The register extension preserves the product structure used by
Lemma~\ref{lem:pauli-diagonality}, so both ingredients of the fidelity bound
of Subsection~\ref{subsec:memory} carry over to it. This is the form in which
Section~\ref{sec:periodic} uses them.

\begin{proof}[Proof of Corollary~\ref{cor:classical-register-memory}]
Fix \(t'\ge0\) and let \(\widetilde{\mathcal W}_*\) be the reduced data
evolution
\begin{equation}
    \widetilde{\mathcal W}_*(\rho)
    :=
    \operatorname{Tr}_{\mathrm{reg}}\!\left[
        \mathcal W_{\mathrm{ext},*}(t',0)
        \bigl(\rho\otimes|\mathbf 0\rangle\langle\mathbf 0|_{\mathrm{reg}}\bigr)
    \right],
    \qquad
    \widetilde{\mathcal W}^*(O)
    =
    \langle\mathbf 0|\,
    \mathcal W_{\mathrm{ext}}^*(t',0)\bigl(O\otimes\mathbb I_{\mathrm{reg}}\bigr)
    \,|\mathbf 0\rangle .
\end{equation}
Preparing the registers in \(|\mathbf 0\rangle\) at encoding and tracing
them out at readout gives
\(\mathcal E_L^{\mathrm{ext}}(t')=\mathcal E_L[\widetilde{\mathcal W}_*]\).

We check that \(\widetilde{\mathcal W}^*\) satisfies
Eq.~\eqref{eq:bare-logical-covariance}. Let \(\bar P\in\{\bar X,\bar Z\}\).
The bare logical \(\bar P\) commutes with \(\mathcal A^{\mathrm{ext}}\),
hence with \(H(t)\) and with every syndrome function, and it commutes or
anticommutes with every Pauli factor \(P_i\Theta\). Hence
\(\bar PJ\bar P=\pm J\) for every jump operator \(J\) of
\(\mathcal L_{\mathrm{ext}}(t)\), and since \(\mathcal D[-J]=\mathcal D[J]\),
the generator commutes with \(\bar P\)-conjugation. Therefore
\(\mathcal W_{\mathrm{ext}}^*(t',0)\) does too, and since \(\bar P\) acts on
the data only, so does \(\widetilde{\mathcal W}^*\).

Lemma~\ref{lem:pauli-diagonality} therefore shows that
\(\mathcal E_L^{\mathrm{ext}}(t')\) is a Pauli channel, with
\[
    \lambda_X
    =
    \bigl\langle T_X,\widetilde{\mathcal W}^*(T_X)\bigr\rangle_{\mathrm{code}}
    =
    \bigl\langle T_X,\mathcal W_{\mathrm{ext}}^*(t',0)(T_X)\bigr\rangle_{\mathrm{code}}^{\mathrm{ext}},
\]
where the second equality is the definition of
\(\rho_{\mathrm{code}}^{\mathrm{ext}}\); similarly for \(\lambda_Z\).
Proposition~\ref{prop:minimal-fidelity}, applied to
\(\widetilde{\mathcal W}_*\), together with
Equation~\eqref{eq:classical-register-autocorrelation}, now gives both
directions exactly as in the proof of Theorem~\ref{thm:main-async-formal}.
\end{proof}

\section{Supporting Proofs for Section~\ref{sec:periodic}}
\label{app:pca}

\begin{proof}[Proof of Proposition~\ref{prop:sampled-sector-spca}]
Fix a cycle \(r\), and set
\begin{equation}
    X_i:=\mathbf 1\{i\notin S_r^{(\sigma)}\},
    \qquad i\in[n].
\end{equation}
By Lemma~\ref{lem:local-failure-set}(ii),
\begin{equation}
    \Pr[X_i=1]\ge 1-\delta
\end{equation}
for every \(i\). Moreover, by
Lemma~\ref{lem:local-failure-set}(iv), if
\(A,B\subseteq[n]\) are at graph distance greater than \(1\) in
\(\mathcal G_\sigma\), then the random vectors
\((X_i)_{i\in A}\) and \((X_j)_{j\in B}\) are independent.
Thus \(X=(X_i)_{i\in[n]}\) is a \(1\)-dependent Bernoulli field on
\(\mathcal G_\sigma\). By Lemma~\ref{lem:locality-bound},
\(\mathcal G_\sigma\) has degree at most \(K^2\), uniformly in \(L\).

The Liggett--Schonmann--Stacey theorem
\cite{liggett_domination_1997} therefore gives constants
\(\delta_*>0\) and a function
\begin{equation}
    \effrate:[0,\delta_*]\longrightarrow[0,1],
    \qquad
    \effrate(\delta)\xrightarrow[\delta\to0]{}0,
\end{equation}
depending only on \(K^2\), with the following property. For
\(\delta\le\delta_*\), let
\begin{equation}
    Y=(Y_i)_{i\in[n]},
    \qquad
    Y_i\overset{\mathrm{iid}}{\sim}
    \operatorname{Ber}\bigl(1-\effrate(\delta)\bigr).
\end{equation}
Then \(X\) stochastically dominates \(Y\), meaning that
\begin{equation}
    \Pr[X\in A]\ge \Pr[Y\in A]
\end{equation}
for every increasing event \(A\subseteq\{0,1\}^n\).

By Strassen's theorem~\cite{strassen_existence_1965}, \(X\) and \(Y\) may be coupled so that
\begin{equation}
    Y_i\le X_i
    \qquad\text{for every }i.
\end{equation}
We realize this coupling by sampling \(Y\), conditional
on the realized value of \(X\), using fresh randomness \(U_r\)
 independent of all the
Poisson clocks and of the auxiliary randomness used
in earlier cycles.
Define
\begin{equation}
    F_r^{(\sigma)}
    :=
    \{i:Y_i=0\}.
\end{equation}
Since \(Y\) has the product law \(\nu_\delta\), the indicators
\(\mathbf 1\{i\in F_r^{(\sigma)}\}=1-Y_i\) are independent Bernoulli
variables of parameter \(\effrate(\delta)\), proving item~(i).

Moreover, if \(i\in S_r^{(\sigma)}\), then \(X_i=0\), and hence
\(Y_i=0\) because \(Y_i\le X_i\). Therefore
\begin{equation}
    S_r^{(\sigma)}\subseteq F_r^{(\sigma)}.
\end{equation}
Combining this with Lemma~\ref{lem:local-failure-set}(i) gives
\begin{equation}
    \operatorname{supp}\zeta_r^{(\sigma)}
    \subseteq
    S_r^{(\sigma)}
    \subseteq
    F_r^{(\sigma)},
\end{equation}
which proves item~(iii).

Finally, by Lemma~\ref{lem:local-failure-set}(iii),
\(S_r^{(\sigma)}\) depends only on the driving randomness of cycle \(r\),
which is independent of the randomness from all earlier cycles. Choose
the auxiliary variables \(U_0,U_1,\ldots\) independently across cycles
and independently of the sector process. Since \(F_r^{(\sigma)}\) is
constructed from \(S_r^{(\sigma)}\) and \(U_r\) alone, it is independent
of the sector trajectory up to time \(r\Delta\) and of the earlier sets
    $F_0^{(\sigma)},\ldots,F_{r-1}^{(\sigma)}.$
This proves item~(ii).

For \(\delta>\delta_*\), define \(\effrate(\delta):=1\) and take
\(F_r^{(\sigma)}=[n]\). Then items~(i)--(iii) hold trivially.
\end{proof}

\newpage
\phantomsection
\addcontentsline{toc}{section}{\refname}
\bibliographystyle{apsrev4-1}
\bibliography{bibfile_final}
\end{document}